\documentclass[11pt, oneside]{article}   	

\usepackage{jheppub}

\usepackage{amsmath}
\usepackage{amssymb}
\usepackage{amsfonts}
\usepackage{amsthm}
\usepackage{amsbsy}
\usepackage{array}

\usepackage{mathtools}
\usepackage[usenames,dvipsnames]{xcolor}

\usepackage{dsdshorthand}

\usepackage[vcentermath]{youngtab}
\newcommand{\myng}[1]{\,{\tiny\yng #1}\,}

\usepackage{tikz}
\usetikzlibrary{trees}
\usetikzlibrary{decorations.pathmorphing}
\usetikzlibrary{decorations.markings}
\usetikzlibrary{shapes.misc}

\tikzset{
  threept/.style={
    circle,
    draw,
    inner sep=2pt,
  },
  twopt/.style={
    circle,
    draw,
    fill=black,
    inner sep=1pt,
    minimum size=1pt
  },
  cross/.style={
    cross out,
    draw=black, 
    minimum size=7pt, 
    inner sep=0pt,
    outer sep=0pt
  },
  scalar/.style={
    thick,
    dashed,
    postaction={
      decorate,
      decoration={
        markings,
        mark=at position 0.5 with {\arrow{>}}
      }
    }
  },
  spinning/.style={
    thick,
    postaction={
      decorate,
      decoration={
        markings,
        mark=at position 0.5 with {\arrow{>}}
      }
    }
  },
  spinning no arrow/.style={
    thick,
  },
  finite with arrow/.style={
    decoration={
      snake,
      amplitude=1pt,
      segment length=6pt,
      post length=2pt
    },
    decorate,
    thick,->
  },
  finite/.style={
    decoration={
      snake,
      amplitude=1pt,
      segment length=6pt,
    },
    decorate,
    thick
  }
}

\def\be#1\ee{\begin{align}#1\end{align}}
\newcommand\cN{\mathcal{N}}
\newcommand\nn\nonumber
\newcommand*{\uniq}{\raisebox{-0.7ex}{\scalebox{1.8}{$\cdot$}}}

\newcommand{\diagramEnvelope}[1]{#1}

\newcommand{\thalf}{\tfrac{1}{2}}

\newcommand{\conj}[1]{{\overline #1}}

\newcommand{\normalD}[2]{\nabla_{#1}[{#2}]}
\newcommand{\hyperD}[2]{D_{#1}^{(#2)}}

\newcommand{\aParam}[2]{a^{#1}_{#2}}
\newcommand{\bParam}[2]{b^{#1}_{#2}}
\newcommand{\cParam}[2]{c^{#1}_{#2}}

\newcommand{\Dmzp}[1]{\cD_{#1,-0+}}

\newcommand{\Dppz}[1]{\cD_{#1,++0}}
\newcommand{\Dpmz}[1]{\cD_{#1,+-0}}

\newcommand{\Dmpz}[1]{\overline{\cD}^{-+0}_{#1}}
\newcommand{\Dmmz}[1]{\overline{\cD}^{--0}_{#1}}
\newcommand{\Dpzp}[1]{\overline{\cD}^{+0+}_{#1}}
\newcommand{\Dpzm}[1]{\overline{\cD}^{+0-}_{#1}}

\newcommand{\II}{\mathbb{I}}

\makeatletter
\def\@fpheader{\ }
\makeatother

\title{Weight Shifting Operators and Conformal Blocks}
\author{Denis Karateev$^{+}$, Petr Kravchuk$^{-}$, and David Simmons-Duffin$^{-,0}$}
\affiliation{
${}^+$SISSA and INFN, Via Bonomea 265, I-34136 Trieste, Italy \\
${}^-$Walter Burke Institute for Theoretical Physics, Caltech, Pasadena, California 91125, USA \\
${}^{\,0}$School of Natural Sciences, Institute for Advanced Study, Princeton, New Jersey 08540, USA
}
\date{}
\abstract{We introduce a large class of conformally-covariant differential operators and a crossing equation that they obey. Together, these tools dramatically simplify calculations involving operators with spin in conformal field theories. As an application, we derive a formula for a general conformal block (with arbitrary internal and external representations) in terms of derivatives of blocks for external scalars. In particular, our formula gives new expressions for ``seed conformal blocks" in 3d and 4d CFTs. We also find simple derivations of identities between external-scalar blocks with different dimensions and internal spins. We comment on additional applications, including deriving recursion relations for general conformal blocks, reducing inversion formulae for spinning operators to inversion formulae for scalars, and deriving identities between general $6j$ symbols (Racah-Wigner coefficients/``crossing kernels") of the conformal group.}

\preprint{CALT-TH 2017-031}

\begin{document}

\maketitle

\section{Introduction}

Concrete results in conformal representation theory have played a crucial role in the recent resurgence of the conformal bootstrap \cite{Rattazzi:2008pe,Rychkov:2009ij,Caracciolo:2009bx,Rattazzi:2010gj,Poland:2010wg,Rattazzi:2010yc,Vichi:2011ux,Poland:2011ey,Rychkov:2011et,ElShowk:2012ht,Liendo:2012hy,ElShowk:2012hu,Gliozzi:2013ysa,Kos:2013tga,Alday:2013opa,Gaiotto:2013nva,Berkooz:2014yda,El-Showk:2014dwa,Nakayama:2014lva,Nakayama:2014yia,Chester:2014fya,Kos:2014bka,Caracciolo:2014cxa,Nakayama:2014sba,Paulos:2014vya,Bae:2014hia,Beem:2014zpa,Chester:2014gqa,Alday:2015eya,Simmons-Duffin:2015qma,Bobev:2015jxa,Kos:2015mba,Chester:2015qca,Beem:2015aoa,Iliesiu:2015qra,Rejon-Barrera:2015bpa,Poland:2015mta,Lemos:2015awa,Kim:2015oca,Lin:2015wcg,Chester:2015lej,Alday:2015ota,Alday:2015ewa,Chester:2016wrc,Behan:2016dtz,Dey:2016zbg,Nakayama:2016knq,El-Showk:2016mxr,Li:2016wdp,Pang:2016xno,Lin:2016gcl,Alday:2016mxe,Alday:2016njk,Alday:2016jfr,Lemos:2016xke,Beem:2016wfs,Simmons-Duffin:2016wlq,Li:2017ddj,Collier:2017shs,Cornagliotto:2017dup,Gliozzi:2017hni,Gopakumar:2016cpb,Rychkov:2017tpc,Nakayama:2017vdd,Chang:2017xmr,Dymarsky:2017xzb}. Compact expressions for conformal blocks with external scalars \cite{DO1,DO2} were crucial for the development of modern numerical bootstrap techniques \cite{Rattazzi:2008pe}. Subsequently, techniques for computing blocks of operators with spin \cite{Costa:2011mg,Costa:2011dw,SimmonsDuffin:2012uy,Iliesiu:2015qra,Iliesiu:2015akf,Costa:2016xah,Costa:2016hju,Cuomo:2017wme,Echeverri:2015rwa,Echeverri:2016dun} have led to universal numerical bounds on wide classes of CFTs \cite{Iliesiu:2015qra,Dymarsky:2017xzb,TTTT}, in addition to analytical results like proofs of the conformal collider bounds \cite{Hartman:2015lfa,Hartman:2016dxc,Li:2015itl,Hofman:2016awc} and the average null energy condition \cite{Hartman:2016lgu}, and new results on the Regge limit in CFTs \cite{Afkhami-Jeddi:2016ntf,Li:2017lmh,Kulaxizi:2017ixa}. In parallel developments, harmonic analysis on the conformal group \cite{Dobrev:1977qv} has played an important role in several recent works \cite{Cornalba:2007fs,Costa:2012cb,Caron-Huot:2017vep,Gadde:2017sjg,Hogervorst:2017kbj,Hogervorst:2017sfd}, including the large-$N$ solution of the SYK model \cite{Kitaev,Maldacena:2016hyu,Polchinski:2016xgd,Murugan:2017eto}. Relationships between Witten diagrams and conformal blocks have also received recent attention \cite{Hijano:2015zsa,Nishida:2016vds,Dyer:2017zef,Chen:2017yia,Sleight:2017fpc,Castro:2017hpx}.

More sophisticated analyses will require new results for operators with spin. Several efficient techniques for dealing with spinning operators have been developed over the last decade, including index-free/embedding-space methods \cite{Giombi:2011rz,Costa:2011mg, Costa:2011dw,SimmonsDuffin:2012uy,Iliesiu:2015qra,Costa:2016hju}, the shadow formalism in the embedding space \cite{SimmonsDuffin:2012uy}, ``differential bases" for three-point functions \cite{Costa:2011dw,Iliesiu:2015qra}, and recursion relations \cite{Penedones:2015aga,Iliesiu:2015akf,Dymarsky:2017xzb}. While these methods are superior to naive approaches, they still aren't enough to solve some difficult problems. For example, the shadow formalism lets one write integral expressions for general blocks, but the integrals are difficult to evaluate in practice in all but the simplest cases. The differential basis approach lets one compute spinning blocks in terms of simpler ``seed blocks," but doesn't explain how to compute the seed blocks.\footnote{A recursion relation for seed blocks in 3d was guessed in \cite{Iliesiu:2015akf} by solving the Casimir equation order-by-order in an OPE expansion. Expressions for seed blocks in 4d were derived in \cite{Echeverri:2016dun} by solving the Casimir equation using a suitable ansatz.}

In this work, we introduce new tools that dramatically simplify computations in conformal representation theory, particularly involving operators with spin. The first key idea is to consider a (fictitious) operator $w(x)$ that transforms in a {\it finite-dimensional}\ representation $W$ of the conformal group. By studying the OPE of this highly degenerate operator with a non-degenerate operator $\cO(x)$, we find (in section~\ref{sec:differentialoperators}) a large class of conformally-covariant differential operators $\cD^{v}_{A}$ that can be used for computations.  Here, $A=1,\dots,\dim W$ is an index for $W$, and $v$ is a weight vector of $W$ (i.e.\ a common eigenvector of the Cartan subalgebra).\footnote{Some examples of such operators appear in the conformal tractor calculus, which originally deals with the case of tensor $W$~\cite{Thomas1926,bailey1994}. The theory of local twistors~\cite{Dighton1974,PENROSE1973241,Friedrich1977} deals with the case of spinor $W$. The primary interest of these theories is in curved conformal manifolds. Part of our results can be viewed as a classification of differential operators involving tractor or local twistor bundles in the conformally flat setting. It is an interesting question whether our results generalize to the curved setting.}

The action of $\cD^{v}_{A}$ on $\cO(x)$ shifts the weights of $\cO$ by the weights of $v$, in addition to introducing a free $A$ index. For this reason, we call $\cD^{v}_A$ a {\it weight-shifting operator}. For example, weight-shifting operators can increase or decrease the spin of $\cO$.\footnote{When $\cD^v_A$ lowers the spin of $\cO$, its missing spin degrees of freedom are (roughly speaking) transferred to the index $A$ for $W$.} Weight-shifting operators can be written explicitly using the embedding space formalism \cite{Dirac:1936fq,Mack:1969rr,Boulware:1970ty,Ferrara:1973eg,Ferrara:1973yt,Cornalba:2009ax,Weinberg:2010fx,Costa:2011mg,SimmonsDuffin:2012uy,Iliesiu:2015qra,Costa:2016hju}, e.g.\ (\ref{eq:vectoroperators}) in general spacetime dimensions, (\ref{eq:3Doperators}) in 3d, and (\ref{eq:4Doperators}) in 4d. However, our construction applies independently of the embedding space formalism, and in fact works for generalized Verma modules of any Lie (super-)algebra.\footnote{Our construction is based on the ``translation functor" of Zuckerman and Jantzen \cite{10.2307/1971097,jantzen}.}

A second key observation is that weight-shifting operators obey a type of crossing equation,
\be
\cD^v_{A,x_1} \<\cO_1'(x_1)\cO_2(x_2)\cO_3(x_3)\>^{(a)} &= \sum_{\cO_2',v',b} \left\{\cdots \right\}\cD^{v'}_{A,x_2} \<\cO_1(x_1) \cO_2'(x_2) \cO_3(x_3)\>^{(b)},
\label{eq:crossingfordifferentialops}
\ee
which we derive in section~\ref{sec:crossingfordifferentialoperators}.
Here, $a$ and $b$ label conformally-invariant three-point structures that can appear in a correlator of the given operators. The coefficients $\{\cdots\}$ are examples of $6j$ symbols (or Racah-Wigner coefficients) for the conformal group, which in this case are computable with simple algebra. Equation~(\ref{eq:crossingfordifferentialops}) lets us move a covariant differential operator acting on $x_1$ to an operator acting on $x_2$. As we will see, this provides enough flexibility to perform a variety of computations involving weight-shifting operators. We also introduce a diagrammatic language that makes these computations easy to understand.

As an application, in section~\ref{sec:conformalblocks} we focus on computing conformal blocks and understanding some of their properties. In section~\ref{sec:generalblockexpression}, we derive an expression for a general conformal block involving operators (both external and internal) in arbitrary representations of $\SO(d)$ in terms of derivatives of blocks with external scalars.\footnote{The rough idea is that weight-shifting operators allow us to exchange a tensor product $W\otimes V_{\De,\ell}$, where $W$ is finite-dimensional, and $V_{\De,\ell}$ is the generalized Verma module of a symmetric traceless tensor (STT) operator. This tensor product then contains many new types of generalized Verma modules that can include operators in non-STT representations of $\SO(d)$.} This generalizes the beautiful result of \cite{Costa:2011dw} for conformal blocks of symmetric traceless tensors (STTs). Weight-shifting operators also explain where the differential operators of \cite{Costa:2011dw} come from (as we discuss in section~\ref{sec:differentialbases}). Our formula can be simplified in special cases. For example, in section~\ref{sec:seedalgo} we give new expressions for so-called ``seed blocks" in 3d and 4d CFTs in terms of derivatives of scalar blocks.

Our techniques also give a new way to understand many identities and recursion relations satisfied by conformal blocks. In section~\ref{sec:DolanOsbornShifts}, we rederive and explain diagrammatically several identities relating scalar conformal blocks with different dimensions and spins.\footnote{These identities can also be understood using techniques from integrability \cite{Isachenkov:2016gim,Chen:2016bxc,Schomerus:2016epl}.} In section~\ref{sec:recursionrelationsection}, we discuss how to use derivative-based expressions for blocks to find recursion relations of the type introduced by Zamolodchikov \cite{Zamolodchikov:1985ie,Zamolodchikov:1987} and used in numerical bootstrap computations \cite{Kos:2013tga,Kos:2014bka,Dymarsky:2017xzb,Lin:2015wcg,Cho:2017oxl,Lin:2016gcl}.

In section~\ref{sec:furtherapps}, we comment on some additional applications beyond computing conformal blocks. Weight-shifting operators are helpful for studying inner products between conformal blocks that appear in inversion formulae \cite{Caron-Huot:2017vep,Gadde:2017sjg,Hogervorst:2017kbj,Hogervorst:2017sfd}. By integrating weight-shifting operators by parts, one can reduce inversion formulae for spinning operators to inversion formulae for scalars. In particular, one can express $6j$ symbols for arbitrary generalized Verma modules of the conformal group in terms of $6j$ symbols for four scalar (and two STT) representations. We pursue this idea in more detail in \cite{ShadowFuture}.

A related idea is ``spinning-down" a crossing equation: by applying spin-lowering operators to both sides of a crossing equation, we can express it in terms of a crossing equations for scalar operators. Spinning-down may be useful in the numerical bootstrap --- it could perhaps obviate the need to explicitly compute spinning blocks.

Finally, in section~\ref{sec:discussion}, we discuss further applications and future directions. We give several details and examples in the appendices.

\section{Weight-shifting operators}
\label{sec:differentialoperators}

\subsection{Finite-dimensional conformal representations}
\label{sec:finitedim}

Let $W$ be a finite-dimensional irreducible representation of $\SO(d+1,1)$.  We
can think of $W$ in two different ways. Firstly, $W$ is a vector space
with basis $e^A$ ($A=1,\dots,\dim W$), in which the action of the
conformal group is given by
\be\label{eq:eAtransformation}
g \. e^A &= D_B{}^A(g) e^B,
\ee
where $D_B{}^A(g)$ are representation matrices.

Secondly, $W$ is the conformal representation of a (very)
degenerate primary operator $w^a(x)$. Under the subgroup $\SO(1,1) \x \SO(d)
\subset \SO(d+1,1)$ generated by dilatations and and rotations, $W$
decomposes into a direct sum\footnote{$j$ is equal to the sum of all Dynkin labels of $W$, with spinor labels counted with multiplicity $\frac 1 2$, which is the same as the length of the first row of the $\SO(d+1,1)$ Young diagram for $W$.}
\be
\label{eq:Wdecomposition}
W &\to \bigoplus_{i=-j}^{j} (W_i)_{i},\qquad j\in \tfrac{1}{2}\mathbb{N}.
\ee
Here, $(\rho)_\De$ denotes a representation of $\SO(1,1)\x \SO(d)$ with dimension $\De$ and $\SO(d)$ representation $\rho$. The dimensions
in the decomposition (\ref{eq:Wdecomposition}) are integer-spaced and must be invariant under the Weyl
reflection $\De\to -\De$, which implies that they are integers or
half-integers.\footnote{In general this Weyl reflection also acts non-trivially on the $\SO(d)$ representations.}

The lowest-dimension summand in~\eqref{eq:Wdecomposition} is spanned by the multiplet $w^a(0)$ which has scaling dimension $-j$ and carries an index $a$ for the $\SO(d)$ representation $W_{-j}$ (which is always irreducible). Because it has the lowest dimension in $W$, it is annihilated by $K_\mu$ and thus is a primary. 
The
position-dependent operator $w^a(x) = e^{x\.P} w^a(0)$ is a polynomial
in $x$ of degree $2j$ because the representation $W$ contains only $2j+1$
levels of descendants. In other words, almost all descendants of $w^a(x)$ are null and this is reflected in the fact that $w^a(x)$ satisfies a particular generalization of the conformal Killing equation that admits only polynomial solutions.

We can relate these two pictures by expanding $w^a(x)$ in our basis
\be
w^a(x) &= w^a_A(x) e^A.
\ee
The coefficients in this expansion $w^a_A(x)$ are conformal Killing
(spin-)tensors. As an example, consider the adjoint representation
$\myng{(1,1)}$ of the
conformal group. Under $\SO(1,1)\x \SO(d)$, it decomposes
as (here and throughout, ``$\bullet$" denotes the trivial representation)
\be
\myng{(1,1)} &\to (\myng{(1)})_{-1} 
\oplus (\bullet \oplus
\myng{(1,1)})_0 \oplus (\myng{(1)})_{1}
\ee
The operator $w^\mu(x)$ is thus a vector with dimension $-1$. A basis for $W=\myng{(1,1)}$ is
given by $e^A \in \{K^\mu, D, M^{\mu\nu}, P^\mu\}$, and the coefficients $w^\mu_A(x)$ in this basis are the usual conformal Killing vectors on $\R^d$,
\be
\label{eq:conformalkillingvectors}
w^\mu(x) &= K^\mu -2x^\mu D + (x_\rho \de^{\mu}_{\nu} - x_\nu
\de^{\mu}_{\rho}) M^{\nu\rho} + (2 x^\mu x_\nu-x^2 \de^{\mu}_\nu) P^\nu.
\ee
In this case the differential equation satisfied by $w^\mu(x)$ is the usual conformal Killing equation,
\be
\partial^\mu w^\nu(x)+\partial^\nu w^\mu(x)-\textrm{trace}=0.
\ee

\subsection{Tensor products with finite-dimensional representations}
\label{sec:tensorproducts}

Consider a primary operator $\cO$ with $\SO(1,1)\x\SO(d)$ representation $(\rho)_{\De}$ for generic $\De$. The conformal multiplet of $\cO$ is a generalized Verma module which we denote $V_{\De,\rho}$.\footnote{Recall that a generalized Verma module (also called a parabolic Verma module) is roughly-speaking obtained by  starting with a finite-dimensional representation of a subgroup (in this case $\SO(1,1)\x\SO(d)$) and acting with arbitrary products of lowering operators (in this case the momentum generators $P_\mu$). See, e.g.\ \cite{Yamazaki:2016vqi}. This is the usual construction of long multiplets in conformal field theory.} Under a conformal transformation $x'=g(x)$, $\cO$ transforms in the usual way\footnote{When we think of $\cO^a(x)$ as an operator on a Hilbert space, then $g\.\cO^a(x)$ means $U_g \cO^a(x) U_g^{-1}$, where $U_g$ is the unitary operator implementing $g$. Equation~\eqref{eq:primarytransformation} should thus be understood as defining the action of $g$ on the \emph{value} $\cO(x)$ rather than the function $\cO$.}
\be\label{eq:primarytransformation}
g\cdot \cO^a(x) &= \Omega(x')^\De \rho^a{}_b(R(x')^{-1})\cO^b(x'),\nn\\
\Omega(x')R^\mu{}_\nu(x') &= \pdr{x'^\mu}{x^\nu},
\ee
where $R^\mu{}_\nu\in \SO(d)$ and $\rho^a{}_b(R^{-1})$ is the action of $R^{-1}$ in the representation $\rho$.

We would like to understand the decomposition of the tensor product
\be
	W \otimes V_{\De,\rho},
\ee
when $W$ is finite-dimensional. This is equivalent to finding primary operators built out of $w^a(x)$ and $\cO^b(x)$. Formally, we must take an OPE between $w^a(x)$ and $\cO^b(x)$, treating them as operators in decoupled theories.\footnote{We are not assuming that $w^a(x)$ is an operator in a physical theory --- it is simply a mathematical object that serves as a useful tool for understanding consequences of conformal symmetry.} The simplest primary in the OPE is
\be\label{eq:simplestensortprimary}
	w^a(0)\otimes \cO^b(0),
\ee
which is primary because it vanishes under the action of the special conformal generator $1\otimes K_\mu+K_\mu\otimes 1$. This particular state is not generally in an irreducible representation of $\SO(d)$. Decomposing it further, we obtain primary states in irreducible representations $\lambda\in W_{-j}\otimes\rho$ of $\SO(d)$ and with scaling dimensions $\De-j$.

To find the other primaries in the OPE, we can use the following trick. Define $M=W\otimes V_{\De,\rho}$ and consider the factor space $M'=M/(\oplus_\mu P_\mu M)$, i.e.\ treat all total derivatives in $M'$ as zero. Then any two states in $M$ differing by a descendant will be equal in $M'$. As we show in appendix~\ref{app:verma}, for generic $\De$ the tensor product $M$ decomposes into a direct sum of simple generalized Verma modules, and in this case it is easy to see that the non-zero states in $M'$ are in one-to-one correspondence with the primary states in $M$. 

We can easily find a basis for $M'$: given any expression of the form $\partial\cdots\partial w^a(0)\otimes \partial\cdots \partial \cO(0)$, we can ``integrate by parts'' and move all the derivatives to act on $w$. Thus a basis for $M'$ is given by the non-trivial states of the form\footnote{If $\cO^b(0)$ had null descendants (for example, if it itself were the primary of a finite-dimensional representation), it would be possible that some of these states are total derivatives and thus vanish in $M'$. Since we assume that $\De$ is generic, this does not happen.}
\be
	\partial_{\mu_1}\cdots \partial_{\mu_m} w^a(0)\otimes \cO^b(0).
\ee
Note that because $w$ has a finite number of non-zero descendants, $M'$ is finite-dimensional.

To find the primaries in $M$ corresponding to this basis, we need to add total derivatives with the same scaling dimension to the above basis elements. This leads to the following ansatz with some undetermined coefficients $c_k$,
\be
\label{eq:primaryansatz}
	c_1\partial_{\mu_1}\cdots\partial_{\mu_m} w^a(0) \otimes \cO^b(0)+c_2\partial_{\mu_1}\cdots\partial_{\mu_{m-1}} w^a(0) \otimes \partial_{\mu_m} \cO^b(0)+\ldots.
\ee
After projecting onto an irreducible $\SO(d)$ representation $\l\in W_{-j+m}\otimes \rho$, we obtain an ansatz for a primary in representation $(\l)_{\De-j+m}$. We can fix the coefficients $c_k$ up to an overall normalization by requiring that the state (\ref{eq:primaryansatz}) is annihilated by $1\otimes K_\mu+K_\mu\otimes 1$. In this way, we find a primary operator of scaling dimension $\De+i$ for each of the irreducible components in $W_i\otimes \rho$ and every $i=-j,\ldots,j$.

It is not hard to confirm that these primaries account for all the states in $W\otimes V_{\De,\rho}$ by checking that the $\SO(1,1)\times \SO(d)$ characters agree. We thus conclude
\be\label{eq:tensorproductdecomposition}
	W\otimes V_{\De,\rho} = \bigoplus_{i=-j}^{j}\bigoplus_{\l \in W_i\otimes \rho} V_{\De+i,\l},\qquad(\textrm{generic $\De$}).
\ee

As a simple example, consider the case where $W=\myng{(1)}$ is the vector representation of $\SO(d+1,1)$ and $\rho$ is the trivial representation of $\SO(d)$. We have the decomposition
\be
\label{eq:vectordecomposition}
	\myng{(1)}\to (\bullet)_{-1}\oplus (\myng{(1)})_0\oplus (\bullet)_{+1},
\ee
so the primary state of $W$ is the scalar $w(0)$ of scaling dimension $-1$.
We thus find
\be
	\myng{(1)}\otimes V_{\De,\bullet}=V_{\De-1,\bullet}\oplus V_{\De,\myng{(1)}}\oplus V_{\De+1,\bullet}.
\ee
According to the above discussion, we have the following ansatz for the primaries in this decomposition
\begin{align}\label{eq:exansatz}
	V_{\De-1,\bullet}:&\quad \phi_{-}(0)=w(0)\otimes \cO(0),\nn\\
	V_{\De,\myng{(1)}}:&\quad V_\mu(0)=t_1\partial_\mu w(0)\otimes \cO(0)+t_2 w(0)\otimes \partial_\mu\cO(0),\nn\\
	V_{\De+1,\bullet}:&\quad \phi_{+}(0)=b_1\partial^2 w(0)\otimes \cO(0) + b_2 \partial_\mu w(0)\otimes \partial^\mu\cO(0) + b_3 w(0)\otimes \partial^2\cO(0).
\end{align}
Recalling that $\partial_\mu$ is the same as the action of $P_\mu$ and using the conformal algebra in appendix~\ref{app:conformalalgebra}, we find
\begin{align}
(1\otimes K_\mu+K_\mu\otimes 1)\cdot \phi_-(0)&=0,\nn\\
	(1\otimes K_\mu+K_\mu\otimes 1)\cdot V_\nu(0)&=2\delta_{\mu\nu}(\De t_2-t_1)w(0)\otimes \cO(0),\nn\\
	(1\otimes K_\mu+K_\mu\otimes 1)\cdot \phi_+(0)&=2(\De b_2-d b_1)\partial_\mu w(0)\otimes \cO(0)\nonumber\\
	&\quad
	+2\left( b_3 \p{2\De-d+2}-b_2\right)w(0)\otimes \partial_\mu\cO(0).
\end{align}
It follows that these states are primary if
\be\label{eq:exsolution}
	t_1&=\Delta t_2,\nn\\
	b_1&=\frac{\De b_3}d (2\De-d+2),\quad b_2=b_3(2\De-d+2).
\ee
We must assume that $\De$ is generic because e.g.\ for $\De=1$, $V_\mu$ becomes a primary descendant of $\phi_-$, $V_\mu = \partial_\mu\phi_-$. In this special case, there are not sufficiently many primaries to account for all states of dimension $\Delta$. In particular there is no combination of descendants which gives $\partial_\mu w(0)\otimes \cO(0)$, and consequently $\myng{(1)}\otimes V_{1,\bullet}$ does not decompose into generalized Verma modules of primary operators. These subtleties will not be important in this work, and we will always assume $\Delta$ to be generic.

\subsection{Covariant differential operators from tensor products}
\label{sec:diffopsfromtensor}

Consider now the primary state (\ref{eq:primaryansatz}), and let us write it in the form
\be\label{eq:tensorprimary}
\cO'^c(x) &= e^A\otimes (\cD_A)^c{}_b \cO^b(x),
\ee
where the differential operators $\cD_A$ are defined by\footnote{Note that $\cD_A$ depends explicitly on $x$. This is because $P_\mu$ acts non-trivially on $W$ and thus these operators are translation-covariant rather than translation-invariant.}
\be\label{eq:solvedansatz}
(\cD_A)^c{}_b \cO^b(x) &\equiv \pi^c_{ab\mu_1\cdots\mu_m}\left(
c_1\partial^{\mu_1}\cdots\partial^{\mu_m} w_A^a(x) \cO^b(x)\right.\nn\\
&\left.\qquad\qquad\qquad+c_2\partial^{\mu_1}\cdots\partial^{\mu_{m-1}} w_A^a(x) \partial^{\mu_m} \cO^b(x)+\ldots\right).
\ee
Again, the $c_i$ are chosen so that $\cO'^c(0)$ is a primary transforming in the representation $(\l)_{\De'}$. Here, $\pi^c_{ab\mu_1\cdots\mu_m}$ is a projector onto the $\SO(d)$ representation $\l \in W_{-j +m}\otimes \rho$.

By construction, $\cO'$ transforms under a conformal transformation as
\be
g\.\cO'^c(x) &= \Omega(x')^{\De'} \l^c{}_d(R^{-1}(x')) \cO'^d(x').
\ee
On the other hand, we also have
\be
g\.\cO'^c(x) &= g\.e^A \otimes g\.(\cD_A\cO)^c(x)\nn\\
&= D_B{}^A(g) e^B \otimes g\.(\cD_A\cO)^c(x).
\ee
It follows that
\be
g\.(\cD_A\cO)^c(x) &=\Omega(x')^{\De'} \l^c{}_d(R^{-1}(x'))  D_A{}^B(g^{-1})(\cD_B\cO)^d(x').
\ee
In other words, $\cD_A$ takes a primary operator that transforms in $(\rho)_\De$ to a primary operator that transforms in $(\l)_{\De'}$, up to the additional action of the finite-dimensional matrix $D_B{}^A(g^{-1})$.  We summarize this situation by writing
\be
\label{eq:summary}
\cD_A : [\De,\rho] &\to [\De',\l].
\ee
Here, for all practical purposes $[\De,\rho]$ is just a convenient notation. We give it a precise meaning in appendix~\ref{app:verma}.

Notice that $\cD_A\cO$ has a lowered index for $W$, so it transforms in the same way as the basis elements of the dual representation $W^*$. For this reason, we will say that $\cD_A$ is \emph{associated} with $W^*$. Similarly, exchanging $W$ and $W^*$, $\cD^A$ is associated with $W$. This convention will be useful when we discuss the action of differential operators on tensor structures in section~\ref{sec:tensorstructures}.

This general construction shows that there exists a huge variety of conformally covariant differential operators, corresponding to tensor products with different finite-dimensional representations. In fact, as explained in appendix~\ref{app:verma}, all conformally-covariant differential operators acting on generic Verma modules arise in this way. For reference, let us summarize this result in the following 
\begin{theorem}\label{thm:operatorclassification}
The conformally-covariant operators $\cD^A:[\De,\rho]\to [\De-i,\lambda]$  associated with $W$ are (for generic $\De$) in one-to-one correspondence with the irreducible components in the tensor product decomposition
\be
\label{eq:differentialops}
	W^*\otimes V_{\De,\rho}=\bigoplus_{i=-j}^j \bigoplus_{\lambda\in (W_i)^*\otimes\rho} V_{\De-i,\lambda}.
\ee
\end{theorem}

When the Dynkin indices of $\rho$ are sufficiently large, Brauer's formula (also known as Klimyk's rule)~\cite{Brauer,MR0237712} implies that the tensor products simplify, giving
\be
	W^*\otimes V_{\De,\rho}= \bigoplus_{(\de,\pi)\in \Pi(W^*)} V_{\De+\de,\rho+\pi}.
	\label{eq:usingBrauer}
\ee
Here, $\Pi(W^*)$ denotes the weights of $W^*$ (with multiplicity). A consequence of (\ref{eq:usingBrauer}) is that for generic $\De,\rho$, the number of differential operators acting on $[\De,\rho]$ and transforming in $W$ is equal to $\dim(W^*)$. Further, each operator is labeled by a weight vector of $W^*$ (i.e.\ an element of $W^*$ which is an eigenvector of the Cartan subalgebra) and shifts $(\De,\rho)$ by that weight. For this reason, we call the $\cD^A$ {\it weight-shifting operators}.

One of the most important weight-shifting operators comes from the adjoint representation of the conformal group, $W=\myng{(1,1)}$. The tensor product $\myng{(1,1)}\otimes V_{\De,\rho}$ always contains $V_{\De,\rho}$ itself as a factor. The corresponding $\cD^A:[\De,\rho]\to [\De,\rho]$ are the usual differential operators generating the action of the conformal algebra (see e.g.\ \cite{Simmons-Duffin:2016gjk}),
\be
\cD^A &= w^A\.\ptl + \frac{\De}{d}(\ptl\.w^A) - \frac 1 2 (\ptl^\mu w^{A\nu})\mathcal{S}_{\mu\nu},
\ee
where $w^{A\mu}$ are conformal Killing vectors (\ref{eq:conformalkillingvectors}), and $\mathcal{S}_{\mu\nu}$ are the generators of $\SO(d)$ rotations in the representation $\rho$.

\subsection{Algebra of weight-shifting operators}
\label{sec:algebraofweightshifting}

What is the algebra of weight-shifting operators?\footnote{The results of this section are not used in the rest of this work. The reader should feel free to skip this section on first reading} Before answering this question, let us rephrase our construction in a slightly different language. Recall from~\eqref{eq:tensorprimary} and~\eqref{eq:solvedansatz} that we identify primaries in $W\otimes V_{\De,\rho}$ of the form  
\be
\label{eq:exampletensorstate}
\cO'^c(0) &= e^A \otimes (\cD_A)^c{}_b \cO^b(0).
\ee
Note that $\cO'^c(0)\in W\otimes V_{\De,\rho}$ but it transforms in the same way as the primary of $V_{\De',\lambda}$. This means that \eqref{eq:exampletensorstate} gives a homomorphism 
\be\label{eq:Vermahomomorphism}
	\Phi : V_{\De',\lambda} \to W\otimes V_{\De,\rho}, 
\ee
defined by mapping the primary of $V_{\De',\lambda}$ to the right hand side of \eqref{eq:exampletensorstate}. The action of $\Phi$ on descendants follows by acting with $P_\mu\otimes 1+1\otimes P_\mu$ on \eqref{eq:exampletensorstate}.

Composition of differential operators is equivalent to composition of the corresponding homomorphisms in the opposite order. Specifically, suppose
\be
\Phi_1 : V_{\De',\rho'} &\to W_1 \otimes V_{\De,\rho},\nn\\
\Phi_2 : V_{\De'',\rho''} &\to W_2 \otimes V_{\De',\rho'}.
\ee
Then
\be
(1\otimes \Phi_1)\circ\Phi_2 : V_{\De'',\rho''} &\to W_2 \otimes W_1 \otimes V_{\De,\rho}\nn\\
(1\otimes \Phi_1)(\Phi_2(\cO''(x))) &= e_2^B \otimes e_1^A \otimes \cD_{2B}\cD_{1A} \cO(x).
\label{eq:compositionofhomomorphisms}
\ee
Thus, to find the algebra of weight-shifting operators, we must express the right-hand side of (\ref{eq:compositionofhomomorphisms}) in terms of homomorphisms associated to the irreducible factors of $W_2\otimes W_1$. 

As we will see in the next section, the embedding formalism lets us define weight-shifting operators that make sense even when $\rho$ is a generic (i.e.\ not necessarily dominant) weight. For example, the spin $\ell$ of a symmetric traceless tensor operator can be written as $Z\.\pdr{}{Z}$, where $Z$ is a polarization vector. The operator $Z\.\pdr{}{Z}$ is then well-defined when acting on functions of non-integer homogeneity in $Z$.

The correct way to understand differential operators with generic weights is to consider homomorphisms between Verma modules as opposed to {\it generalized}\ Verma modules. Consider the triangular decomposition
\be
\mathfrak{g} &= \mathfrak{g}_- \oplus \mathfrak{h} \oplus \mathfrak{g}_+,
\ee
where $\mathfrak{h}$ is the Cartan subalgebra, and $\mathfrak{g}_\pm$ are generated by positive/negative roots of $\mathfrak{g}$.  Let $M_\l$ be the Verma module of $\mathfrak{g}$ with highest-weight $\l$, and denote the corresponding highest-weight vector by $x_\l$.\footnote{When $\l=(\De,\rho)$ with $\rho$ a dominant weight of $\mathfrak{so}(d)$, then $M_\l$ is reducible and contains the generalized Verma module $V_{\De,\rho}$ as a subfactor.}

Let $W$ be a finite-dimensional representation of $\mathfrak{g}$. For each weight-vector\footnote{Not to be confused with the conformal Killing tensors $w^A$ from the previous section.} $w\in W$, we can construct a $\mathfrak{g}$-homomorphism
\be
\Phi_\l^w : M_\l &\to W\otimes M_\mu,\qquad
\mu = \l - \mathrm{wt}\,w,
\ee
such that
\be
\Phi_\l^w(x_\l) &= w \otimes x_\mu + \dots.
\label{eq:highestweightverma}
\ee
Here, ``$\dots$" is a sum of terms of the form
\be
e_{\a_1}\cdots e_{\a_k} w \otimes e_{-\a_{k+1}} \cdots e_{-\a_m} x_\mu,
\ee
where $e_{\pm\a}\in \mathfrak{g}_\pm$ are raising/lowering operators.   Their coefficients are fixed by demanding that $\Phi_\l^w(x_\l)$ is $\mathfrak{g}_+$-primary, i.e.\ that it is killed by $1\otimes e_\a+e_\a\otimes 1$ for all positive roots $\a$. Finally, the action of $\Phi_\l^w$ on $\mathfrak{g}_-$-descendants of $x_\l$ is fixed by $\mathfrak{g}$-invariance. The construction of $\Phi_\l^w$ is completely analogous to the construction of $\Phi$ in (\ref{eq:Vermahomomorphism}) above. The vector (\ref{eq:highestweightverma}) is the analog of the primary state (\ref{eq:primaryansatz}).

Weight-shifting operators in the embedding space are in one-to-one correspondence with the homomorphisms $\Phi_\l^w$. In particular, they are labeled by weight-vectors of $W$. This is consistent with our argument based on Brauer's formula in the previous section.

The homomorphisms (\ref{eq:highestweightverma}) have been studied in \cite{YangBaxter}. Given two finite-dimensional representations $V,W$ with weight-vectors $v\in V$, $w\in W$, they satisfy the algebra
\be
(1\otimes \Phi_{\l-\mathrm{wt}\,v}^w)\circ  \Phi^v_\l = \Phi^{J(\l)(v\otimes w)}_\l,
\ee
where
\be
J(\l) \in \textrm{Aut}(V\otimes W)
\ee
is an invertible operator called the {\it fusion operator}. The fusion operator thus completely encodes the algebra of weight-shifting operators. It satisfies a number of interesting properties, and is closely related to solutions of the Yang-Baxter equations and integrability \cite{YangBaxter}. Most importantly for our discussion, the Arnaudon-Buffenoir-Ragoucy-Roche equation gives an explicit expression for $J(\l)$ in terms of generators of $\mathfrak{g}$ \cite{Arnaudon1998}. In principle, this answers the question posed at the beginning of this section. In practice, we will not need such a general answer in this work. We leave further exploration of the fusion operator and its applications to the future.

Another point of view on the algebra of weight-shifting operators is given by a special kind of $6j$ symbols, as we explain in appendix~\ref{app:sixjalgebra}.

\subsection{Weight-shifting operators in the embedding space}

Our construction of weight-shifting operators is extremely general, but it is inconvenient for computations because it is cumbersome to find the primary states $\cO'$. For practical computations, we can use the embedding  formalism \cite{Dirac:1936fq,Mack:1969rr,Boulware:1970ty,Ferrara:1973eg,Ferrara:1973yt,Cornalba:2009ax,Weinberg:2010fx,Costa:2011mg,SimmonsDuffin:2012uy,Iliesiu:2015qra,Costa:2016hju}, where the conformal group acts linearly. The tradeoff is that coordinates in the embedding space satisfy constraints and gauge redundancies, and we must take care to find differential operators respecting these conditions. The above construction tells us precisely when this should be possible.

The formalism described in \cite{Costa:2011mg} makes it easy to study operators in tensor representations of $\SO(d)$. Symmetric traceless tensors (STTs) of $\SO(d)$ are particularly simple. We will describe this case first in order to make contact with the examples above. However, our primary interest is in general representations, and for these it will be useful to use specialized formalisms for different spacetime dimensions.

\subsubsection{General dimensions}
\label{sec:embeddinggeneral}

In the embedding formalism, the conformal compactification of $\R^d$ is realized as the projective null cone in $\R^{d+1,1}$. We take the metric on $\R^{d+1,1}$ to be
\be
X^2 &= \eta_{mn}X^m X^n = -X^+ X^- + \sum_{\mu=1}^d X_\mu X^\mu.
\ee
A primary scalar $\cO(x)$ lifts to a function on the null cone $\cO(X)$ with homogeneity
\be
\label{eq:operatorhomogeneity}
\cO(\l X) &= \l^{-\De_\cO} \cO(X).
\ee
It is convenient to arbitrarily extend $\cO(X)$ outside the null cone, introducing the gauge redundancy
\be
	\cO(X)\sim \cO(X)+X^2\L(X).
\ee
A tensor operator $\cO^{\mu_1\cdots\mu_\ell}(X)$ lifts to a tensor $\cO^{m_1\cdots m_\ell}(X)$ in the embedding space, subject to gauge redundancies and transverseness
\be
\cO^{m_1\cdots m_\ell}(X) &\sim \cO^{m_1\cdots m_\ell}(X) + X^{m_i} \L^{m_1\cdots \widehat{m}_i \cdots m_\ell}(X),
\label{eq:gaugeinvsO}
\\
X_{m_i}\cO^{m_1\cdots m_i\cdots m_\ell}(X) &= 0,
\label{eq:transversenessO}
\ee
in addition to the homogeneity condition (\ref{eq:operatorhomogeneity}). For symmetric tensors, it is useful to introduce a polarization vector $Z^m$ and define
\be
\cO(X,Z) &\equiv \cO^{m_1\cdots m_\ell}(X)  Z_{m_1}\cdots Z_{m_\ell}.
\ee
Because of (\ref{eq:gaugeinvsO}), we must take $Z\.X=0$, and because of (\ref{eq:transversenessO}), we must identify $Z\sim Z+\l X$. Finally, when $\cO^{m_1\cdots m_\ell}$ is traceless, we can impose $Z^2=0$.

We can summarize these constraints as follows. Let $I$ be the ideal generated by $\{X^2,X\.Z,Z^2\}$, and let $R$ be the ring of functions of $(X,Z)$ invariant under $Z\to Z+\l X$.  Symmetric tensor operators are elements of $R/(R\cap I)$ which are homogeneous in both $X$ and $Z$.  For a differential operator in $X,Z$ to be well-defined on this space, it must take $R\to R$ and also preserve the ideal $R\cap I$.

The construction in section~\ref{sec:diffopsfromtensor} tells us when such operators should exist. For example, consider the case where $W=\myng{(1)}$ is the vector representation of $\SO(d+1,1)$ and $\cO(X,Z)$ has spin $\ell$ and dimension $\De$. Given the decomposition (\ref{eq:vectordecomposition}), we should be able to find differential operators with a vector index in the embedding space, taking\footnote{There will also exist differential operators producing other representations in the tensor product of the vector and spin-$\ell$ representations of $\SO(d)$ (generically there is also the hook Young diagram). According to~\eqref{eq:usingBrauer}, when acting on general (non-STT) representations generically there are $d+2$ operators corresponding to the vector representation. However, to describe these we would need a formalism with more polarization vectors as in \cite{Costa:2014rya}.}
\be\label{eq:vectoroperators_action}
\cD_m^{-0}: [\De,\ell] &\to [\De-1,\ell],\nn\\
\cD_m^{0-}: [\De,\ell] &\to [\De,\ell-1],\nn\\
\cD_m^{0+}: [\De,\ell] &\to [\De,\ell+1],\nn\\
\cD_m^{+0}: [\De,\ell] &\to [\De+1,\ell].
\ee
Our strategy for finding them is to start with a suitable ansatz and fix the coefficients by requiring that $\cD_m$ preserve $R$ and $R\cap I$. (We give more details in appendix~\ref{app:D4coefficients}.) We find
\be\label{eq:vectoroperators}
\cD_m^{-0} &= X_m, \nn\\
\cD_m^{0-} &=
\p{(\De-d + 2 - \ell)\de_m^n + X_m \pdr{}{X_n}} \p{(d - 4 + 2 \ell)  \pdr{}{Z^n}
-Z_n \frac{\ptl^2}{\ptl Z^2}},
\nn\\
\cD_m^{0+} &= (\ell+\De)Z_m + X_m Z\.\pdr{}{X},\nn\\
\cD_m^{+0} &= c_1\pdr{}{X^m} + c_2 X_m \frac{\ptl^2}{\ptl X^2} + c_3 Z_m \frac{\ptl^2}{\ptl Z\.\ptl X} + c_4 Z\.\pdr{}{X} \pdr{}{Z^m}\nn\\
&\quad + c_5 X_m Z\.\pdr{}{X}\frac{\ptl^2}{\ptl Z\.\ptl X} + c_6 Z_m Z\.\pdr{}{X} \frac{\ptl^2}{\ptl Z^2} + c_7 X_m \p{Z\.\pdr{}{X}}^2 \frac{\ptl^2}{\ptl Z^2},
\ee
where the coefficients $c_i$ are given in appendix~\ref{app:D4coefficients}. For now, we simply quote
\be
\frac{c_1}{c_2} &= -2\p{\frac d 2-1-\De}.
\ee
Thus, when acting on scalar operators $\cO(X)$, $\cD_m^{+0}$ is proportional to the familiar Todorov operator \cite{Dobrev:1975ru}
\be\label{eq:todorovop}
\cD_m^{+0} &\propto \p{\frac d 2 + X\.\pdr{}{X}}\pdr{}{X^m} - \frac 1 2 X_m \frac{\ptl^2}{\ptl X^2} + O\p{\pdr{}{Z}}.
\ee
This simplified version of $\cD^{+0}_m$ (together with $\cD^{-0}_m$) appears in tractor calculus, where it is known as Thomas operator~\cite{Thomas1926,bailey1994}.

The overall normalization of our differential operators is a convention. It is useful to choose conventions where the coefficients $c_i$ are polynomials in $\De,\ell$ of the smallest possible degree. If we like, factors of $\De,\ell$ can then be replaced with
\be
\De = -X\.\pdr{}{X},\qquad \ell = Z\.\pdr{}{Z},
\ee
so that $\cD$ can be expressed without reference to the operator it acts on.

Note that when acting on scalar $\cO$ there a unique non-vanishing operator of the lowest scaling dimension, $\cD^{-0}_m$. According to theorem~\ref{thm:operatorclassification}, this is true in general. From the discussion in section~\ref{sec:tensorproducts} it follows that this operator should correspond to multiplication by the conformal Killing tensor $w^a_A(x)$ as in~\eqref{eq:simplestensortprimary}. This gives a general way of finding $w^a_A(x)$ from the embedding space formalism.

For example, one can check that the primary operator $w(x)$ corresponding to the vector representation of the conformal group is given by
\be
	w(x)=w_m(x)e^m=e^m\cD^{-0}_m=X^m e_m=e_+ + x^\mu e_\mu + x^2 e_-,
\ee
where $e_+,e_-$ and $e_\mu$ form the light cone coordinate basis of the vector representation. It solves the equation
\be
	\ptl_\mu\ptl_\nu w(x)-\mathrm{trace}=0.
\ee
Let us now revisit the example from section~\ref{sec:tensorproducts}. 
Let $\cO(x)$ be a scalar primary of dimension $\De$, as in section~\ref{sec:tensorproducts}. We then compute\footnote{Recall that on the Poincare section we have $X=(1,x^2,x^\mu)$ and $Z=(0,2(x\cdot z),z^\mu)=z^\mu\ptl_\mu X$ where the coordinates are ordered as $(X^+,X^-,X^\mu)$.}
\be
	e^m\otimes \cD^{-0}_m\cO(x)&=w(x)\otimes \cO(x),\nn\\
	e^m\otimes \cD^{0+}_m\cO(x)&=z^\mu\p{\De\ptl_\mu w(x)\otimes \cO(x)+w(x)\otimes\ptl_\mu\cO(x)},\nn\\
	e^m\otimes \cD^{+0}_m\cO(x)&=\frac{c_1\De}{d}\ptl^2 w(x)\otimes \cO(x)+c_1\ptl_\mu w(x)\otimes \ptl^\mu\cO(x)+c_2 w(x)\otimes \ptl^2\cO(x),
\ee
where $c_i$ are as in~\eqref{eq:vectoroperators}. It is easy to see that this is consistent with~\eqref{eq:exansatz} and~\eqref{eq:exsolution}. Naturally, $e^m\otimes \cD^{0-}_m\cO(x)$ vanishes when $\cO(x)$ is a scalar.

\subsubsection{1 dimension}
\label{sec:embedding1d}

To find the most general conformally-covariant differential operators, it is useful to employ a formalism specialized to the given spacetime dimension. The simplest case is $1$-dimension, where the conformal group is $Spin(2,1)$.\footnote{We use the conventions of~\cite{Iliesiu:2015qra} for 2+1 dimensions.} The Lorentz group is $Spin(1)=\mathbb{Z}_2$ (see below) and the primary operators are labeled by a scaling dimension $\De$ and a spin $s=\pm$. We will denote the corresponding Verma modules by~$V_{\De,s}$. Because the global 2-dimensional conformal group is a product of 1-dimensional groups, the results of this section can also be applied in 2-dimensions.

Note that the simply-connected conformal group is $Spin(2,1)\simeq SL(2,\mathbb{R})$. It acts by M\"obius transformations, 
\be
	\begin{pmatrix}
	a & b \\ c & d
	\end{pmatrix}:\;
	x\to \frac{ax+b}{cx+d},\qquad  ad-bc=1.
\ee
The subgroup which fixes the origin is given by $b=0$. We can exclude special conformal transformations by setting $c=0$. The remaining subgroup is a product of dilatations $\mathbb{R}_+$ parametrized by $|a|$ and the Lorentz group $\mathbb{Z}_2$ parametrized by the sign of $a$. This is why we say that $Spin(1)=\mathbb{Z}_2$.\footnote{The fields which have spin $s=+$ are the usual scalars on the circle. The fields which have $s=-$ are anti-periodic fermions.}

The vector representation of $Spin(2,1)$ is equivalent to the symmetric square of the spinor represenation, and in the embedding formalism we can define
\be
	X_{(\alpha\beta)} = \gamma^m_{(\alpha\beta)} X_m,\qquad \gamma^m_{(\alpha\beta)}=\Omega_{\alpha\alpha'}(\gamma^m)^{\alpha'}{}_\beta.
\ee
In this notation the constraint $X^2=0$ can be solved as 
\be
	X_{(\alpha\beta)} = \chi_\alpha\chi_\beta,
\ee
where $\chi_\alpha$ is a real spinor in the fundamental representation of $SL(2,\mathbb{R})$. Note that $\chi$ is odd under the center of $SL(2,\mathbb{R})$. This parametrization has the advantage that now the embedding-space operators can be taken to depend on $\chi_\alpha$,
\be\label{eq:chiscaling}
	\cO(\lambda\chi)=\lambda^{-2\De_\cO}\cO(\chi),\quad \lambda>0.
\ee
Notice that both $\chi$ and $-\chi$ correspond to the same $X$. The correct transformation property of $\cO(\chi)$ under this transformation comes from the $\mathbb{Z}_2$-spin,
\be\label{eq:onedimcenterproperty}
\cO(-\chi)=s\cO(\chi).
\ee
This property will be important for the construction of tensor structures in section~\ref{sec:sixJ1d}.

The embedding formalism in terms of $\chi$ is useful because the conformal group still acts linearly, but now there is no analogue of the ideal $I$ which needs to be preserved by the embedding space differential operators. We have the following relation between $\chi$ and $X$ derivatives,
\be\label{eq:chiXderrelation}
	\pdr{}{\chi^\alpha} = (\gamma^m)^{\beta}{}_\alpha\chi_\beta \pdr{}{X^m}.
\ee
Using this relation in an arbitrary differential operator written in terms of $\chi$ will automatically produce the terms necessary to preserve the ideal $I$ in $X$-space. For example, we can recover the $1$-dimensional version of the operator $\cD^{+0}_m$ (c.f. \eqref{eq:todorovop}),
\be
	(\gamma^m)^{(\alpha\beta)}\pdr{}{\chi^\alpha}\pdr{}{\chi^\beta} \propto \left(X\cdot \pdr{}{X}-\frac{1}{2}\right)\pdr{}{X_m}-\frac{1}{2}X^m\pdr{{}^2}{X^2}.
\ee

A general embedding space differential operator is an arbitrary combination of $\chi_\alpha$ and $\ptl_\alpha=\pdr{}{\chi^\alpha}$. The combinations irreducible under $Spin(2,1)$ are 
\be
	\cD^{j,i}_{\alpha_1\ldots \alpha_j} = \chi_{(\alpha_1}\cdots\chi_{\alpha_{j-i}}\ptl_{\alpha_{j-i+1}}\cdots\ptl_{\alpha_{2j})},\qquad i=-j,\ldots j.
	\label{eq:general1dop}
\ee
Of course, we can also add combinations of $\chi_\alpha \ptl^\alpha$, but these simply act as scalars due to \eqref{eq:chiscaling}, so we can ignore this possibility.
By construction, this differential operator transforms in the spin-$j$ representation of $Spin(2,1)$, changes the scaling dimension by $i$, and exchanges bosons with fermions if $j$ is half-integer,
\be\label{eq:diffmap1d}
\cD^{j,i}: [\De,s] &\to [\De+i,(-1)^{2j}s].
\ee

It is easy to find the group-theoretic interpretation for $\cD^{j,i}$. Indeed, the spin-$j$ representation decomposes as
\be
	j\to \bigoplus_{i=-j}^{j} ((-1)^{2j})_i,  
\ee
which means that for a generic $\De$ we have the tensor product decomposition
\be
	j\otimes V_{\De,s} = \bigoplus_{i=-j}^{j} V_{\De+i,(-1)^{2j}s}.
\ee
Thus, we find explicitly the expected one-to-one correspondence between the differential operators $\cD^{j,i}$ and the terms in this tensor product. We also see explicitly that the differential operators are labeled by the weights of the spin-$j$ representation, in accordance with~\eqref{eq:usingBrauer}.

Let us see what our operators look like in $x$-coordinate space. It is easy to check that the usual Poincare section $X^+=1$ corresponds to $\chi^1=x$, $\chi^2=1$.\footnote{And also to minus these values, since there is a redundancy $\chi\sim -\chi$.} We can therefore write the embedding space operator in terms of the $x$-space operator as (multiplying also by $\mathrm{sign}\,\chi^2$ for $s=-$)
\be
	\cO(\chi)=\frac{1}{|\chi^2|^{2\De}}\cO\p{\frac{\chi^1}{\chi^2}}.
\ee
We therefore see that $\chi_1$ and $\chi_2$ derivatives act as
\be
	\ptl_1&=\pdr{}{\chi^1}=\pdr{}{x},\\
	\ptl_2&=\pdr{}{\chi^2}=-x\pdr{}{x}-2\De.
\ee
These formulas are valid for higher order derivatives if we follow the convention that $\De$ in the last formula is increased by $\half$ by every $\ptl_\alpha$.

\subsubsection{3 dimensions}

In 3-dimensions, we use the formalism and conventions of \cite{Iliesiu:2015qra}.\footnote{In particular, we use Lorentzian signature in this section.} The conformal group is $\SO(3,2)$, which has $\Sp(4,\R)$ as a double cover. The most general Lorentz representation is the $2\ell$-th symmetric power of the spinor representation of $\SO(2,1)$, where $\ell\in \frac 1 2\mathbb{N}$. An operator $\cO^{\a_1\cdots \a_{2\ell}}(x)$ lifts to an embedding space operator
$\cO^{a_1\cdots a_{2\ell}}(X)$
with $2\ell$ indices for the fundamental of $\Sp(4,\R)$, satisfying the homogeneity property
\be
\cO^{a_1\cdots a_{2\ell}}(\l X) &= \l^{-\De_\cO - \ell} \cO^{a_1\cdots a_{2\ell}}(X).
\ee
 It is useful to introduce a polarization spinor $S_a$, and define
\be
\cO(X,S) &\equiv S_{a_1} \cdots S_{a_{2\ell}} \cO^{a_1\cdots a_{2\ell}}(X).
\ee
The polarization spinors are constrained to satisfy
\be
S_a X^a{}_b &= 0,\quad\textrm{where}\quad X^a{}_b \equiv X^m (\G_m)^a{}_b,
\ee
where $(\G_m)^a{}_b$ are generators of the Clifford algebra of $\SO(3,2)$.  For convenience, we also introduce the notation
\be
X_{ab} = \Omega_{ac}X^c{}_b,\qquad X^{ab}=X^a{}_c \Omega^{cb},
\ee
where $\Omega_{ac}=\Omega^{ac}$ is the symplectic form for $\Sp(4,\R)$.

Arbitrary finite-dimensional representations of $\SO(3,2)$ can be obtained from tensors of the spinor representation $\mathcal{S}$. Thus, all the weight-shifting operators in 3d can be obtained from products of weight-shifting operators for $\mathcal{S}$. Under $\SO(3,2)\to \SO(1,1)\x \SO(2,1)$, we have the decomposition
\be
\mathcal{S} &\to (\mathcal{S})_{-\frac 1 2} \oplus (\mathcal{S})_{\frac 1 2}.
\ee
Thus, we should be able to find differential operators with a fundamental index for $\Sp(4,\R)$ that take
\be\label{eq:3Doperators_action}
\cD^{\pm\pm}_a:[\De,\ell] &\to \left[\De\pm \tfrac 1 2, \ell \pm \tfrac 1 2\right].
\ee
Note that again the differential operators are labeled by weights of $\mathcal{S}$, consistently with~\eqref{eq:usingBrauer}. They are given by 
\be\label{eq:3Doperators}
\cD^{-+}_a &= S_a\nn\\
\cD^{--}_a &= X_{ab} \pdr{}{S_b}\nn\\
\cD^{++}_a &= 2(\De-1)(\ptl_X)_{ab} \Omega^{bc} S_c + S_a \p{S_b \Omega^{bc}(\ptl_X)_{cd}\pdr{}{S_d}}\nn\\
\cD^{+-}_a &=
4(\De-1)(1+\ell-\De)\Omega_{ab}\pdr{}{S_b} - 2(1+\ell-\De) X_{ab} (\ptl_X)^{bc}\Omega_{cd} \pdr{}{S_d}\nn\\
&\qquad + S_a\p{\pdr{}{S_c} X_{cd} ( \ptl_X)^{de}\Omega_{ef} \pdr{}{S_f}}.
\ee
We have determined the coefficients by demanding that these operators preserve the ideal generated by $X^2$ and $S_a X^a{}_b$. The differential operators (\ref{eq:3Doperators}) are analogous to $\chi_\a$ and $\pdr{}{\chi^\a}$ in the 1-dimensional case. By taking products of them, we can build weight-shifting operators in arbitrary representations of $\SO(3,2)$, analogous to the 1d operators (\ref{eq:general1dop}). See also appendix~\ref{app:sixjalgebra}.

\subsubsection{4 dimensions}
\label{sec:basic_differential_operators_4D}

In 4d, we can use the embedding space formalism of \cite{SimmonsDuffin:2012uy,Siegel:2012di,Elkhidir:2014woa,Echeverri:2015rwa,Echeverri:2016dun,Cuomo:2017wme}. Our conventions are those of \cite{Cuomo:2017wme}. A general Lorentz representation is now labeled by two weights $(\ell,\bar \ell)$, where $\ell,\bar \ell \in \Z_{\geq 0}$.  (Spin-$\ell$ symmetric traceless tensor representations correspond to the case $\ell=\bar \ell$.) An operator $\cO^{\a_1\cdots \a_{\ell}\dot\a_1\cdots\dot\a_{\bar\ell}}(x)$ lifts to an embedding space operator
\be
\cO(X,S,\bar S) &= S_{a_1}\cdots S_{a_{\ell}} \bar S^{b_1}\cdots \bar S^{b_{\bar\ell}} \cO^{a_1\cdots a_{\ell}}_{b_1\cdots b_{\bar \ell}}(X),
\ee
where we have introduced polarization spinors $S_a,\bar S^a$ transforming as left- and right-handed spinors of $\SO(4,2)$, or equivalently fundamentals and anti-fundamentals of $\SU(2,2)$.  The polarization spinors satisfy
\be
\label{eq:idealin4d}
S_a \bar{\mathbf{X}}^{ab} =0,\qquad \bar S^a \mathbf{X}_{ab} = 0, \qquad\bar S^a S_a = 0,
\ee
where\footnote{Our conventions for the conformal algebra and embedding space in 4d are those of \cite{Cuomo:2017wme}.}
\be
\mathbf{X}_{ab} \equiv \Sigma_{ab}^m X_m,\qquad \bar{\mathbf{X}}^{ab} \equiv \bar \Sigma^{ab}_m X^m.
\ee
Let us also introduce the shorthand notation
\be
\ptl_{\bar S,a}&\equiv \frac{\partial}{\partial \overline S^a},
&
\ptl_S^a&\equiv \frac{\partial}{\partial S_{a}},
\nn\\
\partial_{ab}&\equiv\Sigma^m_{ab}\frac{\partial}{\partial X^m},
&
\overline\partial^{ab}&\equiv\overline\Sigma^{m\,ab}\frac{\partial}{\partial X^m}.
\ee

General representations of the conformal group $\SO(4,2)$ can be obtained by tensoring with the left and right-handed spinors. Thus, our algebra of differential operators is generated by those associated with the spinor representations. To label these operators, it is convenient to use~\eqref{eq:usingBrauer}. Let us denote the weights so that the highest weight of the Verma module for $\cO$ is $(2\Delta,\ell,\bar\ell)$. Then the representations $\mathcal{S}$ and $\overline{\mathcal{S}}$ consist of the following weights,
\begin{align}\label{eq:Sweights}
	\Pi(\mathcal{S})&=\{(-,+,0),(-,-,0),(+,0,+),(+,0,-)\},\\\label{eq:Sbarweights}
	\Pi(\overline{\mathcal{S}})&=\{(-,0,+),(-,0,-),(+,+,0),(+,-,0)\}.
\end{align}
Note that basis vectors for $\mathcal{S}$ are $e^a$ (so that we can contract them with $S_a$) and for $\overline{\mathcal{S}}$ the basis vectors are $e_a$.

According to~\eqref{eq:usingBrauer}, the operators $\cD^a$ associated with $\mathcal{S}$ are then labeled by the weights~\eqref{eq:Sbarweights} of $\mathcal{S}^*=\overline{\mathcal{S}}$, and the operators $\overline{\cD}_a$ associated with $\overline{\mathcal{S}}$ are labeled by the weights~\eqref{eq:Sweights} of $\overline{\mathcal{S}}^*=\mathcal{S}$. These operators have the following explicit expressions, 
\begin{align}\label{eq:4Doperators}
\cD^a_{-0+} &\equiv \overline S^a,\nn\\
\cD^a_{-0-} &\equiv \overline{\mathbf{X}}^{ab}\ptl_{\bar S,b},\nn\\
\cD^a_{++0} &\equiv {\bar a}\overline\partial^{ab}S_{b} +\overline S^a(S\overline\partial\ptl_{\bar S}),\nn\\
\cD^a_{+-0} &\equiv {\bar b} c\ptl_S^{a}+
{\bar b}\overline S^a(\ptl_S \ptl_{\bar S})+
c\mathbf{X}_{bc}\overline\partial^{ab}\ptl_S^{c} -
\overline S^{a}(\mathbf{X}_{bc}\overline\partial^{bd} \ptl_S^{c}\ptl_{\bar S,d}),\nn\\
\overline{\cD}_a^{-+0} &\equiv S_{a},\nn\\
\overline{\cD}_a^{--0} &\equiv \mathbf{X}_{ab}\ptl_S^b,\nn\\
\overline{\cD}_a^{+0+} &\equiv  a\partial_{ab}\overline S^b +S_{a}(\overline S\partial \ptl_S),\nn\\
\overline{\cD}_a^{+0-} &\equiv b c\ptl_{\bar S,a} +
 bS_{a}(\ptl_{\bar S}\ptl_S)+
 c\overline{\mathbf{X}}^{bc}\partial_{ab}\ptl_{\bar S,c}-
S_{a}(\overline{\mathbf{X}}^{bc}\partial_{bd}   \ptl_{\bar S,c}\ptl_S^d),
\end{align}
where
\be
a &= 1-\De+\tfrac \ell 2 -\tfrac {\bar \ell} 2,
& 
\bar a &= 1-\De - \tfrac \ell 2 + \tfrac {\bar \ell} 2,\nn\\
b &= 2(\bar \ell+1),
&
\bar b &= 2(\ell+1),\nn\\
c &= -2 + \De - \tfrac{\ell+\bar\ell}{2}.
\ee
The coefficients above come from requiring that the operators preserve the ideal generated by the relations (\ref{eq:idealin4d}), together with $X^2=0$. We have added these operators to the \texttt{CFTs4D} Mathematica package described in~\cite{Cuomo:2017wme}.

\section{Crossing for differential operators}
\label{sec:crossingfordifferentialoperators}

The results of section~\ref{sec:differentialoperators} give us a large variety of conformally-covariant differential operators. In the present section we consider their action on conformally-invariant\footnote{We are making a distinction between conformally-covariant and conformally-invariant objects. For us, the former carry finite-dimensional $\SO(d+1,1)$ labels, whereas the latter do not.} correlation functions of local operators. The result of such an action is a conformally-covariant $n$-point function, which can also be interpreted as a conformally-invariant $(n+1)$-point function that includes the degenerate field $w^a(x)$. We will first describe the structure of such correlation functions and then establish a convenient graphical notation for the action of the differential operators. This will help us elucidate a rich structure of such actions at the end of this section.

\subsection{Conformally-covariant tensor structures}
\label{sec:tensorstructures}

Consider an $n$-point correlation function with an additional formal insertion of an element $e^A$ of the finite-dimensional representation $W$ of the conformal group $\SO(d+1,1)$,
\be\label{eq:covariantnpoint}
	\<\cO_1^{a_1}(x_1)\cdots \cO_n^{a_n}(x_n)\>^A\equiv \<\cO_1^{a_1}(x_1)\cdots \cO_n^{a_n}(x_n)e^A\>.
\ee
Note that this is a purely formal construct, i.e.\ this expression is simply a shorthand for a function of $n$ points which carries indices $a_i$, $A$, and has transformation properties identical to those satisfied by a correlation function under the assumption that
\be
	U_g e^A U^{-1}_g=g\cdot e^A,
\ee
and $g\cdot e^A$ is defined by~\eqref{eq:eAtransformation}.

As discussed in section~\ref{sec:finitedim}, we can also view~\eqref{eq:covariantnpoint} as a $(n+1)$-point conformally-invariant correlation function with the primary $w^b(y)$ of $W$,
\be\label{eq:covarianttonplus1}
	\<\cO_1^{a_1}(x_1)\cdots \cO_n^{a_n}(x_n) w^b(y)\>\equiv \<\cO_1^{a_1}(x_1)\cdots \cO_n^{a_n}(x_n) e^A\> w^{b}_A(y),
\ee
subject to the conformal Killing differential equation satisfied by $w^b(y)$. This interpretation will be useful to us later on. In this section we stick with~\eqref{eq:covariantnpoint}.

Similarly to the usual conformally-invariant correlation functions, we have an expansion in tensor structures,
\be
	\<\cO_1^{a_1}(x_1)\cdots \cO_n^{a_n}(x_n)e^A\>=\mathbb{T}^{a_1\ldots a_n,A}_I(x_i) g^I(\bu),
\ee
which now carry the $\SO(d+1,1)$ index $A$. Here $\bu$ are the conformal cross-ratios of points $x_i$. The structures $\mathbb{T}^{a_1\ldots a_n,A}$ can be constructed using embedding space methods, since there one explicitly works with objects which transform in fundamental representations of $\SO(d+1,1)$. In this subsection we are going to classify such tensor structures by extending the conformal frame approach of~\cite{Mack:1976pa,Kravchuk:2016qvl}.

The basic idea is to maximally use conformal symmetry to bring as many $x_i$ as possible to some standard positions $x'_i$. The resulting configuration $x'_i$ will be invariant under the subgroup $\mathcal{G}_n\subset \SO(d+1,1)$ of the conformal group that stabilizes $n$ points. In particular
\be
\mathcal{G}_n &= \begin{cases}
\SO(1,1)\times \SO(d) & n=2,\\
\SO(d+2-m) & n\geq 3,
\end{cases}
\ee
where $m=\min(n,d+2)$. The tensor $\mathbb{T}(x'_i)$ transforms as an element in\footnote{In writing a tensor product of representations of different groups, we assume that each representation is restricted to the largest common subgroup. In (\ref{eq:Ttensorspace}), we implicitly restrict $W$ to $\SO(1,1)\x \SO(d)\subset \SO(d+1,1)$.}
\be
\label{eq:Ttensorspace}
	W\otimes \bigotimes_{k=1}^n(\rho_k)_{\De_k},
\ee
and by construction is invariant under $\mathcal{G}_n$. It is easy to check~\cite{Kravchuk:2016qvl} that this is the only restriction for the tensor $\mathbb{T}^{a_1\ldots a_n,A}(x'_i)$ and the conformally-covariant tensor structures are then in one-to-one correspondence with the invariants of $\mathcal{G}_n$,\footnote{The notation $(\rho)^H$ denotes the $H$-invariant subspace of $\rho$, where $\rho$ is a representation of $G$ and $H\subseteq G$.}
\be
\left(W\otimes \bigotimes_{k=1}^n(\rho_k)_{\De_k}\right)^{\mathcal{G}_n}.
\ee
In practice we always use the decomposition~\eqref{eq:Wdecomposition} in this formula and identify the tensor structures with
\be\label{eq:cfrule}
\bigoplus_{i=-j}^{j}\left((W_i)_i\otimes \bigotimes_{k=1}^n(\rho_k)_{\De_k}\right)^{\mathcal{G}_n}.
\ee

\subsection{Tensor structures and diagrams}

Let us work through some examples of covariant $n$-point functions and the counting rule (\ref{eq:cfrule}). At the same time, we will introduce a useful diagrammatic language for describing tensor structures and differential operators.

\subsubsection{Invariant two-point functions}

Let us denote a conformally-invariant two-point structure by 
\be
\<\cO_1\cO_2\> \quad=\quad
\diagramEnvelope{\begin{tikzpicture}[anchor=base,baseline]
	\node (opO) at (-1,0) [left] {$\cO_1$};
	\node (opOprime) at (1,0) [right] {$\cO_2$};
	\node (vert) at (0,0) [twopt] {};
	\draw [spinning] (vert) -- (opO);
	\draw [spinning] (vert) -- (opOprime);
\end{tikzpicture}}.
\ee
It is well-known that there is at most one such structure, but let us re-derive this fact in the language of section~\ref{sec:tensorstructures}, where it corresponds to the case $n=2$ and $W=\bullet$.

Given $x_1$ and $x_2$, we can apply a conformal transformation to set $x_1=0$ and $x_2=\infty$. Then the group $\mathcal{G}_2=\SO(1,1)\times \SO(d)$ which fixes the two points consists of dilatations and rotations around $0$. Sending the second operator to infinity has the effect that $\cO_2$ effectively changes the sign of its scaling dimension, and transforms in the reflected representation\footnote{Given a representation $\rho$ with generators $\rho_{\mu\nu}$ the reflected representation is defined as $\rho^P_{\mu\nu}=P^{\mu'}_\mu P^{\nu'}_\nu \rho_{\mu'\nu'}$, where $P$ is a spatial reflection matrix. Formally, conjugating by reflection is an outer automorphism of $\SO(d)$, and hence permutes the representations of $\SO(d)$.} $\rho_2^P$ under $\SO(d)$. Thus, two-point structures correspond to the $\mathcal{G}_2$-invariants in
\be
	(\rho_1)_{\Delta_1}\otimes (\rho_2^P)_{-\Delta_2}.
\ee
There is at most one such invariant, which exists iff $\rho_1=(\rho_2^P)^*$ and $\Delta_1=\Delta_2$. The dual-reflected representation, which we denote by $(\rho_2^P)^*\equiv \rho_2^\dag$ is the same as the complex conjugate representation in Lorentzian signature.

\subsubsection{Differential operators}

A differential operator $\cD^A:\cO\to \cO'$ takes a conformally-invariant structure for $\cO$ to a conformally-covariant structure for $\cO'$, or equivalently an invariant structure for $\cO'$ and $W$:
\be
\cD^A\<\cO \cdots\> &\sim \<e^A \cO' \cdots\>.
\ee
We denote such a differential operator by 
\begin{equation}
\cD^{(a)A} \quad=\quad
\diagramEnvelope{\begin{tikzpicture}[anchor=base,baseline]
	\node (vert) at (0,0) [threept] {$a$};
	\node (opO) at (-0.5,-1) [below] {$\cO$};
	\node (opOprime) at (-0.5,1) [above] {$\cO'$};
	\node (opFinite) at (1,0) [right] {$W$};	
	\draw [spinning] (vert)-- (opOprime);
	\draw [spinning] (opO) -- (vert);
	\draw [finite with arrow] (vert) -- (opFinite);
\end{tikzpicture}}.
\label{eq:diffoppicture}
\end{equation}
The label $a$ runs over the possible operators classified by theorem~\ref{thm:operatorclassification}. We use a wavy line to indicate a finite-dimensional representation. 

\subsubsection{Covariant two-point functions}
\label{sec:convarianttwopt}
Consider acting with a differential operator $\cD^{(m)A}:[\De_1,\rho_1]\to [\De'_1,\lambda_1]$ on an invariant two-point function. In diagrammatic language, this is denoted by connecting an outgoing arrow from the two-point function with an incoming arrow for the differential operator,
\be
(\cD^{(m)A})^c{}_a\<\cO_1^a(x_1)\cO_2^b(x_2)\> \quad = \quad
\diagramEnvelope{\begin{tikzpicture}[anchor=base,baseline]
	\node (vert) at (0,0) [threept] {$m$};
	\node (opO1mid) at (-0.35,-0.8) [twopt] {};
	\node (opO1) at (-0.8,-1.6) [below] {$\cO_2$};
	\node (opO2) at (-0.5,1) [above] {$\cO_1'$};
	\node (opW) at (1,0) [right] {$W$};		
	\draw [spinning] (opO1mid) -- (vert);
	\draw [spinning] (opO1mid) -- (opO1);
	\draw [spinning] (vert)-- (opO2);
	\draw [finite with arrow] (vert)-- (opW);
\end{tikzpicture}}.
\label{eq:diagramforcovarianttwopt}
\ee
The result can be interpreted as a covariant two-point structure for $\cO_1'$, $\cO_2$, and $W$.
Such structures are counted by $\SO(1,1)\x\SO(d)$-invariants in
\be
	\bigoplus_{i=-j/2}^{j/2} (W_i)_{i}\otimes (\lambda_1)_{\De'_1}\otimes (\rho_2^P)_{-\Delta_2}.
\ee
Invariants exist whenever $\Delta'_1=\Delta_2-i=\Delta_1-i$ and $\lambda_1\in (W_i)^*\otimes (\rho_2^P)^*=(W_i)^*\otimes \rho_1$.\footnote{We have assumed that $\De_1=\De_2$ and $\rho_1=(\rho_2^P)^*$ so that $\<\cO_1\cO_2\>$ is nonvanishing.}
Note that these are exactly the conditions for the existence of $\cD^A$ in theorem~\ref{thm:operatorclassification}. Thus, the number of non-vanishing diagrams (\ref{eq:diagramforcovarianttwopt}) is precisely equal to the number of tensor structures for $\<\cO_2\cO_1'e^A\>$. In other words, all covariant two-point structures can be obtained by acting with differential operators on an invariant two-point structure.

\subsubsection{Invariant three-point functions}

We denote conformally-invariant three-point structures by
\begin{equation}
\<\cO_1\cO_2\cO_3\>^{(a)} \quad=\quad
\diagramEnvelope{\begin{tikzpicture}[anchor=base,baseline]
	\node (vert) at (0,0) [threept] {$a$};
	\node (opO1) at (-0.5,-1) [below] {$\cO_1$};
	\node (opO2) at (-0.5,1) [above] {$\cO_2$};
	\node (opO3) at (1,0) [right] {$\cO_3$};	
	
	\draw [spinning] (vert)-- (opO1);
	\draw [spinning] (vert)-- (opO2);
	\draw [spinning] (vert)-- (opO3);
\end{tikzpicture}}.
\end{equation}
The label $a$ runs over possible tensor structures, which are classified by $\mathcal{G}_3=\SO(d-1)$ singlets
\be\label{eq:cfthreeptcount}
(\rho_1\otimes \rho_2\otimes \rho_3)^{\SO(d-1)}.
\ee
A physical three-point function is a sum over tensor structures with different OPE coefficients $\l_m$,
\be\label{eq:physicalthreept}
\<\cO_1\cO_2\cO_3\>=\sum_{m=1}^{N_3}\lambda_m\,\<\cO_1\cO_2\cO_3\>^{(m)},
\ee
where $N_3=\dim (\rho_1\otimes \rho_2\otimes \rho_3)^{\SO(d-1)}$. When there is a unique three-point structure ($N_3=1$), we often omit the index $m$.\footnote{Since we never work with physical three-point functions~\eqref{eq:physicalthreept}, there is no danger of confusion.}

\subsubsection{Covariant three-point functions}
\label{sec:convarinantthreept}
Consider now acting on an invariant three-point structure with a differential operator. Let us begin with a three-point structure $\<\cO_1\cO_2\cO_3'\>^{(a)}$, and suppose that $\cO_3'$ transforms in the representation $[\De_3+i,\l]$.  The label $a$ runs over singlets in 
\be
(\rho_1\otimes\rho_2\otimes \l)^{\SO(d-1)}.
\ee
By theorem~\ref{thm:operatorclassification}, we have a differential operator $\cD^{(b)A}:[\De_3+i,\l]\to [\De_3,\rho_3]$ whenever
\be\label{eq:cbselectionrule}
\rho_3\in (W_i)^*\otimes \lambda \quad\Leftrightarrow\quad\lambda\in W_i\otimes\rho_3.
\ee
By acting with $\cD^{(b)A}$ on $\<\cO_1\cO_2\cO_3'\>^{(a)}$, we can form a covariant three-point structure for $\<\cO_1\cO_2\cO_3e^A\>$,
\be
\label{eq:covariantthreeptdiagrams}
(\cD^{(b)A})^{a_3}{}_c\<\cO^{a_1}_1(x_1)\cO_2^{a_2}(x_2)\cO_3'^{c}(x_3)\>^{(a)} \quad =\quad
\diagramEnvelope{\begin{tikzpicture}[anchor=base,baseline]
	\node (vertL) at (0,0) [threept] {$a$};
	\node (vertR) at (2,-0.05) [threept] {$b$};
	\node (opO1) at (-0.5,-1) [below] {$\cO_1$};
	\node (opO2) at (-0.5,1) [above] {$\cO_2$};
	\node (opO3) at (2.5,1) [above] {$\cO_3$};
	\node (opW) at (2.5,-1) [below] {$W$};	
	\node at (1,0.1) [above] {$\cO_3'$};	
	\draw [spinning] (vertL)-- (opO1);
	\draw [spinning] (vertL)-- (opO2);
	\draw [spinning] (vertL)-- (vertR);
	\draw [spinning] (vertR)-- (opO3);
	\draw [finite with arrow] (vertR)-- (opW);
\end{tikzpicture}}.
\ee

Let us count the number of diagrams (\ref{eq:covariantthreeptdiagrams}) by summing over the allowed $\cO'_3$, $a$ and $b$. Taking into account the selection rule~\eqref{eq:cbselectionrule}, we have
\begin{multline}
\label{eq:dimofcovariantthreept}
	\sum_{i=-j}^j\sum_{\lambda\in W_i\otimes \rho_3}\dim\p{\rho_1\otimes\rho_2\otimes \l}^{\SO(d-1)}\\=\dim\p{\bigoplus_{i=-j}^j\bigoplus_{\lambda\in W_i\otimes \rho_3}\rho_1\otimes\rho_2\otimes \l}^{\SO(d-1)}=
\dim\p{\bigoplus_{i=-j}^j W_i\otimes \rho_1\otimes\rho_2\otimes\rho_3}^{\SO(d-1)}.	
\end{multline}
According to~(\ref{eq:cfrule}), this is precisely the total number of covariant three-point structures for $\<\cO_1\cO_2\cO_3e^A\>$. In other words, generically, every conformally-covariant three-point structure can be obtained by acting with differential operators on conformally-invariant three-point structures.

Note that according to the discussion in section~\ref{sec:finitedim} we can interpret the conformally-covariant three-point functions as conformally-invariant four-point functions involving a degenerate primary $w^a(x)$. Analogously, we can interpret~\eqref{eq:covariantthreeptdiagrams} as conformal blocks for these four-point functions. We have just proven a highly degenerate case of the folklore theorem which states that that the number of such conformal blocks is equal the dimension of the space of degenerate four-point functions.\footnote{In the non-degenerate case we have the number of families of conformal blocks and the number of ``functional degrees of freedom.''} Importantly, in our case this number is finite. This brings us to a very powerful observation.

\subsection{Crossing and $6j$ symbols}

The diagrams (\ref{eq:covariantthreeptdiagrams}) give a basis for the finite-dimensional space of covariant three-point structures $\<\cO_1\cO_2\cO_3e^A\>$. However, this is not the only interesting basis. The distinguishing feature of  (\ref{eq:covariantthreeptdiagrams}) is that it selects a particular operator $\cO_3'$ appearing in the $\cO_1\x\cO_2$ OPE\@. In other words, it diagonalizes the action of the Casimir $(L_1+L_2)^2$ acting simultaneously on $\cO_1,\cO_2$ (equivalently $\cO_3,w$). However, we may wish to select an operator in a different channel, e.g.\ $\cO_1'\in\cO_2\x\cO_3$. This would correspond to starting with a three-point structure $\<\cO_1'\cO_2\cO_3\>^{(m)}$ and acting with a differential operator $\cD^{(n)A}:\cO_1'\to \cO_1$.

These two bases are related by a linear transformation, which gives a type of crossing equation for differential operators,
\be
\diagramEnvelope{\begin{tikzpicture}[anchor=base,baseline]
	\node (vertL) at (0,0) [threept] {$a$};
	\node (vertR) at (2,-0.05) [threept] {$b$};
	\node (opO1) at (-0.5,-1) [below] {$\cO_1$};
	\node (opO2) at (-0.5,1) [above] {$\cO_2$};
	\node (opO3) at (2.5,1) [above] {$\cO_3$};
	\node (opW) at (2.5,-1) [below] {$W$};	
	\node at (1,0.1) [above] {$\cO_3'$};	
	\draw [spinning] (vertL)-- (opO1);
	\draw [spinning] (vertL)-- (opO2);
	\draw [spinning] (vertL)-- (vertR);
	\draw [spinning] (vertR)-- (opO3);
	\draw [finite with arrow] (vertR)-- (opW);
\end{tikzpicture}}
	\quad=\quad
	\sum_{\cO_1',m,n}
	\left\{
		\begin{matrix}
		\cO_1 & \cO_2 & \cO_1' \\
		\cO_3 & W & \cO_3'
		\end{matrix}
	\right\}^{ab}_{mn}
\diagramEnvelope{\begin{tikzpicture}[anchor=base,baseline]
	\node (vertU) at (0,0.7) [threept] {$m$};
	\node (vertD) at (0,-0.7) [threept] {$n$};
	\node (opO1) at (-1,-1.5) [below] {$\cO_1$};
	\node (opO2) at (-1,1.5) [above] {$\cO_2$};
	\node (opO3) at (1,1.5) [above] {$\cO_3$};
	\node (opW) at (1,-1.5) [below] {$W$};	
	\node at (0.1,0) [right] {$\cO_1'$};	
	\draw [spinning] (vertD)-- (opO1);
	\draw [spinning] (vertU)-- (opO2);
	\draw [spinning] (vertU)-- (vertD);
	\draw [spinning] (vertU)-- (opO3);
	\draw [finite with arrow] (vertD)-- (opW);
\end{tikzpicture}}.
\label{eq:6jdefinition}
\ee
In equations, (\ref{eq:6jdefinition}) reads 
\be\label{eq:definition_6j_symbols_algrbraic}
\cD^{(b)A}_{x_3}\<\cO_1(x_1)\cO_2(x_2)\cO_3'(x_3)\>^{(a)} &= \sum_{\cO_1',m,n}
	\left\{
		\begin{matrix}
		\cO_1 & \cO_2 & \cO_1' \\
		\cO_3 & W & \cO_3'
		\end{matrix}
	\right\}^{ab}_{mn} 
\cD^{(n)A}_{x_1}\<\cO_1'(x_1)\cO_2(x_2)\cO_3(x_3)\>^{(m)}.
\ee
Note that the sum over $\cO_1'$ is finite with $\cO_1'$ taking values in the tensor product $\cO_1\otimes W$. The coefficients in this transformation are called Racah coefficients, or $6j$ symbols.\footnote{Technically, Racah coefficients and $6j$ symbols are sometimes defined to differ by various normalization factors. We will not distinguish between them and use both terms to refer to the coefficients in (\ref{eq:6jdefinition}).}${}^,$\footnote{$6j$ symbols depend only on a set of representations and three-point structures. However, for brevity, we often label them with operators $\cO_i$ transforming in those representations, as in (\ref{eq:6jdefinition}).} The $6j$ symbols for operator representations (generalized Verma modules) of the conformal group have seen some recent interest for their role in the crossing equations for CFT four-point functions \cite{Gadde:2017sjg,Hogervorst:2017sfd,Hogervorst:2017kbj}. Here, we have a degenerate form of these objects, where one of the representations appearing is finite-dimensional. These degenerate $6j$ symbols enter in a degenerate crossing equation  (\ref{eq:6jdefinition}) where the objects on both sides live in a finite dimensional space. One can ask what happens if we consider $6j$ symbols with more finite-dimensional representations. As we show in appendix~\ref{app:sixjalgebra}, such $6j$ symbols are related to the algebra of  conformally-covariant differential operators.

An analogy for understanding (\ref{eq:6jdefinition}) is to consider a four-point function containing at least one degenerate Virasoro primary in a 2d CFT. The shortening condition on the degenerate primary implies that its four-point function lives in a finite-dimensional space spanned by a finite number of conformal blocks. The crossing transformation for these blocks is a finite-dimensional matrix. Similarly in (\ref{eq:6jdefinition}), the left-hand side can be interpreted as the conformal block for $\cO_3'$ exchange in a four-point function $\<\cO_1\cO_2\cO_3 w\>$. Because $w$ satisfies a highly-constraining differential equation, the crossing transformation for this block is a finite-dimensional matrix.

\subsection{Examples}

Because the space of covariant three-point structures is finite dimensional (its dimension is given by (\ref{eq:dimofcovariantthreept})), it is straightforward to find the degenerate $6j$ symbols by direct computation: we apply differential operators on both sides and invert a finite-dimensional matrix. Let us work through some examples.

\subsubsection{$6j$ symbols in 1 dimension}
\label{sec:sixJ1d}
\paragraph{3-point functions}
Before computing the $6j$ symbols, we need to choose a basis of three-point structures. The three-point functions in 1-dimension are not completely trivial, and it is important to get them right in order to have well-defined $6j$ symbols. 

According to the discussion of section~\ref{sec:embedding1d}, there are two types of fields with different ``spins'' $s=\pm$. The fields with $s=+$ are the usual scalars. The simplest three-point function for the scalars is
\be
	\<\Phi_1^+(\chi_1)\Phi_2^+(\chi_2)\Phi_3^+(\chi_3)\>^{(+)}=\frac{1}{|\chi_1\chi_2|^{\Delta_1+\Delta_2-\Delta_3}|\chi_2\chi_3|^{\Delta_2+\Delta_3-\Delta_1}|\chi_3\chi_1|^{\Delta_3+\Delta_1-\Delta_2}}.
\ee
Here we have added the label ${(+)}$ to indicate that this is a parity-even three-point structure. We need this because there in fact exists a parity-odd three-point structure,
\be
\<\Phi_1^+(\chi_1)\Phi_2^+(\chi_2)\Phi_3^+(\chi_3)\>^{(-)}=\frac{(\chi_1\chi_2)(\chi_2\chi_3)(\chi_3\chi_1)}{|\chi_1\chi_2|^{\Delta_1+\Delta_2-\Delta_3+1}|\chi_2\chi_3|^{\Delta_2+\Delta_3-\Delta_1+1}|\chi_3\chi_1|^{\Delta_3+\Delta_1-\Delta_2+1}}.
\ee
This is related to the fact that unless we allow reflections, all conformal transformations preserve the cyclic ordering of three points on the circle $S^1$. One can see that this structure is parity-odd from the parity transformation $\chi\to\gamma^2\chi$. 

We will compute the $6j$ symbols for differential operators in the fundamental representation which, according to~\eqref{eq:diffmap1d}, change the spin $s$. Therefore, we will also need the parity even and parity odd structures for the three point function with two $s=-$ operators,
\begin{align}
\<\Phi_1^-(\chi_1)\Phi_2^+(\chi_2)\Phi_3^-(\chi_3)\>^{(-)}&=\frac{(\chi_3\chi_1)}{|\chi_1\chi_2|^{\Delta_1+\Delta_2-\Delta_3}|\chi_2\chi_3|^{\Delta_2+\Delta_3-\Delta_1}|\chi_3\chi_1|^{\Delta_3+\Delta_1-\Delta_2+1}},\\
\<\Phi_1^-(\chi_1)\Phi_2^+(\chi_2)\Phi_3^-(\chi_3)\>^{(+)}&=\frac{(\chi_1\chi_2)(\chi_2\chi_3)}{|\chi_1\chi_2|^{\Delta_1+\Delta_2-\Delta_3+1}|\chi_2\chi_3|^{\Delta_2+\Delta_3-\Delta_1+1}|\chi_3\chi_1|^{\Delta_3+\Delta_1-\Delta_2}}.
\end{align}
The difference between $s=+$ and $s=-$ tensor structures is in their transformation properties under~\eqref{eq:onedimcenterproperty}.

\paragraph{$6j$ symbols}
As noted above, we will specialize to $W=F$ being the fundamental representation of $\SL(2,\R)$, which has weights $\De=\pm \tfrac 1 2$. The corresponding differential operators are
\be
\cD^{+}_\a = \ptl_\a,\qquad \cD^{-}_\a = \chi_\a.
\ee
It will be convenient to contract each differential with a polarization spinor $\chi_4$, giving $\chi_4^\a \cD^\pm_\a$. This spinor may be interpreted as the coordinate of the fourth operator in representation~$[-\tfrac{1}{2},-]$. The operator $\chi_4\cD^+$ is even under space parity, while the operator $\chi_4\cD^-$ is odd under space parity.

The definition of $6j$ symbols in this case is
\be\label{eq:sixJ1ddefn}
\diagramEnvelope{\begin{tikzpicture}[anchor=base,baseline]
	\node (vertL) at (0,-.08) [threept] {$a$};
	\node (vertR) at (2.2,0) [threept] {};
	\node (opO1) at (-0.5,-1) [below] {$[\Delta_1,s_1]$};
	\node (opO2) at (-0.5,1) [above] {$[\Delta_2,s_2]$};
	\node (opO3) at (2.7,1) [above] {$[\Delta_3,s_3]$};
	\node (opW) at (2.7,-1) [below] {$F$};	
	\node at (1.1,0.1) [above] {$[\Delta_3\pm\tfrac{1}{2},-s_3]$};	
	\draw [spinning] (vertL)-- (opO1);
	\draw [spinning] (vertL)-- (opO2);
	\draw [spinning] (vertL)-- (vertR);
	\draw [spinning] (vertR)-- (opO3);
	\draw [finite with arrow] (vertR)-- (opW);
\end{tikzpicture}}
	=
	\sum_{\Delta= \Delta_1\pm\frac{1}{2}\atop m}
	\left\{
		\begin{matrix}
		[\Delta_1,s_1] & [\Delta_2,s_2] & [\Delta,-s_1] \\
		[\Delta_3,s_3] & F & [\Delta_3\pm\tfrac{1}{2},-s_3]
		\end{matrix}
	\right\}^{a\uniq}_{m\uniq}
\diagramEnvelope{\begin{tikzpicture}[anchor=base,baseline]
	\node (vertU) at (0,0.7) [threept] {$m$};
	\node (vertD) at (0,-0.7) [threept] {};
	\node (opO1) at (-1,-1.5) [below] {$[\Delta_1,s_1]$};
	\node (opO2) at (-1,1.5) [above] {$[\Delta_2,s_2]$};
	\node (opO3) at (1,1.5) [above] {$[\Delta_3,s_3]$};
	\node (opW) at (1,-1.5) [below] {$F$};	
	\node at (0.1,0) [right] {$[\Delta,-s_1]$};	
	\draw [spinning] (vertD)-- (opO1);
	\draw [spinning] (vertU)-- (opO2);
	\draw [spinning] (vertU)-- (vertD);
	\draw [spinning] (vertU)-- (opO3);
	\draw [finite with arrow] (vertD)-- (opW);
\end{tikzpicture}}.
\ee
We don't need to label the vertices for differential operators, since there is always a unique choice of differential operator for the given dimensions. For example, on the left-hand side, when the internal line has dimension $\De_3\pm\frac 1 2$, the $F$-differential operator must be $\cD^\mp$. The notation ``$\uniq$" on the $6j$ symbols means there is a unique corresponding structure or differential operator.

It is now straightforward to compute the objects above. Let us take for example $s_1=s_2=+$, $s_3=-$ and specialize to the case when both sides of~\eqref{eq:sixJ1ddefn} are parity-odd. For the left-hand side we then have,

\be
\diagramEnvelope{\begin{tikzpicture}[anchor=base,baseline]
	\node (vertL) at (0,-.1) [threept,inner sep=0pt] {$+$};
	\node (vertR) at (2,0) [threept] {};
	\node (opO1) at (-0.5,-1) [below] {$[\Delta_1,+]$};
	\node (opO2) at (-0.5,1) [above] {$[\Delta_2,+]$};
	\node (opO3) at (2.5,1) [above] {$[\Delta_3,-]$};
	\node (opW) at (2.5,-1) [below] {$F$};	
	\node at (1,0.1) [above] {$[\Delta_3+\frac{1}{2},+]$};	
	\draw [spinning] (vertL)-- (opO1);
	\draw [spinning] (vertL)-- (opO2);
	\draw [spinning] (vertL)-- (vertR);
	\draw [spinning] (vertR)-- (opO3);
	\draw [finite with arrow] (vertR)-- (opW);
\end{tikzpicture}}
&=
\frac{(\chi_4\chi_3)}{|\chi_1\chi_2|^{\Delta_1+\Delta_2-\Delta_3-1/2}|\chi_2\chi_3|^{\Delta_2+\Delta_3-\Delta_1+1/2}|\chi_3\chi_1|^{\Delta_3+\Delta_1-\Delta_2+1/2}},\nn\\
\diagramEnvelope{\begin{tikzpicture}[anchor=base,baseline]
	\node (vertL) at (0,-.09) [threept,inner sep=0pt] {$-$};
	\node (vertR) at (2,0) [threept] {};
	\node (opO1) at (-0.5,-1) [below] {$[\Delta_1,+]$};
	\node (opO2) at (-0.5,1) [above] {$[\Delta_2,+]$};
	\node (opO3) at (2.5,1) [above] {$[\Delta_3,-]$};
	\node (opW) at (2.5,-1) [below] {$F$};	
	\node at (1,0.1) [above] {$[\Delta_3-\frac{1}{2},+]$};	
	\draw [spinning] (vertL)-- (opO1);
	\draw [spinning] (vertL)-- (opO2);
	\draw [spinning] (vertL)-- (vertR);
	\draw [spinning] (vertR)-- (opO3);
	\draw [finite with arrow] (vertR)-- (opW);
\end{tikzpicture}}
&=
\frac{-(\Delta_1+\Delta_3-\Delta_2-1/2)(\chi_4\chi_1)(\chi_1\chi_2)(\chi_2\chi_3)}{|\chi_1\chi_2|^{\Delta_1+\Delta_2-\Delta_3+3/2}|\chi_2\chi_3|^{\Delta_2+\Delta_3-\Delta_1+1/2}|\chi_3\chi_1|^{\Delta_3+\Delta_1-\Delta_2+1/2}}\nn\\
&\ \ +
\frac{(\Delta_2+\Delta_3-\Delta_1-1/2)(\chi_4\chi_2)(\chi_1\chi_2)(\chi_3\chi_1)}{|\chi_1\chi_2|^{\Delta_1+\Delta_2-\Delta_3+3/2}|\chi_2\chi_3|^{\Delta_2+\Delta_3-\Delta_1+1/2}|\chi_3\chi_1|^{\Delta_3+\Delta_1-\Delta_2+1/2}}.
\ee
For the right-hand side,
\be
\diagramEnvelope{\begin{tikzpicture}[anchor=base,baseline]
	\node (vertU) at (0,0.7) [threept,inner sep=0pt] {$+$};
	\node (vertD) at (0,-0.7) [threept] {};
	\node (opO1) at (-1,-1.5) [below] {$[\Delta_1,+]$};
	\node (opO2) at (-1,1.5) [above] {$[\Delta_2,+]$};
	\node (opO3) at (1,1.5) [above] {$[\Delta_3,-]$};
	\node (opW) at (1,-1.5) [below] {$F$};	
	\node at (0.1,0) [right] {$[\Delta_1+\frac{1}{2},-]$};	
	\draw [spinning] (vertD)-- (opO1);
	\draw [spinning] (vertU)-- (opO2);
	\draw [spinning] (vertU)-- (vertD);
	\draw [spinning] (vertU)-- (opO3);
	\draw [finite with arrow] (vertD)-- (opW);
\end{tikzpicture}}
&=\frac{(\chi_4\chi_1)(\chi_1\chi_2)(\chi_2\chi_3)}{|\chi_1\chi_2|^{\Delta_1+\Delta_2-\Delta_3+3/2}|\chi_2\chi_3|^{\Delta_2+\Delta_3-\Delta_1+1/2}|\chi_3\chi_1|^{\Delta_3+\Delta_1-\Delta_2+1/2}}\nn\\
\diagramEnvelope{\begin{tikzpicture}[anchor=base,baseline]
	\node (vertU) at (0,0.7) [threept,inner sep=0pt] {$-$};
	\node (vertD) at (0,-0.7) [threept] {};
	\node (opO1) at (-1,-1.5) [below] {$[\Delta_1,+]$};
	\node (opO2) at (-1,1.5) [above] {$[\Delta_2,+]$};
	\node (opO3) at (1,1.5) [above] {$[\Delta_3,-]$};
	\node (opW) at (1,-1.5) [below] {$F$};	
	\node at (0.1,0) [right] {$[\Delta_1-\frac{1}{2},-]$};	
	\draw [spinning] (vertD)-- (opO1);
	\draw [spinning] (vertU)-- (opO2);
	\draw [spinning] (vertU)-- (vertD);
	\draw [spinning] (vertU)-- (opO3);
	\draw [finite with arrow] (vertD)-- (opW);
\end{tikzpicture}}
&=
\frac{(\Delta_1+\Delta_3-\Delta_2-1/2)(\chi_4\chi_3)}{|\chi_1\chi_2|^{\Delta_1+\Delta_2-\Delta_3-1/2}|\chi_2\chi_3|^{\Delta_2+\Delta_3-\Delta_1+1/2}|\chi_3\chi_1|^{\Delta_3+\Delta_1-\Delta_2+1/2}}
\nn\\
&\ \ 
+
\frac{-(\Delta_1+\Delta_2-\Delta_3-1/2)(\chi_4\chi_2)(\chi_1\chi_2)(\chi_3\chi_1)}{|\chi_1\chi_2|^{\Delta_1+\Delta_2-\Delta_3+3/2}|\chi_2\chi_3|^{\Delta_2+\Delta_3-\Delta_1+1/2}|\chi_3\chi_1|^{\Delta_3+\Delta_1-\Delta_2+1/2}}.
\ee

After using the Schouten identity
\be
	(\chi_4\chi_1)(\chi_2\chi_3)+(\chi_4\chi_2)(\chi_3\chi_1)+(\chi_4\chi_3)(\chi_1\chi_2)=0,
\ee
we can solve for the $6j$ symbols
\be
	\left\{
		\begin{matrix}
		[\Delta_1,+] & [\Delta_2,+] & [\Delta_1+\tfrac{1}{2},-] \\
		[\Delta_3,-] & F & [\Delta_3+\tfrac{1}{2},+]
		\end{matrix}
	\right\}^{+\uniq}_{+\uniq} &= -\frac{\Delta _1+\Delta_2-\Delta _3-1/2}{2 \Delta _1-1},\\
	\left\{
		\begin{matrix}
		[\Delta_1,+] & [\Delta_2,+] & [\Delta_1-\tfrac{1}{2},-] \\
		[\Delta_3,-] & F & [\Delta_3+\tfrac{1}{2},+]
		\end{matrix}
	\right\}^{+\uniq}_{-\uniq} &= \frac{1}{2 \Delta _1-1},\\
	\left\{
		\begin{matrix}
		[\Delta_1,+] & [\Delta_2,+] & [\Delta_1+\tfrac{1}{2},-] \\
		[\Delta_3,-] & F & [\Delta_3-\tfrac{1}{2},+]
		\end{matrix}
	\right\}^{-\uniq}_{+\uniq} &= -\frac{\left(\Delta _1+\Delta _3-\Delta _2-1/2\right) \left(\Delta _1+ \Delta _2+ \Delta _3-3/2\right)}{2 \Delta _1-1},\\
	\left\{
		\begin{matrix}
		[\Delta_1,+] & [\Delta_2,+] & [\Delta_1-\tfrac{1}{2},-] \\
		[\Delta_3,-] & F & [\Delta_3-\tfrac{1}{2},+]
		\end{matrix}
	\right\}^{-\uniq}_{-\uniq} &= -\frac{\Delta _2+\Delta _3-\Delta _1-1/2}{2 \Delta _1-1}.
\ee

\subsubsection{$6j$ symbols in 3 dimensions}
\paragraph{3-point functions}
\label{sec:3D6j_threepoint}
It is also possible to find the general $6j$ symbols for the spinor representation $\mathcal{S}$ of the 3d conformal group. To do that, it is convenient to use the conformal frame basis of three-point structures from~\cite{Kravchuk:2016qvl}.\footnote{Our conventions in this section are those of~\cite{Iliesiu:2015qra,Iliesiu:2015akf,Kravchuk:2016qvl}.} To construct this basis, one contracts the 3d primary operators with polarization spinors $s_\alpha$,
\be
	\cO(s,x)=s_{\alpha_1}\cdots s_{\alpha_{2\ell}}\cO^{\alpha_1\ldots \alpha_{2\ell}}(x).
\ee
The three point-functions are then evaluated in the configuration
\be
	f_3(s_1,s_2,s_3)=\<\cO_1(s_1,0)\cO_2(s_2,e)\cO_3(s_3,\infty)\>,
\ee
where $e=(0,0,1)$ and $\cO(s_3,\infty)=\lim_{L\to\infty}L^{2\Delta_3}\cO(s_3,Le)$. The polynomial $f_3$ should be invariant under boosts in the 0-1 plane. A basis for such polynomials is given by the monomials
\be\label{eq:3DmonomialBasis}
	[q_1q_2q_3] &=\prod_{i=1}^{3}\xi_i^{\ell_i+q_i}\bar\xi_i^{\ell-q_i},
\ee
where $s_i=(\xi_i,\bar \xi_i)$ and $q_i=-\ell_i\ldots \ell_i$, subject to the constraint $\sum_i q_i=0$.

It will also be convenient to think about the covariant three-point functions as four-point functions with the degenerate spinor primary $w^\alpha(x)$ of dimension $-\frac{1}{2}$. We construct an analogous basis for four-point tensor structures by evaluating
\be\label{eq:fourpointboundary}
	\<\cO_1(s_1,0)s_\alpha w^\alpha(ze) \cO_2(s_2,e)\cO_3(s_3,\infty)\>,
\ee
leading to a monomial basis $[q_1,q,q_2,q_3]$, where $q=\pm\frac{1}{2}$.\footnote{Notice that we used a configuration different from the one used for four-point functions in~\cite{Kravchuk:2016qvl}.} The configuration \eqref{eq:fourpointboundary} is still invariant under boosts in the 0-1 plane, so we again have the condition $q+\sum q_i=0$. We have introduced only one cross-ratio $z$ because $w^\alpha(x)$ is a degenerate field. In fact, the general solution to its Killing equation is given by
\be
	w(x)=w_0+x^\mu\gamma_\mu w_1,
\ee
and thus it is sufficient to know its values for $x=ze$ to determine it completely. Note also that this equation implies that a general four-point function of such form is linear in $z$.

To obtain these degenerate four-point functions, we think about the three-point functions as four-point functions with an identity operator at coordinate $x$ and act with the operators\footnote{Note that the components of $\cD^{-+}_a$ are essentially the conformal Killing spinors $s_\alpha w^\alpha_a(x)$.}
\be
	D^{\pm\pm}_i=\Omega^{ab}\cD^{-+}_{a,x}\cD^{\pm\pm}_{b,x_i}\Sigma^{\mp\mp}_i,
\ee
where $x=ze$, and $\Sigma^{\mp\mp}_i$ formally shifts the scaling dimension and spin of the operator $i$, so that $D^{\pm\pm}_i$ doesn't change the dimensions and spins.\footnote{This notation has the somewhat unfortunate consequence that $D_i^{\pm\pm}[q_1q_2q_3]$ is a four-point function with external spins $\ell_1,\ell_2,\ell_3$ but the structure labels $[q_1q_2q_3]$ relate to a shifted set of spins, see the diagrams below.} In this notation we have\footnote{As in the 1d case, we omit the labels for the differential operators in the diagrams (\ref{eq:sixjthreeddiagrams}) because the differential operator is always fixed by the given representations.}
\be
	D^{\pm\pm}_3[q_1q_2q_3]&\equiv
	\diagramEnvelope{\begin{tikzpicture}[anchor=base,baseline]
		\node (vertL) at (0,-.08) [threept,inner sep=1pt] {$q_i$};
		\node (vertR) at (3,-.0335) [threept] {};
		\node (opO1) at (-0.5,-1) [below] {$[\Delta_1,\ell_1]$};
		\node (opO2) at (-0.5,1) [above] {$[\Delta_2,\ell_2]$};
		\node (opO3) at (3.5,1) [above] {$[\Delta_3,\ell_3]$};
		\node (opW) at (3.5,-1) [below] {$s_\alpha w^\alpha$};	
		\node at (1.5,0.1) [above] {$[\Delta_3\mp\tfrac{1}{2},\ell_3\mp \tfrac{1}{2}]$};	
		\draw [spinning] (vertL)-- (opO1);
		\draw [spinning] (vertL)-- (opO2);
		\draw [spinning] (vertL)-- (vertR);
		\draw [spinning] (vertR)-- (opO3);
		\draw [finite with arrow] (vertR)-- (opW);
	\end{tikzpicture}},\nn\\
	D^{\pm\pm}_1[q_1q_2q_3]&\equiv
	\diagramEnvelope{\begin{tikzpicture}[anchor=base,baseline]
		\node (vertU) at (0,0.7) [threept,inner sep=1pt] {$q_i$};
		\node (vertD) at (0,-0.7) [threept] {};
		\node (opO1) at (-1,-1.5) [below] {$[\Delta_1,\ell_1]$};
		\node (opO2) at (-1,1.5) [above] {$[\Delta_2,\ell_2]$};
		\node (opO3) at (1,1.5) [above] {$[\Delta_3,\ell_3]$};
		\node (opW) at (1,-1.5) [below] {$s_\alpha w^\alpha$};	
		\node at (0.1,0) [right] {$[\Delta_1\mp\tfrac{1}{2},\ell_1\mp\tfrac{1}{2}]$};	
		\draw [spinning] (vertD)-- (opO1);
		\draw [spinning] (vertU)-- (opO2);
		\draw [spinning] (vertU)-- (vertD);
		\draw [spinning] (vertU)-- (opO3);
		\draw [finite with arrow] (vertD)-- (opW);
	\end{tikzpicture}}.
	\label{eq:sixjthreeddiagrams}
\ee
Our goal is therefore to find the transformation between the bases $D_3^{\pm\pm}[q_1q_2q_3]$ and $D_1^{\pm\pm}[q_1q_2q_3]$

\paragraph{$6j$ symbols}

It is obvious that since the operators $D_i^{\pm\pm}$ contain a finite number of derivatives in the polarization spinors, they take a three-point structure $[q_1q_2q_3]$ to four-point structures $[q'_1,q,q'_2,q'_3]$ for~\eqref{eq:fourpointboundary} with $q'_i$ differing from $q_i$ by only finite shifts. We can say that $D_i^{\pm\pm}$ are local in $q$-space. It turns out that the inverse operation, which expresses an arbitrary four-point function~\eqref{eq:fourpointboundary} in terms of $D_i^{\pm\pm}[q_1q_2q_3]$, is also local in $q$-space. In this language the $6j$ symbols essentially give the composition of the inverse to $D_1^{\pm\pm}$ with $D_3^{\pm\pm}$ and are thus also local in $q$-space. This allows us to write down a general expression for these $6j$ symbols.

The number of shifts in $q$ for which the $6j$ symbols are generically non-zero is however rather large. We therefore take an indirect approach in this section, describing how the $6j$ symbols can be straightforwardly generated from relatively simple expressions. Our strategy will be to write the action of $D_1^{\pm\pm}$ and $D_3^{\pm\pm}$ on $[q_1q_2q_3]$ in a form from which both the direct action and the inverse can be easily obtained. One can then simply substitute the inverse of $D_1^{\pm\pm}$ into the expressions for $D_3^{\pm\pm}[q_1q_2q_3]$ to generate the general $6j$ symbols.

First, we evaluate the expressions for $D_3^{\pm\pm}[q_1q_2q_3]$ and $D_1^{\pm\pm}[q_1q_2q_3]$. This can be done relatively easily in a computer algebra system. The result can be expressed in terms of the four-point tensor structures $[q_1,q,q_2,q_3]$, for instance,
\be
	D_1^{--}[q_1q_2q_3]=z(\ell_1+q_1+\tfrac{1}{2})[q_1-\tfrac{1}{2},+\tfrac{1}{2},q_2,q_3]-
		z(\ell_1-q_1+\tfrac{1}{2})[q_1+\tfrac{1}{2},-\tfrac{1}{2},q_2,q_3].
\ee
We will now describe these actions in a compact form. We first define
\be\label{eq:Adefinition}
	\cA^\pm_1[q_1q_2q_3]=\p{-D_1^{--}\mp(\ell_1\mp q_1+\tfrac{1}{2})D_1^{-+}}[q_1q_2q_3],
\ee
These operators satisfy
\be\label{eq:Aaction}
	\cA^\pm_1[q_1q_2q_3]=\mp z(2\ell_1+1)[q_1\mp\tfrac{1}{2},\pm\tfrac{1}{2},q_2,q_3].
\ee
Note that this solves the inversion problem for the linear terms $z[q_1,q,q_2,q_3]$ and is also sufficient to find the action $D_1^{-\pm}[q_1q_2q_3]$. We then define the analogous operators
\be\label{eq:Bdefinition}
	\cB^\pm_1[q_1q_2q_3]=\p{-D_1^{+-}\mp(\ell_1\mp q_1+\tfrac{1}{2})D_1^{++}}[q_1q_2q_3]+\cC^\pm_1[q_1q_2q_3],
\ee
where the correction term $\cC^\pm_1$ is a linear combination of $\cA^\pm_1$ given below. The operators $\cB^\pm_1$ act on $[q_1q_2q_3]$ as follows,
\begin{multline}\label{eq:Caction}
\p{(\De_1\pm q_1-\tfrac{3}{2})\cB^\pm_1+(\ell_1\mp q_1+\tfrac{1}{2})\cB^\mp_1}[q_1q_2q_3]=\\ =4(2\ell_1+1)(\De_1-\tfrac{3}{2})(\ell_1+\De_1-1)(\ell_1-\De_1+2)[q_1\mp \tfrac{1}{2},\pm\tfrac{1}{2},q_2,q_3].
\end{multline}
This solves the inversion problem for the constant terms $[q_1,q,q_2,q_3]$ and is also sufficient to write down the action of $\cB_1^\pm$ and thus also of $D^{+\pm}_1$.

We can describe the action of $D_3^{\pm\pm}$ and its inverse in a similar fashion. In particular, we define
\be
\cA^\pm_3[q_1q_2q_3]&=\p{-D_3^{--}\mp(\ell_3\mp q_3+\tfrac{1}{2})D_3^{-+}}[q_1q_2q_3],\\\label{eq:B3defn}
\cB^\pm_3[q_1q_2q_3]&=\p{-D_3^{+-}\mp(\ell_3\mp q_3+\tfrac{1}{2})D_3^{++}}[q_1q_2q_3]-\cC^\pm_3[q_1q_2q_3].
\ee
The correction term $\cC^\pm_3$ is defined below. For these operators we have the analogue of~\eqref{eq:Aaction}
\be
\cA^\pm_3[q_1q_2q_3]=\pm (2\ell_3+1)[q_1,\pm\tfrac{1}{2},q_2,q_3\mp\tfrac{1}{2}],
\ee
and the analogue of~\eqref{eq:Caction},
\be
	\cB_3^\pm[q_1q_2q_3]=&- 4z(2\ell_3+1)(\De_3-\tfrac{3}{2})(\De_3\mp q_3-\tfrac{3}{2})[q_1,\pm\tfrac{1}{2},q_2,q_3\mp \tfrac{1}{2}]\nn\\
	&+4z(2\ell_3+1)(\De_3-\tfrac{3}{2})(\ell_3\mp q_3+\tfrac{1}{2})[q_1,\mp\tfrac{1}{2},q_2,q_3\pm \tfrac{1}{2}].
\ee
We can use these expressions to find the action of $D_3^{\pm\pm}$ and then substitute the expressions~\eqref{eq:Aaction} and \eqref{eq:Caction} for the four-point functions $z[q_1,q,q_2,q_3]$ and $[q_1,q,q_2,q_3]$ in terms of $D_1^{\pm\pm}$ to find the $6j$ symbols.
As a simple example, we find for $(\ell_1,\ell_2,\ell_3)=(\tfrac 1 2, 0, 0)$,
\be
D_3^{--}[\tfrac{1}{2},0,-\tfrac{1}{2}]&=\frac{1}{2 (2 \Delta _1-3) (2 \Delta _1-1)}D_1^{++}[000]+\frac{1}{2 (2 \Delta _1-5) (2 \Delta _1-3)}D_1^{+-}[000]
\nn\\
&\quad +\frac{2 \Delta _1+2 \Delta _2-2 \Delta _3-5 }{4 (2 \Delta _1-5)}D_1^{-+}[000]-\frac{2 \Delta _1+2 \Delta _2-2 \Delta _3-1}{4 (2 \Delta _1-1)} D_1^{--}[000],
\ee
from which we can read off the for example the following $6j$ symbol,
\be
\left\{\begin{matrix}
[\Delta _1,\frac{1}{2}] & [\Delta _2,0] & [\Delta _1+\frac{1}{2},1] \\
[\Delta _3,0]           & \Psi          & [\Delta _3+\frac{1}{2},\frac{1}{2}]
\end{matrix}\right\}^{[\frac{1}{2},0,-\frac{1}{2}](--)}_{[000](--)}
&=
-\frac{2 \Delta _1+2 \Delta _2-2 \Delta _3-1}{4 (2 \Delta _1-1)}.
\ee

The correction term $\cC_1^\pm$ is given by
\be\label{eq:Cdefinition}
	\cC^\pm_1[q_1q_2q_3]=&{(\ell_1+q_1\mp \tfrac{1}{2})(\ell_3-q_3)}\cA_1^\pm[q_1-1,q_2,q_3+1]\nn\\
	&-(\ell_1-q_1\pm \tfrac{1}{2})(\ell_3+q_3)\cA_1^\pm[q_1+1,q_2,q_3-1]\nn\\
	&-(\ell_1+q_1\mp \tfrac{1}{2})(\ell_2-q_2)\cA_1^\pm[q_1-1,q_2+1,q_3]\nn\\
	&+(\ell_1-q_1\pm \tfrac{1}{2})(\ell_2+q_2)\cA_1^\pm[q_1+1,q_2-1,q_3]\nn\\
	&\mp2(\ell_3\mp q_3)(\De_1-2)\cA_1^{\mp}[q_1\mp 1,q_2,q_3\pm 1]\nn\\
	&\pm2(\ell_2\mp q_2)(\De_1-2)\cA_1^{\mp}[q_1\mp 1,q_2\pm 1,q_3]\nn\\
	&\pm2(\De_1\mp q_1- \tfrac{3}{2})(\De_1+\De_2-\De_3-\tfrac{3}{2})\cA_1^{\pm}[q_1,q_2,q_3]\nn\\
	&\pm2(\ell_1\mp q_1+\tfrac{1}{2})(\De_1-\tfrac{3}{2})\cA_1^{\mp}[q_1,q_2,q_3].
\ee
The correction term $\cC_3^\pm$ is obtained from the above expression by replacing $1\leftrightarrow 3$ in the coefficients, replacing $\cA_1$ by $\cA_3$, and exchanging the shifts applied to $q_1$ and $q_3$ in the three-point structures. Note that $\cC_3^\pm$ enters~\eqref{eq:B3defn} with a minus sign.

\subsection{Differential bases from $6j$ symbols}
\label{sec:differentialbases}

The crossing equation (\ref{eq:6jdefinition}) will be our key computational tool in this work. Using it, we can perform a variety of calculations with differential operators. As a brief example, consider contracting both sides of  (\ref{eq:6jdefinition}) with a differential operator $\cD^{(c)}_A:\cO_1\to \cO_1''$, which we denote
\be
\diagramEnvelope{\begin{tikzpicture}[anchor=base,baseline]
\node (tl) at (-1,0.8) [above] {$\cO_1$};
\node (tr) at (1,0.8) [above] {$W$};
\node (vert) at (0,0) [threept] {$c$};
\node (bot) at (0,-1) [below] {$\cO_1''$};
\draw [spinning] (tl) -- (vert);
\draw [spinning] (vert) -- (bot);
\draw [finite with arrow] (tr) -- (vert);
\end{tikzpicture}}.
\label{eq:incomingdiffop}
\ee
Here, the incoming arrow for $W$ indicates that this operator is associated to the dual representation $W^*$. Let us connect the incoming $W$ line in (\ref{eq:incomingdiffop}) with the outgoing $W$ line in (\ref{eq:6jdefinition}), i.e.\ contract the $A$ indices. In equations, we find
\be
& \cD^{(c)}_{A,x_1}\cD^{(b)A}_{x_3}\<\cO_1(x_1)\cO_2(x_2)\cO_3'(x_3)\>^{(a)} \nn\\
&= \sum_{\cO',m,n}
	\left\{
		\begin{matrix}
		\cO_1 & \cO_2 & \cO_1' \\
		\cO_3 & W & \cO_3'
		\end{matrix}
	\right\}^{ab}_{mn} 
\cD^{(c)}_{A,x_1}\cD^{(n)A}_{x_1}\<\cO_1'(x_1)\cO_2(x_2)\cO_3(x_3)\>^{(m)},
\label{eq:actinequations}
\ee
where we have given the differential operators subscripts $x_i$ to indicate which leg they act on.

The composition of differential operators $\cD^{(c)}_{A,x_1}\cD^{(n)A}_{x_1}$ on a single leg corresponds to a bubble diagram
\be
\cD^{(c)}_{A}\cD^{(n)A}\ =\ 
\diagramEnvelope{\begin{tikzpicture}[anchor=base,baseline]
\node (top) at (0,1.8) [above] {$\cO_1'$};
\node (topmid) at (0,0.9) [threept] {$n$};
\node (botmid) at (0,-0.9) [threept] {$c$};
\node (bot) at (0,-1.8) [below] {$\cO_1''$};
\node () at (0.7,0) [right] {$W$};
\node () at (-0.7,0) [left] {$\cO_1$};
\draw [spinning] (top) -- (topmid);
\draw [spinning] (botmid) -- (bot);
\draw [finite with arrow] (topmid) to[out=-20,in=20] (botmid);
\draw [spinning] (topmid) to[out=-160,in=160] (botmid);
\end{tikzpicture}}
&=
\begin{pmatrix}
\cO_1'\\
\cO_1\ W
\end{pmatrix}^{cn} \de_{\cO_1'\cO_1''}.
\label{eq:bubble}
\ee
This vanishes unless the representations for $\cO_1'$ and $\cO_1''$ are the same, in which case it is proportional to the identity (at least for generic scaling dimensions $\De_1',\De'$). The reason is that (\ref{eq:bubble}) represents a homomorphism between generalized Verma modules, which are irreducible when the scaling dimensions are generic.
The constant of proportionality, given by the symbol in parentheses on the right-hand side of (\ref{eq:bubble}), is actually related to another type of $6j$ symbol, as we explain in appendix~\ref{app:sixjalgebra}. For now, we take (\ref{eq:bubble}) as a definition of these symbols.

Using (\ref{eq:bubble}) with $\cO_1'=\cO_1''$, we can simplify the right-hand side of (\ref{eq:actinequations}) to obtain
\be
\diagramEnvelope{\begin{tikzpicture}[anchor=base,baseline]
	\node (vertL) at (0.3,0) [threept] {$a$};
	\node (vertR) at (1.7,-0.05) [threept] {$b$};
	\node (opO1) at (0.2,-0.3) [below] {$\cO_1$};
	\node (opO2) at (-0.7,0.8) [above] {$\cO_2$};
	\node (opO3) at (2.7,0.8) [above] {$\cO_3$};
	\node (opW) at (1.8,-0.3) [below] {$W$};	
	\node (vertB) at (1,-1.2) [threept] {$c$};
	\node (bot) at (1,-2) [below] {$\cO_1'$};
	\node at (1,0.2) [above] {$\cO_3'$};	
	\draw [spinning] (vertL)-- (vertB);
	\draw [spinning] (vertL)-- (opO2);
	\draw [spinning] (vertL)-- (vertR);
	\draw [spinning] (vertR)-- (opO3);
	\draw [finite with arrow] (vertR)-- (vertB);
	\draw [spinning] (vertB) -- (bot);
\end{tikzpicture}}
	\quad&=\quad
	\sum_{m,n}
	\left\{
		\begin{matrix}
		\cO_1 & \cO_2 & \cO_1' \\
		\cO_3 & W & \cO_3'
		\end{matrix}
	\right\}^{ab}_{mn}
	\begin{pmatrix}
\cO_1'\\
\cO_1\ W
\end{pmatrix}^{cn}
\diagramEnvelope{\begin{tikzpicture}[anchor=base,baseline]
	\node (vertU) at (0,0) [threept] {$m$};
	\node (vertD) at (0,-1) [below] {$\cO_1'$};
	\node (opO2) at (-1,0.8) [above] {$\cO_2$};
	\node (opO3) at (1,0.8) [above] {$\cO_3$};
	\draw [spinning] (vertU)-- (opO2);
	\draw [spinning] (vertU)-- (vertD);
	\draw [spinning] (vertU)-- (opO3);
\end{tikzpicture}}.
\label{eq:spinningcbtrick}
\ee
The left-hand side of (\ref{eq:spinningcbtrick}) is a conformally-invariant differential operator $\cD^{(c)}_{A,x_1}\cD^{(b)A}_{x_3}$ acting on a three-point structure at two different points. The right-hand side is a sum of structures where the representations at those points have been modified. The existence of such invariant two-point differential operators was a key observation of \cite{Costa:2011dw}. Here, we see that they factorize into a product of covariant differential operators, each acting on a single point. Indeed, it is easy to verify that all ``basic'' differential operators in \cite{Costa:2011dw} are of this form, with $W$ being either the vector or the adjoint representations of the conformal group. Furthermore, from the discussion in section~\ref{sec:algebraofweightshifting} and appendix~\ref{app:sixjalgebra} it follows that arbitrary compositions of the basic differential operators of~\cite{Costa:2011dw} are also of the form~\eqref{eq:spinningcbtrick} with more complicated representations $W$. In this sense,~\eqref{eq:spinningcbtrick} gives a more fundamental point of view on such operators.

The main purpose of the differential operators in \cite{Costa:2011dw} was to raise the spins of the operators they act on. Here, we see that it is also possible to {\it lower\/} spins, an idea that we discuss briefly in section~\ref{sec:furtherapps}.

Another observation of \cite{Costa:2011dw} is that (\ref{eq:spinningcbtrick}) can sometimes be inverted to express a basis of tensor structures in terms of differential operators acting on simpler structures. For example, when one of the operators $\cO_\ell$ is a traceless-symmetric tensor, one can write three-point structures involving $\cO_\ell$ in terms of derivatives of three-pt structures involving scalars. In our notation, this reads
\be
\label{eq:differentialbasistrick}
\diagramEnvelope{\begin{tikzpicture}[anchor=base,baseline]
\node (tl) at (-0.8,0.8) [above] {$\cO_1$};
\node (vert) at (0,0) [threept] {$a$};
\node (bl) at (-0.8,-0.8) [below] {$\cO_2$};
\node (r) at (0.8,0) [right] {$\cO_\ell$};
\draw [spinning] (vert) -- (tl);
\draw [spinning] (vert) -- (bl);
\draw [spinning] (vert) -- (r);
\end{tikzpicture}}
&=
\sum_{W,b,c}
(\dots)
\diagramEnvelope{\begin{tikzpicture}[anchor=base,baseline]
\node (tl) at (-1.6,1.6) [above] {$\cO_1$};
\node (vertr) at (0.6,0) [twopt] {};
\node (verttl) at (-0.7,0.7) [threept] {$b$};
\node (vertbl) at (-0.7,-0.7) [threept] {$c$};
\node (bl) at (-1.6,-1.6) [below] {$\cO_2$};
\node (r) at (1.6,0) [right] {$\cO_\ell$};
\node (w) at (-0.7,0) [left] {$W$};
\node (f1) at (0.3,0.35) [above] {$\phi_1$};
\node (f2) at (0.3,-0.3) [below] {$\phi_2$};
\draw [spinning] (verttl) -- (tl);
\draw [spinning] (vertbl) -- (bl);
\draw [spinning] (vertr) -- (r);
\draw [scalar] (vertr) -- (verttl);
\draw [scalar] (vertr) -- (vertbl);
\draw [finite with arrow] (verttl) -- (vertbl);
\end{tikzpicture}}.
\ee
Here, the dashed lines denote scalar operators $\f_1,\f_2$. Note that the labels $b,c$ determine the dimensions of $\f_1,\f_2$ in terms of $\De_{\cO_1},\De_{\cO_2}$, respectively. Thus, the right-hand side will involve derivatives of scalar structures with dimensions shifted by half-/integers from those of $\cO_{1},\cO_2$. In equations, we write
\be
\<\cO_1\cO_2\cO_\ell\>^{(a)} &= \mathscr{D}^{(a)\cO_1\cO_2}_{\f_1\f_2} \<\f_1\f_2\cO_\ell\>,
\label{eq:diffbasisexplicit}
\ee
where $\mathscr{D}$ is a combination of derivatives $\ptl_{x_1},\ptl_{x_2}$ and formal operators $\Sigma_{i,j}:\De_i\to \De_i+j$ that shift the dimensions $\De_1,\De_2$. We have suppressed $\SO(d)$ indices in (\ref{eq:diffbasisexplicit}) for simplicity.

The coefficients $(\dots)$ expressing $\mathscr{D}^{(a)\cO_1\cO_2}_{\f_1\f_2}$ in terms of products of weight-shifting operators $\cD^{(b)A}\cD^{(c)}_A$ are determined by inverting (\ref{eq:spinningcbtrick}). In writing (\ref{eq:differentialbasistrick}), there are infinitely many possible choices of representation $W$ and labels $b,c$. Generically, we expect that it should always be possible to choose enough $W,b,c$'s to solve (\ref{eq:differentialbasistrick}). This was shown explicitly in \cite{Costa:2011dw} when $\cO_1,\cO_2$ are traceless-symmetric tensors.\footnote{It would be interesting to characterize the minimal set of $W$'s needed to build all possible structures.}

For simplicity, we will sometimes write (\ref{eq:differentialbasistrick}) as
\be
\label{eq:differentialbasistrickshorthand}
\diagramEnvelope{\begin{tikzpicture}[anchor=base,baseline]
\node (tl) at (-0.8,0.8) [above] {$\cO_1$};
\node (vert) at (0,0) [threept] {$a$};
\node (bl) at (-0.8,-0.8) [below] {$\cO_2$};
\node (r) at (0.8,0) [right] {$\cO_\ell$};
\draw [spinning] (vert) -- (tl);
\draw [spinning] (vert) -- (bl);
\draw [spinning] (vert) -- (r);
\end{tikzpicture}}
&=
\diagramEnvelope{\begin{tikzpicture}[anchor=base,baseline]
\node (tl) at (-1.6,1.6) [above] {$\cO_1$};
\node (vertr) at (0.6,0) [twopt] {};
\node (verttl) at (-0.7,0.7) [threept] {};
\node (vertbl) at (-0.7,-0.7) [threept] {};
\node (bl) at (-1.6,-1.6) [below] {$\cO_2$};
\node (r) at (1.6,0) [right] {$\cO_\ell$};
\node (w) at (-0.8,-0.05) [threept] {$a$};
\draw [spinning] (verttl) -- (tl);
\draw [spinning] (vertbl) -- (bl);
\draw [spinning] (vertr) -- (r);
\draw [scalar] (vertr) -- (verttl);
\draw [scalar] (vertr) -- (vertbl);
\draw [finite] (verttl) -- (w);
\draw [finite] (w) -- (vertbl);
\end{tikzpicture}}.
\ee

\section{Conformal blocks}
\label{sec:conformalblocks}

\subsection{Gluing three-point functions}
\label{sec:conformalblocksreview}

A general conformal block can be expressed as the integral of a product of three-point functions.  For simplicity, consider the case where the external and internal operators are scalars. Given three-point functions $\<\f_1(x_1)\f_2(x_3)\f(x)\>$ and $\<\f(y)\f_3(x_3)\f_4(x_4)\>$, the following object is a solution to the conformal Casimir equation with the correct  transformation properties to be a conformal block,
\be
\label{eq:shadowintegral}
\frac{1}{\cN_\De}\int d^d x\, d^d y \<\f_1(x_1)\f_2(x_3)\f(x)\> \frac{1}{(x-y)^{2(d-\De)}} \<\f(y)\f_3(x_3)\f_4(x_4)\>,
\ee
where $\De=\De_\f$.
This can be understood, for example, by writing the integral in a manifestly conformally-invariant way \cite{SimmonsDuffin:2012uy}.\footnote{In Euclidean signature, we take the range of integration of $x,y$ to be all of $\R^d$. In this case (\ref{eq:shadowintegral}) produces a solution to the conformal Casimir equation with the wrong boundary conditions to be a conformal block. However, the conformal block can be extracted by taking a suitable linear combination of analytic continuations of the integral \cite{SimmonsDuffin:2012uy}. One can alternatively isolate the conformal block by performing the integral in Lorentzian signature over a domain defined by the lightcones of the four points $x_1,x_2,x_3,x_4$ \cite{Czech:2016xec}. Calculations involving differential operators are insensitive to these issues because the differential operators always transform trivially under monodromy. Thus, our methods allow us to study spinning versions of any of the solutions to the Casimir equation.}${}^,$\footnote{We expect that (\ref{eq:shadowintegral}) only converges when $\De$ lies on the principal series $\De\in \frac d 2 + i \R$. We obtain a general conformal block by analytically continuing in $\De$.}

Let us denote the operation which glues two $\f$-correlators by\footnote{Instead of thinking of the gluing operation (\ref{eq:gluing}) in terms of shadow integrals, we can alternatively think of it as simply a sum over normalized descendants of $\f$. The only properties of the gluing procedure that we use in this work are that it is bilinear, conformally-invariant, and satisfies the normalization condition (\ref{eq:shadownormalization}).}
\be
\label{eq:gluing}
|\f\rangle\bowtie\langle \f|\equiv
\frac{1}{\cN_\De} \int d^d x\, d^dy |\f(x)\> \frac{1}{(x-y)^{2(d-\De)}} \<\f(y)|
&\quad=\quad
\diagramEnvelope{\begin{tikzpicture}[anchor=base,baseline]
	\node (opO) at (-1,0) [left] {$\f$};
	\node (opOprime) at (1,0) [right] {$\f$};
	\node (vert) at (0,0) [cross] {};
	\draw [scalar] (opO) -- (vert);
	\draw [scalar] (opOprime) -- (vert);
\end{tikzpicture}}.
\ee
We should choose the normalization $\cN_\De$ by demanding that
\be
\label{eq:shadownormalization}
\diagramEnvelope{\begin{tikzpicture}[anchor=base,baseline]
	\node (opO) at (-1,0) [left] {$\f$};
	\node (vert) at (0,0) [twopt] {};
	\node (shad) at (1,0) [cross] {};
	\node (opOprime) at (2,0) [right] {$\f$};
	\draw [scalar] (vert) -- (opO);
	\draw [scalar] (vert) -- (shad);
	\draw [scalar] (opOprime) -- (shad);
\end{tikzpicture}}
&\quad=\quad
\diagramEnvelope{\begin{tikzpicture}[anchor=base,baseline]
	\node (opO) at (-1,0) [left] {$\f$};
	\node (opOprime) at (0,0) [right] {$\f$};
	\draw [scalar] (opOprime) -- (opO);
\end{tikzpicture}}.
\ee
That is, we demand that the shadow integral acting on a two-point function $\<\f\f\>$ gives the identity transformation. In the case of scalars, this fixes the normalization factor to be \cite{Dobrev:1977qv,DO1,SimmonsDuffin:2012uy}
\be
\cN_{\De} &= \frac{\pi^d \G(\De-\frac{d}{2})\G(\frac d 2 - \De)}{\G(\De)\G(d-\De)}.
\ee

For spinning operators, $\cO$ glues to its dual-reflected representation $\cO^\dag$ --- i.e.\ the representation with which $\cO$ has a nonzero two-point function,
\be
\label{eq:spinningshadow}
|\cO_{\Delta,\rho}\rangle\bowtie\langle \cO^\dag_{\Delta,\rho^\dag}|&\equiv
\diagramEnvelope{\begin{tikzpicture}[anchor=base,baseline]
	\node (opO) at (-1,0) [left] {$\cO$};
	\node (opOprime) at (1,0) [right] {$\cO^\dag$};
	\node (vert) at (0,0) [cross] {};
	\draw [spinning] (opO) -- (vert);
	\draw [spinning] (opOprime) -- (vert);
\end{tikzpicture}}\nn\\
&\equiv
\frac{1}{\cN_{\De,\rho}} \int d^d x\, d^dy |\cO_{\Delta,\alpha}(x)\> \frac{t^{\alpha\bar\alpha}(x-y)}{(x-y)^{2(d-\De)}} \<\cO^\dagger_{\Delta,\bar \alpha}(y)|.
\ee
Here, $t^{\a\bar\a}(x-y)$ is the tensor structure appearing in the two point function of the shadow operators $\<\tl\cO\tl\cO^\dag\>$.
We will not need the explicit expression, but simply the normalization condition
\be
\diagramEnvelope{\begin{tikzpicture}[anchor=base,baseline]
	\node (opO) at (-1,0) [left] {$\cO$};
	\node (vert) at (0,0) [twopt] {};
	\node (shad) at (1,0) [cross] {};
	\node (opOprime) at (2,0) [right] {$\cO$};
	\draw [spinning] (vert) -- (opO);
	\draw [spinning] (vert) -- (shad);
	\draw [spinning] (opOprime) -- (shad);
\end{tikzpicture}}
&\quad=\quad
\diagramEnvelope{\begin{tikzpicture}[anchor=base,baseline]
	\node (opO) at (-1,0) [left] {$\cO$};
	\node (opOprime) at (0,0) [right] {$\cO$};
	\draw [spinning] (opOprime) -- (opO);
\end{tikzpicture}}.
\label{eq:shadownormalizationtwo}
\ee

A general conformal block is given by
\be
\label{eq:generalblock}
W^{ab}\equiv
\langle \cO_1\cO_2\cO\rangle^{(a)}
\bowtie
{}^{(b)}\langle  \cO^\dag \cO_3\cO_4 \rangle=
\diagramEnvelope{\begin{tikzpicture}[anchor=base,baseline]
	\node (vertL) at (-2,0) [threept] {$a$};
	\node (vertR) at (2,-0.05) [threept] {$b$};
	\node (opO1) at (-2.5,-1) [below] {$\cO_1$};
	\node (opO2) at (-2.5,1) [above] {$\cO_2$};
	\node (opO3) at (2.5,1) [above] {$\cO_3$};
	\node (opW) at (2.5,-1) [below] {$\cO_4$};
	\node (shad) at (0,.085) [cross] {};
	\node at (1,0.1) [above] {$\cO^\dag$};	
	\node at (-1,0.1) [above] {$\cO$};	
	\draw [spinning] (vertL)-- (opO1);
	\draw [spinning] (vertL)-- (opO2);
	\draw [spinning] (vertL)-- (shad);
	\draw [spinning] (vertR) -- (shad);
	\draw [spinning] (vertR)-- (opO3);
	\draw [spinning] (vertR)-- (opW);
\end{tikzpicture}}.
\ee
To perform computations with differential operators and shadow integrals, we must understand how to move differential operators from one side of a shadow integration to another --- i.e.\ how to integrate by parts. This can be done purely diagrammatically, just from the definition (\ref{eq:shadownormalizationtwo}).

First, consider a two-point function. Moving a differential operator past a two-point vertex is a special case of the definition of a $6j$ symbol,
\be
\label{eq:twoptcrossing}
\diagramEnvelope{\begin{tikzpicture}[anchor=base,baseline]
	\node (vertL) at (0,0) [twopt] {};
	\node (vertR) at (2,-0.08) [threept] {$c$};
	\node (opO1) at (-0.5,-1) [below] {$\cO^\dag$};
	\node (opO2) at (-0.5,1) [above] {$\mathbf{1}$};
	\node (opO3) at (2.5,1) [above] {$\cO'$};
	\node (opW) at (2.5,-1) [below] {$W$};	
	\node at (1,0.1) [above] {$\cO$};	
	\draw [spinning] (vertL)-- (opO1);
	\draw [finite with arrow] (vertL)-- (opO2);
	\draw [spinning] (vertL)-- (vertR);
	\draw [spinning] (vertR)-- (opO3);
	\draw [finite with arrow] (vertR)-- (opW);
\end{tikzpicture}}
	\quad=\quad
	\sum_{m}
	\left\{
		\begin{matrix}
		\cO^\dag & \mathbf{1} & \cO'^\dag \\
		\cO' & W & \cO
		\end{matrix}
	\right\}^{\uniq c}_{\uniq m}
\diagramEnvelope{\begin{tikzpicture}[anchor=base,baseline]
	\node (vertU) at (0,0.7) [twopt] {};
	\node (vertD) at (0,-0.7) [threept] {$m$};
	\node (opO1) at (-1,-1.5) [below] {$\cO^\dag$};
	\node (opO2) at (-1,1.5) [above] {$\mathbf{1}$};
	\node (opO3) at (1,1.5) [above] {$\cO'$};
	\node (opW) at (1,-1.5) [below] {$W$};	
	\node at (0.1,0) [right] {$\cO'^\dag$};	
	\draw [spinning] (vertD)-- (opO1);
	\draw [finite with arrow] (vertU)-- (opO2);
	\draw [spinning] (vertU)-- (vertD);
	\draw [spinning] (vertU)-- (opO3);
	\draw [finite with arrow] (vertD)-- (opW);
\end{tikzpicture}}.
\ee
A three-point vertex where one of the legs is the unit operator $\mathbf{1}$ is simply a two-point vertex. We could of course omit the unit operator from the above diagram, but we have temporarily included it to emphasize that (\ref{eq:twoptcrossing}) is a special case of (\ref{eq:6jdefinition}).\footnote{To be precise, we have established~\eqref{eq:6jdefinition} only for non-degenerate operators $\cO_i$. However, as explained in section~\ref{sec:convarianttwopt}, the objects on either side of~\eqref{eq:twoptcrossing} span the space of covariant two-point functions, which provides the missing ingredient.} Again, the notation ``$\uniq$" means there is a unique corresponding structure or differential operator.

Now, let us add shadow integrals onto both $\cO$ and $\cO'$ in the above diagram.  Using (\ref{eq:shadownormalizationtwo}), we find
\be
\diagramEnvelope{\begin{tikzpicture}[anchor=base,baseline]
	\node (vertU) at (1,0) [cross] {};
	\node (vertD) at (0,-0.08) [threept] {$c$};
	\node (opO1) at (-1,0) [left] {$\cO$};
	\node (opO3) at (2,0) [right] {$\cO'^\dag$};
	\node (opW) at (0,-1) [below] {$W$};	
	\node at (0.6,0.1) [above] {$\cO'$};
	\draw [spinning] (opO1)-- (vertD);
	\draw [spinning] (vertD)-- (vertU);
	\draw [spinning] (opO3)-- (vertU);
	\draw [finite with arrow] (vertD)-- (opW);
\end{tikzpicture}}
	\quad=\quad
	\sum_{m}
	\left\{
		\begin{matrix}
		\cO^\dag & \mathbf{1} & \cO'^\dag \\
		\cO' & W & \cO
		\end{matrix}
	\right\}^{\uniq c}_{\uniq m}
\diagramEnvelope{\begin{tikzpicture}[anchor=base,baseline]
	\node (vertL) at (0,0) [cross] {};
	\node (vertR) at (1,-0.08) [threept, inner sep=1pt] {$m$};
	\node (opO1) at (-1,0) [left] {$\cO$};
	\node (opO3) at (2,0) [right] {$\cO'^\dag$};
	\node (opW) at (1,-1) [below] {$W$};
	\node at (0.4,0.1) [above] {$\cO^\dag$};
	\draw [spinning] (opO1)-- (vertL);
	\draw [spinning] (vertR)-- (vertL);
	\draw [spinning] (opO3)-- (vertR);
	\draw [finite with arrow] (vertR)-- (opW);
\end{tikzpicture}}
\label{eq:integrationbyparts}
\ee
Equation~(\ref{eq:integrationbyparts}) essentially implements two integrations by parts in the double integral (\ref{eq:gluing}), allowing us to move a differential operator from one side of a shadow integral to another.
In symbolic notation it has the form
\begin{equation}\label{eq:integrationbyparts_algebraic}
|\cD^{(c)A}\cO\rangle \bowtie  \langle \cO'^\dag|=\sum_{m}
\left\{
\begin{matrix}
\cO^\dag & \mathbf{1} &  \cO'^\dag \\
\cO' & W & \cO
\end{matrix}
\right\}^{\uniq c}_{\uniq m}
|\cO\rangle \bowtie \, \langle \cD^{(m)A}\cO'^\dag|.
\end{equation}

\subsection{Spinning conformal blocks review}

The expression (\ref{eq:generalblock}) for a general block can be combined with the ``differential basis" trick (\ref{eq:differentialbasistrick}) to express certain conformal blocks as derivatives of scalar blocks \cite{Costa:2011dw}. Suppose the exchanged operator $\cO=\cO_\ell$ is a traceless-symmetric tensor of spin $\ell$. Applying (\ref{eq:differentialbasistrick}) twice, we find
\be
\diagramEnvelope{\begin{tikzpicture}[anchor=base,baseline]
	\node (vertL) at (-2,0) [threept] {$a$};
	\node (vertR) at (2,-0.05) [threept] {$b$};
	\node (opO1) at (-2.5,-1) [below] {$\cO_1$};
	\node (opO2) at (-2.5,1) [above] {$\cO_2$};
	\node (opO3) at (2.5,1) [above] {$\cO_3$};
	\node (opW) at (2.5,-1) [below] {$\cO_4$};
	\node (shad) at (0,0.085) [cross] {};
	\node at (1,0.1) [above] {$\cO_\ell$};	
	\node at (-1,0.1) [above] {$\cO_\ell$};	
	\draw [spinning] (vertL)-- (opO1);
	\draw [spinning] (vertL)-- (opO2);
	\draw [spinning] (vertL)-- (shad);
	\draw [spinning] (vertR) -- (shad);
	\draw [spinning] (vertR)-- (opO3);
	\draw [spinning] (vertR)-- (opW);
\end{tikzpicture}}
&=
\diagramEnvelope{\begin{tikzpicture}[anchor=base,baseline]
	\node (vertL) at (-1.7,0) [twopt] {};
	\node (vertR) at (1.7,0) [twopt] {};
	\node (opO1) at (-2.7,-1.5) [below] {$\cO_1$};
	\node (opO2) at (-2.7,1.5) [above] {$\cO_2$};
	\node (opO3) at (2.7,1.5) [above] {$\cO_3$};
	\node (opO4) at (2.7,-1.5) [below] {$\cO_4$};
	\node (sc1) at (-2.4,-0.7) [threept] {};
	\node (sc2) at (-2.4,0.8) [threept] {};
	\node (sc3) at (2.4,0.8) [threept] {};
	\node (sc4) at (2.4,-0.7) [threept] {};
	\node (ld) at (-2.7,0) [threept] {$a$};
	\node (rd) at (2.7,0) [threept] {$b$};
	\node (shad) at (0,0) [cross] {};
	\node at (1,0.1) [above] {$\cO_\ell$};	
	\node at (-1,0.1) [above] {$\cO_\ell$};	
	\draw [scalar] (vertL)-- (sc1);
	\draw [scalar] (vertL)-- (sc2);
	\draw [spinning] (vertL)-- (shad);
	\draw [spinning] (vertR) -- (shad);
	\draw [scalar] (vertR)-- (sc3);
	\draw [scalar] (vertR)-- (sc4);
	\draw [spinning] (sc1) -- (opO1);
	\draw [spinning] (sc2) -- (opO2);
	\draw [spinning] (sc3) -- (opO3);
	\draw [spinning] (sc4) -- (opO4);
	\draw [finite] (ld) -- (sc1);
	\draw [finite] (ld) -- (sc2);
	\draw [finite] (rd) -- (sc3);
	\draw [finite] (rd) -- (sc4);
\end{tikzpicture}}.
\label{eq:spinningfromscalar}
\ee
Note that the right-hand side is a differential operator acting on conformal blocks with external scalars. In equations (\ref{eq:spinningfromscalar}) reads
\be
\label{eq:spinningfromscalareq}
G^{(a,b)\cO_1\cO_2\cO_3\cO_4}_{\De,\ell}(x_i) &= \mathscr{D}^{(a)\cO_1\cO_2}_{\f_1\f_2} \mathscr{D}^{(b)\cO_3\cO_4}_{\f_3\f_4} G^{\f_1\f_2\f_3\f_4}_{\De,\ell}(x_i).
\ee
The objects in (\ref{eq:spinningfromscalareq}) and (\ref{eq:diffbasisexplicit}) carry $\SO(d)$ indices which we have suppressed for simplicity.

Note that symmetric traceless tensors (STTs) are the only representations that can appear in an OPE of two scalars. Because $\mathscr{D}^{(a)\cO_i\cO_j}_{\f_i\f_j}$ can't change the representation of the exchanged operator, the expression (\ref{eq:spinningfromscalareq}) only works for conformal blocks with an exchanged STT\@. This is sufficient to compute all bosonic blocks in 3d, since all bosonic (irreducible) 3d Lorentz representations are STTs. However, in general there exist blocks which cannot be computed using (\ref{eq:spinningfromscalareq}).

To compute more general blocks, an approach advocated in \cite{Costa:2011dw,SimmonsDuffin:2012uy} is to identify the simplest set of blocks with general exchanged representations --- so-called ``seed" blocks --- compute them using some other method and apply the trick (\ref{eq:spinningfromscalareq}) to those.\footnote{Seed blocks for 4d theories were classified in \cite{Echeverri:2015rwa} and computed in \cite{Echeverri:2016dun} using the Casimir equation. In 3d, there are two types of seed blocks: external scalars with exchange of spin $\ell \in \Z$, and external fermion+scalars with exchange of spin-$\ell\in \Z+\frac 1 2$. A recursion relation for the latter type of 3d seed block was computed in \cite{Iliesiu:2015akf}.} However, our new techniques will make it simple to modify (\ref{eq:spinningfromscalar}) and (\ref{eq:spinningfromscalareq}) to compute any type of conformal block (including seed blocks).

\subsection{Expression for general conformal blocks}
\label{sec:generalblockexpression}

The basic idea is to allow the differential operators acting on the left and right to be conformally-covariant, instead of simply invariant,
\be
\label{eq:moregeneralbasicidea}
G^{(a,b)\cO_1\cO_2\cO_3\cO_4}_{\cO}(x_i) &= \mathscr{D}_\textrm{left}^{(a)A} \mathscr{D}_\textrm{right}^{(b)}{}_A G^{\f_1\f_2\f_3\f_4}_{\De,\ell}(x_i),
\ee
where $A$ is an index for some finite-dimensional representation $W$ of $\SO(d+1,1)$. The exchanged operator then lives in the tensor product $W\otimes V_{\De,\ell}$, which can contain primaries with more general Lorentz representations. We must be careful to choose $\mathscr{D}_\textrm{left}^{(a)A}$ and $\mathscr{D}_\textrm{right}^{(b)}{}_A$ so that precisely one irreducible subrepresentation of $W\otimes V_{\De,\ell}$ contributes. However, this can be done easily and systematically using the techniques we have developed.

Let us begin with the object we would like to compute: a conformal block for the exchange of an operator $\cO$ transforming in $V_{\De,\rho}$,
\be
\label{eq:generalblockagain}
G_{\cO}^{(a,b)\cO_1\cO_2\cO_3\cO_4}(x_i) &=
\diagramEnvelope{\begin{tikzpicture}[anchor=base,baseline]
	\node (vertL) at (-2,0) [threept] {$a$};
	\node (vertR) at (2,-0.05) [threept] {$b$};
	\node (opO1) at (-2.5,-1) [below] {$\cO_1$};
	\node (opO2) at (-2.5,1) [above] {$\cO_2$};
	\node (opO3) at (2.5,1) [above] {$\cO_3$};
	\node (opW) at (2.5,-1) [below] {$\cO_4$};
	\node (shad) at (0,.08) [cross] {};
	\node at (1,0.1) [above] {$\cO^\dag$};	
	\node at (-1,0.1) [above] {$\cO$};	
	\draw [spinning] (vertL)-- (opO1);
	\draw [spinning] (vertL)-- (opO2);
	\draw [spinning] (vertL)-- (shad);
	\draw [spinning] (vertR) -- (shad);
	\draw [spinning] (vertR)-- (opO3);
	\draw [spinning] (vertR)-- (opW);
\end{tikzpicture}}.
\ee
Let $W$ be a finite-dimensional representation of the conformal group such that $W^*\otimes V_{\De,\rho}$ contains a spin-$\ell$ STT representation $\cO_\ell$. We can introduce a bubble of $W$ and $\cO_\ell$ in the middle of the diagram, so that the shadow integral itself involves a spin-$\ell$ representation. Note that
\be
\diagramEnvelope{\begin{tikzpicture}[anchor=base,baseline]
	\node (opLeft) at (-2,0) [left] {$\cO$};
	\node (vertLeft) at (-1,-0.08) [threept,inner sep=1pt] {$m$};
	\node (vertCenter) at (0,0) [cross] {};
	\node (vertRight) at (1,-0.08) [threept] {$n$};
	\node (opRight) at (2,0) [right] {$\cO^\dag$};
	\node at (0,-0.8) [below] {$W$};
	\node at (0,0.1) [above] {$\cO_\ell$};
	\draw [spinning] (opLeft) -- (vertLeft);
	\draw [spinning] (vertLeft) -- (vertCenter);
	\draw [spinning] (opRight) -- (vertRight);
	\draw [spinning] (vertRight) -- (vertCenter);
	\draw [finite with arrow] (vertLeft) to[out=-90,in=-90] (vertRight);
\end{tikzpicture}}
&=
\sum_{p}
	\left\{
		\begin{matrix}
		\cO^\dag & \mathbf{1} & \cO_\ell \\
		\cO_\ell & W & \cO
		\end{matrix}
	\right\}^{\uniq m}_{\uniq p}
\diagramEnvelope{\begin{tikzpicture}[anchor=base,baseline]
	\node (opLeft) at (-2,0) [left] {$\cO$};
	\node (shad) at (-1,0) [cross] {};
	\node (vert1) at (0,-0.045) [threept,inner sep=1pt] {$p$};
	\node (vert2) at (1.5,-0.08) [threept] {$n$};
	\node (opRight) at (2.5,0) [right] {$\cO^\dag$};
	\node at (0.8,-0.8) [below] {$W$};
	\node at (0.8,0.2) [above] {$\cO_\ell$};
	\draw [spinning] (opLeft) -- (shad);
	\draw [spinning] (vert1) -- (shad);
	\draw [spinning] (vert2) -- (vert1);
	\draw [spinning] (opRight) -- (vert2);
	\draw [finite with arrow] (vert1) to[out=-90,in=-90] (vert2);
\end{tikzpicture}}
\nn\\
&=
\sum_{p}
	\left\{
		\begin{matrix}
		\cO^\dag & \mathbf{1} & \cO_\ell \\
		\cO_\ell & W & \cO
		\end{matrix}
	\right\}^{\uniq m}_{\uniq p}
\begin{pmatrix}
\cO^\dag\\
\cO_\ell W
\end{pmatrix}^{pn}
\diagramEnvelope{\begin{tikzpicture}[anchor=base,baseline]
	\node (opLeft) at (-1,0) [left] {$\cO$};
	\node (shad) at (0,0) [cross] {};
	\node (opRight) at (1,0) [right] {$\cO^\dag$};
	\draw [spinning] (opLeft) -- (shad);
	\draw [spinning] (opRight) -- (shad);
\end{tikzpicture}},
\ee
where we have used (\ref{eq:integrationbyparts}) to move the differential operator $\cD^{(m)A}$ from one side of the shadow integral to the other, and (\ref{eq:bubble}) to simplify a product of differential operators $\cD^{(p)A}\cD^{(n)}_A$ on a single leg.
Thus, we have
\be
\label{eq:newdiagram}
G_{\cO}^{(a,b)\cO_1\cO_2\cO_3\cO_4}(x_i) &= \frac{1}{M_{mn}}
\diagramEnvelope{\begin{tikzpicture}[anchor=base,baseline]
	\node (vertL) at (-3,0) [threept] {$a$};
	\node (vertR) at (3,-0.05) [threept] {$b$};
	\node (opO1) at (-3.5,-1) [below] {$\cO_1$};
	\node (opO2) at (-3.5,1) [above] {$\cO_2$};
	\node (opO3) at (3.5,1) [above] {$\cO_3$};
	\node (opW) at (3.5,-1) [below] {$\cO_4$};
	\node (shad) at (0,0.08) [cross] {};
	\node (bubL) at (-1.5,0) [threept,inner sep=1.5pt] {$m$};
	\node (bubR) at (1.5,0) [threept] {$n$};
	\node at (2.2,0.1) [above] {$\cO^\dag$};	
	\node at (0,0.1) [above] {$\cO_\ell$};	
	\node at (-2.2,0.1) [above] {$\cO$};
	\node at (0,-1.1) [below] {$W$};	
	\draw [spinning] (vertL)-- (opO1);
	\draw [spinning] (vertL)-- (opO2);
	\draw [spinning] (vertL)-- (bubL);
	\draw [spinning] (vertR) -- (bubR);
	\draw [spinning] (bubL) -- (shad);
	\draw [spinning] (bubR) -- (shad);
	\draw [spinning] (vertR)-- (opO3);
	\draw [spinning] (vertR)-- (opW);
 	\draw [finite with arrow] (bubL) to[out=-90,in=-90] (bubR);
\end{tikzpicture}},
\ee
where
\be
\label{eq:defofmn}
M_{mn} &\equiv \sum_{p}
	\left\{
		\begin{matrix}
		\cO^\dag & \mathbf{1} & \cO_\ell \\
		\cO_\ell & W & \cO
		\end{matrix}
	\right\}^{\uniq m}_{\uniq p}
\begin{pmatrix}
\cO^\dag\\
\cO_\ell W
\end{pmatrix}^{pn}.
\ee
We do not sum over $m,n$ in (\ref{eq:newdiagram}) --- rather we can choose any $m,n$ such that $M_{mn}$ is nonzero.

Now we use crossing to move the $W$ vertices to the external legs. Let us focus on the left-hand side of the diagram (\ref{eq:newdiagram}),
\be
\diagramEnvelope{\begin{tikzpicture}[anchor=base,baseline]
	\node (vertL) at (0,0) [threept] {$a$};
	\node (vertR) at (1.4,0) [threept] {$m$};
	\node (opO1) at (-0.5,-1) [below] {$\cO_1$};
	\node (opO2) at (-0.5,1) [above] {$\cO_2$};
	\node (opO3) at (2.4,0.2) [right] {$\cO_\ell$};
	\node (opW) at (2,-1) [below] {$W$};	
	\node at (0.7,0.1) [above] {$\cO$};	
	\draw [spinning] (vertL)-- (opO1);
	\draw [spinning] (vertL)-- (opO2);
	\draw [spinning] (vertL)-- (vertR);
	\draw [spinning] (vertR)-- (opO3);
	\draw [finite with arrow] (vertR)-- (opW);
\end{tikzpicture}}
	&=
	\sum_{\cO',r,s}
	\left\{
		\begin{matrix}
		\cO_1 & \cO_2 & \cO' \\
		\cO_\ell & W & \cO
		\end{matrix}
	\right\}^{am}_{rs}
\diagramEnvelope{\begin{tikzpicture}[anchor=base,baseline]
	\node (vertU) at (0.3,0.7) [threept] {$r$};
	\node (vertD) at (0,-0.7) [threept] {$s$};
	\node (opO1) at (-1,-1.5) [below] {$\cO_1$};
	\node (opO2) at (-1,1.5) [above] {$\cO_2$};
	\node (opO3) at (1.2,0.8) [right] {$\cO_\ell$};
	\node (opW) at (1,-1.5) [below] {$W$};	
	\node at (0.1,0) [right] {$\cO'$};	
	\draw [spinning] (vertD)-- (opO1);
	\draw [spinning] (vertU)-- (opO2);
	\draw [spinning] (vertU)-- (vertD);
	\draw [spinning] (vertU)-- (opO3);
	\draw [finite with arrow] (vertD)-- (opW);
\end{tikzpicture}}
\ee
Now $\cO_2$ and $\cO'$ participate in a three-point vertex with an STT operator $\cO_\ell$, so we can use (\ref{eq:differentialbasistrick}) to obtain
\be
	&=
	\sum_{\cO',r,s}
	\left\{
		\begin{matrix}
		\cO_1 & \cO_2 & \cO' \\
		\cO_\ell & W & \cO
		\end{matrix}
	\right\}^{am}_{rs}
\diagramEnvelope{\begin{tikzpicture}[anchor=base,baseline]
	\node (vertU) at (0.8,0.9) [twopt] {};
	\node (diffU) at (-0.1,0.5) [threept] {$r$};
	\node (vertD) at (0,-0.7) [threept] {$s$};
	\node (opO1) at (-1,-1.5) [below] {$\cO_1$};
	\node (opO2) at (-1,1.5) [above] {$\cO_2$};
	\node (sc2) at (0,1.2) [threept] {};
	\node (sc1) at (0.4,0.2) [threept] {};
	\node (opO3) at (1.8,0.9) [right] {$\cO_\ell$};
	\node (opW) at (1,-1.5) [below] {$W$};	
	\node at (0.1,-0.3) [right] {$\cO'$};	
	\draw [spinning] (vertD)-- (opO1);
	\draw [spinning] (sc2)-- (opO2);
	\draw [scalar] (vertU)-- (sc2);
	\draw [scalar] (vertU)-- (sc1);
	\draw [spinning] (sc1) -- (vertD);
	\draw [spinning] (vertU)-- (opO3);
	\draw [finite with arrow] (vertD)-- (opW);
	\draw [finite] (diffU) -- (sc1);
	\draw [finite] (diffU) -- (sc2);
\end{tikzpicture}}.
\ee
Thus, we find
\be
\mathscr{D}_\textrm{left}^{(a)A} \<\f_1\f_2\cO_\ell\> &= \frac{1}{\sqrt{M_{mn}}}\cD^{(m)A}_x\<\cO_1\cO_2 \cO(x)\>^{(a)},\nn\\
\textrm{where}\quad\mathscr{D}_\textrm{left}^{(a)A} &\equiv \frac{1}{\sqrt{M_{mn}}}\sum_{\cO',r,s}
	\left\{
		\begin{matrix}
		\cO_1 & \cO_2 & \cO' \\
		\cO_\ell & W & \cO
		\end{matrix}
	\right\}^{am}_{rs}\cD^{(s)A}_{x_1} \mathscr{D}_{\f_1\f_2}^{(r)\cO'\cO_2},
\ee
where the $x$ subscript indicates that $\cD_x^{(m)A}$ acts on the operator $\cO(x)$.
Similarly,
\be
\mathscr{D}_\textrm{right}^{(b)}{}_A \<\f_4\f_3\cO_\ell\> &= \frac{1}{\sqrt{M_{mn}}}\cD^{(n)}_{A,x}\<\cO_4\cO_3 \cO^\dag(x)\>^{(b)}
\nn\\
\mathscr{D}_\textrm{right}^{(b)}{}_A &\equiv \frac{1}{\sqrt{M_{mn}}}\sum_{\cO',t,u}
	\left\{
		\begin{matrix}
		\cO_4 & \cO_3 & \cO' \\
		\cO_\ell & W^* & \cO^\dag
		\end{matrix}
	\right\}^{bn}_{tu}\cD^{(u)}_{x_4,A} \mathscr{D}_{\f_4\f_3}^{(t)\cO'\cO_3}.
\ee
Together with (\ref{eq:defofmn}), this gives (\ref{eq:moregeneralbasicidea}).

Schematically, applying $\mathscr{D}_\textrm{left}^{(a)A} \mathscr{D}_\textrm{right}^{(b)}{}_A$ to a scalar block results in a graph with the topology
\be
\diagramEnvelope{\begin{tikzpicture}[anchor=base,baseline]
	\node (vertL) at (-1,0) [threept] {};
	\node (vertR) at (1,0) [threept] {};
	\node (opO1) at (-1.5,-1) [below] {};
	\node (opO2) at (-1.5,1) [above] {};
	\node (opO3) at (1.5,1) [above] {};
	\node (opW) at (1.5,-1) [below] {};
	\node (shad) at (0,0) [cross] {};
	\draw [spinning] (vertL)-- (opO1);
	\draw [spinning] (vertL)-- (opO2);
	\draw [spinning] (vertL)-- (shad);
	\draw [spinning] (vertR) -- (shad);
	\draw [spinning] (vertR)-- (opO3);
	\draw [spinning] (vertR)-- (opW);
\end{tikzpicture}}
&=
\sum
\diagramEnvelope{\begin{tikzpicture}[anchor=base,baseline]
	\node (vertL) at (-0.8,0) [twopt] {};
	\node (vertR) at (0.8,0) [twopt] {};
	\node (opO1) at (-1.7,-2) [below] {};
	\node (opO2) at (-1.7,1.5) [above] {};
	\node (opO3) at (1.7,1.5) [above] {};
	\node (opO4) at (1.7,-2) [below] {};
	\node (v4) at (1.5,-1.3) [threept] {};
	\node (v1) at (-1.5,-1.3) [threept] {};
	\node (sc1) at (-1.4,-0.7) [threept] {};
	\node (sc2) at (-1.4,0.8) [threept] {};
	\node (sc3) at (1.4,0.8) [threept] {};
	\node (sc4) at (1.4,-0.7) [threept] {};
	\node (shad) at (0,0) [cross] {};
	\draw [scalar] (vertL)-- (sc1);
	\draw [scalar] (vertL)-- (sc2);
	\draw [spinning] (vertL)-- (shad);
	\draw [spinning] (vertR) -- (shad);
	\draw [scalar] (vertR)-- (sc3);
	\draw [scalar] (vertR)-- (sc4);
	\draw [spinning] (sc1) -- (v1);
	\draw [spinning] (v1) -- (opO1);
	\draw [spinning] (sc2) -- (opO2);
	\draw [spinning] (sc3) -- (opO3);
	\draw [spinning] (sc4) -- (v4);
	\draw [spinning] (v4) -- (opO4);
	\draw [finite] (sc2) -- (sc1);
	\draw [finite] (sc4) -- (sc3);
	\draw [finite with arrow] (v1) -- (v4);
\end{tikzpicture}}.
\label{eq:spinningfromscalarmoregeneral}
\ee
The inner object is a conformal block for external scalars (dashed lines). Weight-shifting operators dress it in a way  such that (a component of) the tensor $W\otimes \cO_\ell$ propagates from left to right.

The above calculation has the advantage of being extremely general. However, it requires us to make non-canonical choices of $W$ and the differential operators $m,n$. Different choices for these objects will result in naively different, but equivalent expressions for our conformal block in terms of derivatives of scalar blocks. In some cases, to obtain the simplest possible expression, we may want to proceed slightly differently.

\subsection{Expression for seed blocks}
\label{sec:seedalgo}

Let us consider for example the problem of computing the seed blocks. For simplicity of discussion, we will restrict to the case of even $d$. The case of odd $d$ can be analyzed similarly\footnote{The complication in the case of odd $d$ is that when $\cO$ is a fermion, we cannot choose the external operators so that there is a single tensor structure on each side of the seed block. Instead, the minimum is two. This is related to the fact that the irreducible fermionic representations of $\SO(d-1)$ are necessarily chiral when $d$ is odd.} (for example, we construct the 3d seed block in section~\ref{sec:seed3d}).

 As mentioned above, seed blocks are the simplest conformal blocks that exchange a primary $\cO$ in a given $\SO(d)$ representation. In particular, we can always choose the external operators in a way such that there exists a single three-point structure on either side of the block, for example
\be
\diagramEnvelope{\begin{tikzpicture}[anchor=base,baseline]
	\node (vertL) at (-2,0) [threept] {};
	\node (vertR) at (2,0) [threept] {};
	\node (opO1) at (-2.5,-1) [below] {$\cO_1$};
	\node (opO2) at (-2.5,1) [above] {$\cO_2$};
	\node (opO3) at (2.5,1) [above] {$\cO_3$};
	\node (opO4) at (2.5,-1) [below] {$\cO_4$};
	\node (shad) at (0,0) [cross] {};
	\node at (1,0.1) [above] {$\cO^\dag$};	
	\node at (-1,0.1) [above] {$\cO$};	
	\draw [scalar] (vertL)-- (opO1);
	\draw [spinning] (vertL)-- (opO2);
	\draw [spinning] (vertL)-- (shad);
	\draw [spinning] (vertR) -- (shad);
	\draw [scalar] (vertR)-- (opO3);
	\draw [spinning] (vertR)-- (opO4);
\end{tikzpicture}},
\ee
where $\cO_1$ and $\cO_3$ are scalars, while $\cO_2$ and $\cO_4$ transform in representations which are obtained from that of $\cO$ by, for example, removing the first row of the $\SO(d)$ Young diagram.

To express this seed block in terms of scalar blocks, let us first focus on the left three-point structure. We can write
\be\label{eq:seedstep1}
\diagramEnvelope{\begin{tikzpicture}[anchor=base,baseline]
	\node (vertL) at (0,0) [threept] {};
	\node (opO1) at (-0.5,-1) [below] {$\cO_1$};
	\node (opO2) at (-0.5,1) [above] {$\cO_2$};
	\node (shad) at (2,0) [] {};
	\node at (1,0.1) [above] {$\cO$};	
	\draw [scalar] (vertL)-- (opO1);
	\draw [spinning] (vertL)-- (opO2);
	\draw [spinning] (vertL)-- (shad);
\end{tikzpicture}}=
\frac{1}{C_{mn}}\diagramEnvelope{\begin{tikzpicture}[anchor=base,baseline]
	\node (vertL) at (0,0) [threept] {};
	\node (vertWL) at (-0.5,1) [threept] {$n$};
	\node (vertWR) at (1.2,-0.085) [threept,inner sep=1.5pt] {$m$};
	\node (opO1) at (-0.5,-1) [below] {$\cO_1$};
	\node (opO2) at (-1,2) [above] {$\cO_2$};
	\node (shad) at (2.7,0) [] {};
	\node at (1.8,0.1) [above] {$\cO$};	
	\node at (0.6,-0.1) [below] {$\cO_\ell$};	
	\node at (0.8,0.9) [above] {$W$};
	\node at (-0.3,0.3) [left] {$\cO_2'$};
	\draw [scalar] (vertL)-- (opO1);
	\draw [scalar] (vertL)-- (vertWL);
	\draw [spinning] (vertWL)-- (opO2);
	\draw [spinning] (vertL)-- (vertWR);
	\draw [spinning] (vertWR)-- (shad);
	\draw [finite with arrow] (vertWL) to[out=0,in=120] (vertWR);
\end{tikzpicture}},
\ee
where due to the uniqueness of the tensor structures, we are free to choose $n,m$ and $W$ as long as $\cO'_2$ is a scalar and $\cO_\ell$ is a STT\@. In what follows, we will perform manipulations with the operator labeled by $m$, but we will leave $n$ untouched. For this reason, it is convenient to choose $W$ and $n$ so that $n$ is a 0-th order differential operator. According to theorem~\ref{thm:operatorclassification}, this means that the primary of $W^*$ should transform in the same representation as $\cO_2$, i.e.\ $(W^*)_{-j}=(W_j)^*=\rho_2$, where $\rho_i$ is the $\SO(d)$ representation of $\cO_i$.\footnote{Such a $W^*$ always exists. In fact, there are infinitely many choices differing by the value of $j$, and the $W^*$ with minimal $j$ is obtained by prepending a $0$ to the list of Dynkin labels of $\rho_2$ (in the natural ordering where the vector label is the first and the spinor labels are the last).} On the other hand, the condition for existence of the structure on the left is 
\be
	(\rho\otimes\rho_2)^{\SO(d-1)}\neq 0,
\ee
where $\rho$ is the representation of $\cO$. This is equivalent to saying that there is a STT in the tensor product $\rho\otimes\rho_2=\rho\otimes (W_j)^*$. In turn, this leads to
\be
	\rho\in \text{STT}\otimes W_j.
\ee
According to theorem~\ref{thm:operatorclassification}, this implies that we can use an order-$(2j+1)$ differential operator associated to $W^*$ in place of $m$.

We can now use \eqref{eq:integrationbyparts} to move $m$ to the right three-point structure to find the piece
\be\label{eq:seedstep2}
	\sum_{c}
	\left\{
		\begin{matrix}
		\cO_\ell & \mathbf{1} & \cO^\dag \\
		\cO & W^* & \cO_\ell
		\end{matrix}
	\right\}^{\uniq m}_{\uniq c}
\diagramEnvelope{\begin{tikzpicture}[anchor=base,baseline]
	\node (shad) at (0,0) [cross] {};
	\node (vertW) at (1,-0.085) [threept] {$c$};
	\node (vertR) at (2.3,0) [threept] {};
	\node (opOl) at (-1,0) [left] {$\cO_\ell$};
	\node (opO) at (1.7,0.1) [above] {$\cO^\dag$};
	\node (opW) at (-0.5,1) [above] {$W$};
	\node at (0.4,0.1) [above] {$\cO_\ell$};
	\node (opO3) at (2.8,1) [above] {$\cO_3$};
	\node (opO4) at (2.8,-1) [below] {$\cO_4$};
	\draw [spinning] (opOl)-- (shad);
	\draw [spinning] (vertW)-- (shad);
	\draw [spinning] (vertR)-- (vertW);
	\draw [finite with arrow] (opW) to[out=0,in=90] (vertW);
	\draw [scalar] (vertR)-- (opO3);
	\draw [spinning] (vertR)-- (opO4);
\end{tikzpicture}},
\ee
to which we can apply a crossing transformation to find
\be
	=\sum_{c}\sum_{\cO_3',a,b}
	\left\{
		\begin{matrix}
		\cO_\ell & \mathbf{1} & \cO^\dag \\
		\cO & W^* & \cO_\ell
		\end{matrix}
	\right\}^{\uniq m}_{\uniq c}
	\left\{
		\begin{matrix}
		\cO_3 & \cO_4 & \cO_3' \\
		\cO_\ell & W^* & \cO^\dag
		\end{matrix}
	\right\}^{\uniq c}_{a b}
\diagramEnvelope{\begin{tikzpicture}[anchor=base,baseline]
	\node (shad) at (0,0) [cross] {};
	\node (vertW) at (3,1) [threept] {$b$};
	\node (vertR) at (2.3,-0.08) [threept] {$a$};
	\node (opOl) at (-1,0) [left] {$\cO_\ell$};
	\node (opO3p) at (2.7,0.5) [right] {$\cO_3'$};
	\node (opW) at (-0.3,1.1) [left] {$W$};
	\node at (1,0.1) [above] {$\cO_\ell$};
	\node (opO3) at (3.5,2) [above, inner sep=0pt] {$\cO_3$};
	\node (opO4) at (2.8,-1) [below] {$\cO_4$};
	\draw [spinning] (opOl)-- (shad);
	\draw [spinning] (vertR)-- (shad);
	\draw [spinning] (vertR)-- (vertW);
	\draw [finite with arrow] (opW) to[out=0,in=180] (vertW);
	\draw [scalar] (vertW)-- (opO3);
	\draw [spinning] (vertR)-- (opO4);
\end{tikzpicture}}.
\ee
We now use~\eqref{eq:differentialbasistrick} to write the full seed block as
\be\label{eq:seedgeneral}
\frac{1}{C_{mn}}\sum_{c}\sum_{\cO_3',a,b}
\left\{
\begin{matrix}
	\cO_\ell & \mathbf{1} & \cO^\dag \\
	\cO & W^* & \cO_\ell
\end{matrix}
\right\}^{\uniq m}_{\uniq c}
\left\{
\begin{matrix}
	\cO_3 & \cO_4 & \cO_3' \\
	\cO_\ell & W^* & \cO^\dag
\end{matrix}
\right\}^{\uniq c}_{a b}
\diagramEnvelope{\begin{tikzpicture}[anchor=base,baseline]
\node (opO1) at (-1,-2) [below, inner sep=0pt] {$\cO_1$};
\node (opO2) at (-1.5,3) [above, inner sep=0pt] {$\cO_2$};
\node (vertWL) at (-1,2) [threept] {$n$};
\node (vertL) at (0,0) [twopt] {};
\node (shad) at (1,0) [cross] {};
\node (vertW) at (3,2) [threept] {$b$};
\node (vertR) at (2,0) [twopt] {};
\node (vertRD) at (3,0) [threept] {$a$};
\node (vertDU) at (2.5,1) [twopt] {};
\node (vertDD) at (2.5,-1) [twopt] {};
\node (opO3p) at (2.7,1.3) [right] {$\cO_3'$};
\node (opW) at (1,2.1) [above] {$W$};
\node at (1,0.1) [above] {$\cO_\ell$};
\node (opO3) at (3.5,3) [above, inner sep=0pt] {$\cO_3$};
\node (opO4) at (3,-2) [below] {$\cO_4$};
\draw [spinning] (vertL)-- (shad);
\draw [spinning] (vertR)-- (shad);
\draw [scalar] (vertR)-- (vertDU);
\draw [spinning] (vertDU)-- (vertW);
\draw [finite with arrow] (vertWL)-- (vertW);
\draw [scalar] (vertW)-- (opO3);
\draw [scalar] (vertR)--(vertDD);
\draw [spinning] (vertDD)-- (opO4);
\draw [spinning] (vertWL)-- (opO2);
\draw [scalar] (vertL)--(vertWL);
\draw [finite] (vertDU)-- (vertRD);
\draw [finite] (vertDD)-- (vertRD);
\draw [scalar] (vertL) -- (opO1);
\end{tikzpicture}}.
\ee
The advantage of this over the more general~\eqref{eq:spinningfromscalarmoregeneral} is that we have been able to choose the  differential operator $n$ to be of zeroth order, and we also avoided acting with differential operators on one of the legs. This reduces the order of the full differential operator acting on the scalar conformal block relative to the general expression. Let us now consider some examples.

\subsubsection{Example: seed block in 3d}
\label{sec:seed3d}

Our first example is the fermion seed block in 3 dimensions. The $\SO(3)$ representations are labeled by a single (half-)integer $\ell$. If $\ell$ is integral, then the representation is bosonic, and operators $\cO_\ell$ can be exchanged in a four-point function of scalars. If $\ell$ is half-integral, then the representation is fermionic and $\cO_\ell$ can be exchanged in a scalar-fermion four-point function\footnote{Since our analysis is purely kinematical, we will label operators by their scaling dimensions and spins.}
\be\label{eq:3Dfermion_scalar_4pt}
	\<\psi_{\De_1}(s_1,x_1)\phi_{\De_2}(x_2)\phi_{\De_3}(x_3)\psi_{\De_4}(s_4,x_4)\>.
\ee
It is therefore possible to express any conformal block in terms of a scalar or fermion-scalar block. The latter were computed in~\cite{Iliesiu:2015akf} by a Zamolodchikov type recursion relation. In this section we will show how the fermion-scalar block can be expressed as a third-order differential operator acting on a scalar conformal block, thus reducing all conformal blocks in 3d to derivatives of scalar blocks.

For ease of comparison, we will follow the conventions of~\cite{Iliesiu:2015akf}. Let us review basic properties of~\eqref{eq:3Dfermion_scalar_4pt}. On each side of conformal block there exist 2 three-point structures, which can be defined using the 5d embedding formalism as
\be\label{eq:3Dspinorthreepoint}
	\<\psi_{\De_1}\phi_{\De_2}\cO_{\De,\ell}\>^{(+)}=\diagramEnvelope{\begin{tikzpicture}[anchor=base,baseline]
	\node (vert) at (0,0) [threept,inner sep=0pt] {$+$};
	\node (opO1) at (-0.5,-1) [below] {$\psi_1$};
	\node (opO2) at (-0.5,1) [above] {$\phi_2$};
	\node (opO3) at (1,0.1) [right] {$\cO_\ell$};	
	\draw [spinning] (vert)-- (opO1);
	\draw [scalar] (vert)-- (opO2);
	\draw [spinning] (vert)-- (opO3);
	\end{tikzpicture}}&=\frac{\<S_1S_0\>\<S_0X_1X_2S_0\>^{\ell-\frac{1}{2}}}{
			X_{12}^{\frac{\De_1+\De_2-\De+\ell-\frac{1}{2}}{2}}
			X_{20}^{\frac{\De_2+\De-\De_1+\ell-\frac{1}{2}}{2}}
			X_{01}^{\frac{\De_1+\De-\De_2+\ell+\frac{1}{2}}{2}}
		},\nn\\
	\<\psi_{\De_1}\phi_{\De_2}\cO_{\De,\ell}\>^{(-)}=\diagramEnvelope{\begin{tikzpicture}[anchor=base,baseline]
	\node (vert) at (0,0) [threept,inner sep=0pt] {$-$};
	\node (opO1) at (-0.5,-1) [below] {$\psi_1$};
	\node (opO2) at (-0.5,1) [above] {$\phi_2$};
	\node (opO3) at (1,0.1) [right] {$\cO_\ell$};	
	\draw [spinning] (vert)-- (opO1);
	\draw [scalar] (vert)-- (opO2);
	\draw [spinning] (vert)-- (opO3);
	\end{tikzpicture}}&=\frac{\<S_1X_2S_0\>\<S_0X_1X_2S_0\>^{\ell-\frac{1}{2}}}{
		X_{12}^{\frac{\De_1+\De_2-\De+\ell+\frac{1}{2}}{2}}
		X_{20}^{\frac{\De_2+\De-\De_1+\ell+\frac{1}{2}}{2}}
		X_{01}^{\frac{\De_1+\De-\De_2+\ell-\frac{1}{2}}{2}}
	},
\ee
and analogously for the right three-point function ($1\rightarrow 4$, $2\rightarrow3$). Here the index $0$ refers to the intermediate operator $\cO_\ell$ of dimension $\De$, and we labeled the three-point structures by their $P$-parity. Accordingly, there exist 4 conformal blocks, which can be expanded in a basis of four-point tensor structures,
\be
G^{ab}_\mathrm{seed}(s_1,s_4,x_i)=\diagramEnvelope{\begin{tikzpicture}[anchor=base,baseline]
\node (vertL) at (-2,0) [threept] {$a$};
\node (vertR) at (2,-0.055) [threept] {$b$};
\node (opO1) at (-2.5,-1) [below] {$\psi_1$};
\node (opO2) at (-2.5,1) [above] {$\phi_2$};
\node (opO3) at (2.5,1) [above] {$\phi_3$};
\node (opW) at (2.5,-1) [below] {$\psi_4$};
\node (shad) at (0,0.08) [cross] {};
\node at (1,0.2) [above] {$\cO_\ell$};	
\node at (-1,0.2) [above] {$\cO_\ell$};	
\draw [spinning] (vertL)-- (opO1);
\draw [scalar] (vertL)-- (opO2);
\draw [spinning] (vertL)-- (shad);
\draw [spinning] (vertR) -- (shad);
\draw [scalar] (vertR)-- (opO3);
\draw [spinning] (vertR)-- (opW);
\end{tikzpicture}}
=
\sum_{I=1}^4 g^{ab}_I(z,\bar z)\mathbb{T}^I_4(s_1,s_4,x_i).
\ee
As indicated, there exist $4$ four-point tensor structures $\mathbb{T}^I_4$. Out of them, two structures are parity-even and participate in conformal blocks $G^{++}, G^{--}$, and two are parity-odd and participate in $G^{+-}$ and $G^{-+}$. We give their exact form in appendix~\ref{app:3Dblocks}.

We now compute the seed blocks using the algorithm\footnote{Because we want to follow the conventions of~\cite{Iliesiu:2015akf}, some minor modifications to the algorithm are required, such as reordering of the operators.} from section~\ref{sec:seedalgo}, and we will use the spinor representation $W=\mathcal{S}$ of the 3d conformal group to translate traceless-symmetric representations into fermionic representations. The first step is to write the left three-point structures in the form~\eqref{eq:seedstep1}. Let us define the scalar three-point structures as
\be\label{eq:3Dscalartheepoint}
	\<\phi_{\De_1}\phi_{\De_2}\cO_{\De,\ell}\>=
	\diagramEnvelope{\begin{tikzpicture}[anchor=base,baseline]
		\node (vert) at (0,0) [twopt] {};
		\node (opO1) at (-0.5,-1) [below] {$\phi_1$};
		\node (opO2) at (-0.5,1) [above] {$\phi_2$};
		\node (opO3) at (1,0) [right] {$\cO_\ell$};	
		\draw [scalar] (vert)-- (opO1);
		\draw [scalar] (vert)-- (opO2);
		\draw [spinning] (vert)-- (opO3);
	\end{tikzpicture}}&=\frac{\<S_0X_1X_2S_0\>^{\ell}}{
	X_{12}^{\frac{\De_1+\De_2-\De+\ell}{2}}
	X_{20}^{\frac{\De_2+\De-\De_1+\ell}{2}}
	X_{01}^{\frac{\De_1+\De-\De_2+\ell}{2}}
},
\ee
In~\eqref{eq:seedstep1} we will use the zeroth order operator $\cD^{-+}_a$ in place of $n$. For $m$ we can take any differential operator of the appropriate parity. A simple choice is to use $\cD^{-+}_a$ for the parity even structure, and 
$\cD^{--}_a$ for the parity-odd structure. We then have
\be
	\<\psi_{\De_1}\phi_{\De_2}\cO_{\De,\ell}\>^{(\pm)}=\frac{1}{C_\pm}\<\cD^{-+}_1\cD^{-\pm}_0\>\<\phi_{\De_1+\half}\phi_{\De_2}\cO_{\De+\half,\ell\mp\frac{1}{2}}\>.
\ee
It is easy to find by a direct computation that
\be
	C_+=1,\qquad C_-=2\ell+1.
\ee
Note that $\cO_{\De+\half,\ell\mp\frac{1}{2}}$ is the operator which is going to be exchanged in the scalar block. If we chose different operators for $m$ (i.e.\ $\cD^{+\mp}$) in~\eqref{eq:seedstep1}, then we would relate the seed block to different scalar blocks (in particular, it doesn't make sense to mix these choices).

\paragraph{Crossing of 2-point functions} The next step is to learn how to push the operators $\cD^{-\pm}_a$ through the shadow integral. For that we need to fix the normalization of two-point functions, which we choose to be
\be\label{eq:3Dtwopoint}
	\<\cO_{\De,\ell}(S_1,X_1)\cO_{\De,\ell}(S_2,X_2)\>=i^{2\ell}\frac{\<S_1S_2\>^{2\ell}}{X_{12}^{\De+\ell}}.
\ee
The definition of $6j$ symbols~\eqref{eq:twoptcrossing} is in our case
\begin{multline}
	{\cD}^{-\pm}_{2,a}\<\cO_{\De+\half,\ell\mp\frac{1}{2}}(S_1,X_1)\cO_{\De+\half,\ell\mp\frac{1}{2}}(S_2,X_2)\>=\\=
	\left\{
	\begin{matrix}
		\cO_{\De+\half,\ell\mp\frac{1}{2}} & \mathbf{1} & \cO_{\De,\ell} \\
		\cO_{\De,\ell} & \mathcal{S} & \cO_{\De+\half,\ell\mp\frac{1}{2}}
	\end{matrix}
	\right\}^{\uniq (-\pm)}_{\uniq (+\mp)}
	\cD^{+\mp}_{1,a}\<\cO_{\De,\ell}(S_1,X_1)\cO_{\De,\ell}(S_2,X_2)\>.
\end{multline}
We can explicitly compute
\be
	\left\{
	\begin{matrix}
		\cO_{\De+\half,\ell-\frac{1}{2}} & \mathbf{1} & \cO_{\De,\ell} \\
		\cO_{\De,\ell} & \mathcal{S} & \cO_{\De+\half,\ell-\frac{1}{2}}
	\end{matrix}
	\right\}^{\uniq (-+)}_{\uniq (+-)}&=\frac{i}{8\ell(\De-1)(\De-\ell-1)},\\
	\left\{
	\begin{matrix}
		\cO_{\De+\half,\ell+\frac{1}{2}} & \mathbf{1} & \cO_{\De,\ell} \\
		\cO_{\De,\ell} & \mathcal{S} & \cO_{\De+\half,\ell+\frac{1}{2}}
	\end{matrix}
	\right\}^{\uniq (--)}_{\uniq (++)}&=\frac{i(2\ell+1)}{4(\De-1)(\De+\ell)},
\ee
and use these coefficients in~\eqref{eq:integrationbyparts} to arrive at~\eqref{eq:seedstep2}. At this point, we have expressed the seed block in the form
\begin{multline}
G^{\pm b}_\mathrm{seed}(s_1,s_4,x_i)=\\=\frac{1}{C_\pm}
	\left\{
\begin{matrix}
	\cO_{\ell\mp\frac{1}{2}} & \mathbf{1} & \cO_\ell \\
	\cO_\ell & \mathcal{S} & \cO_{\ell\mp\frac{1}{2}}
\end{matrix}
\right\}^{\uniq (-\pm)}_{\uniq (+\mp)}
\Omega^{cd}\cD^{-+}_{1,c}\<\phi_{\De_1+\half}\phi_{\De_2}\cO_{\De+\half,\ell\mp\half}\>\bowtie\cD^{+\mp}_{0,d}\<\cO_{\De,\ell}\phi_{\De_3}\psi_{\De_4}\>^{(b)},
\end{multline}
where $\bowtie$ stands for shadow integral. 

\paragraph{Crossing of three-point functions}Now we are going to perform the crossing transformation on the right three-point function to write it as
\be\label{eq:3DseedCrossing}
\cD^{+\mp}_{0,d}\<\cO_{\De,\ell}\phi_{\De_3}\psi_{\De_4}\>^{(b)}=\sum_{b'}&
\left\{
\begin{matrix}
	\phi_{\De_3} & \psi_{\De_4} & \psi_{\De_3+\half} \\
	\cO_{\De,\ell} & \mathcal{S} & \cO_{\De+\half,\ell\mp\half}
\end{matrix}
\right\}^{b (+\mp)}_{b' (--)}
\cD^{--}_{3,d}\<\cO_{\De+\half,\ell\mp\half}\psi_{\De_3+\half}\psi_{\De_4}\>^{(b')}\nn\\
&+
\left\{
\begin{matrix}
	\phi_{\De_3} & \psi_{\De_4} & \psi_{\De_3-\half} \\
	\cO_{\De,\ell} & \mathcal{S} & \cO_{\De+\half,\ell\mp\half}
\end{matrix}
\right\}^{b (+\mp)}_{b' (+-)}
\cD^{+-}_{3,d}\<\cO_{\De+\half,\ell\mp\half}\psi_{\De_3-\half}\psi_{\De_4}\>^{(b')}.
\ee
To proceed, we need to choose a basis of tensor structures for three-point functions of the type $\<\cO_{\De+\half,\ell\mp\half}\psi_{\De_3\pm\half}\psi_{\De_4}\>$. We define
\be
	t_1&=\frac{\<S_1S_2\>\<S_3X_1X_2S_3\>^\ell}{X_{12}^\ell}+\frac{\<S_1S_3\>\<S_2S_3\>\<S_3X_1X_2S_3\>^{\ell-1}}{X_{12}^{\ell-1}},\\
	t_2&=\frac{\<S_1S_2\>\<S_3X_1X_2S_3\>^\ell}{X_{12}^\ell}+2\frac{\<S_1S_3\>\<S_2S_3\>\<S_3X_1X_2S_3\>^{\ell-1}}{X_{12}^{\ell-1}},\\
\label{eq:structdefinition3}
	t_3&=\frac{\<S_3X_1X_2S_3\>^{\ell-1}}{
			X_{12}^{\ell-\half}X_{23}^{\half}X_{31}^{\half}
		}X_{23}\<S_1S_3\>\<S_2X_1S_3\>,\\
\label{eq:structdefinition4}		
	t_4&=\frac{\<S_3X_1X_2S_3\>^{\ell-1}}{
			X_{12}^{\ell-\half}X_{23}^{\half}X_{31}^{\half}
		}X_{13}\<S_2S_3\>\<S_1X_2S_3\>,
\ee
where the first two structures are parity-even and the second two are parity-odd.\footnote{We choose this peculiar basis only for the purposes of presentation, because in it $6j$ symbols have the simplest form. In practice we used the basis~\eqref{eq:3DmonomialBasis}, in which we know the general $6j$ symbols for the spinor representation.} In terms of these structures we set\footnote{We thank Soner Albayrak for pointing out some typos in previous versions of (\ref{eq:structdefinition3}), (\ref{eq:structdefinition4}), and (\ref{eq:structdefinitionagain}).}
\be
\label{eq:structdefinitionagain}
	\<\psi_1\psi_2\cO_\ell\>^{(b)}=\frac{t_b}{
		X_{12}^{\frac{\De_1+\De_2-\De_3-\ell+1}{2}}
		X_{23}^{\frac{\De_2+\De_3-\De_1+\ell}{2}}
		X_{31}^{\frac{\De_3+\De_1-\De_2+\ell}{2}}
	}.
\ee
We can now compute the $6j$ symbols in~\eqref{eq:3DseedCrossing}. For example, the only non-vanishing symbols for $b=+$ and $\cD^{++}$ on the left of~\eqref{eq:3DseedCrossing} are
\be
\left\{
\begin{matrix}
	\phi_{\De_3} & \psi_{\De_4} & \psi_{\De_3+\half} \\
	\cO_{\De,\ell} & \mathcal{S} & \cO_{\De+\half,\ell+\half}
\end{matrix}
\right\}^{+ (++)}_{1 (--)}&=(-1)^{\ell+\half}\frac{
		(\De-\frac{3}{2})(\De+\ell+\De_{3}-\De_4-\half)(\De+\ell+\De_{3}+\De_4-\frac{3}{2})
	}{
		\De_3-\frac{3}{2}
	},\\
\left\{
\begin{matrix}
	\phi_{\De_3} & \psi_{\De_4} & \psi_{\De_3+\half} \\
	\cO_{\De,\ell} & \mathcal{S} & \cO_{\De+\half,\ell+\half}
\end{matrix}
\right\}^{+ (++)}_{2 (--)}&=(-1)^{\ell+\half}\frac{
	(\De+\ell+\De_{3}-\De_4-\half)((\De-1)(\De+\ell+\De_{3}+\De_4-\frac{5}{2})-\half)
}{
	\De_3-\frac{3}{2}
},\\
\left\{
\begin{matrix}
	\phi_{\De_3} & \psi_{\De_4} & \psi_{\De_3-\half} \\
	\cO_{\De,\ell} & \mathcal{S} & \cO_{\De+\half,\ell+\half}
\end{matrix}
\right\}^{+ (++)}_{3 (+-)}&=(-1)^{\ell+\half}\frac{
	(\De+\ell+\De_{3}-\De_4-\half)
}{
	4(\De_3-\frac{3}{2})(\De_3-2)
},\\
\left\{
\begin{matrix}
	\phi_{\De_3} & \psi_{\De_4} & \psi_{\De_3-\half} \\
	\cO_{\De,\ell} & \mathcal{S} & \cO_{\De+\half,\ell+\half}
\end{matrix}
\right\}^{+ (++)}_{4 (+-)}&=(-1)^{\ell+\half}\frac{
	(\De-1)(\De+\ell-\De_{3}+\De_4+\half)
}{
	2(\De_3-\frac{3}{2})(\De_3-2)
}.
\ee
The other symbols vanish due to space parity. The are 12 more non-vanishing $6j$ symbols for other choices of $b$ and of the operator on the left, which we won't list here since they represent only an intermediate step in our calculation.

\paragraph{Differential basis} The final step is to express the three-point structures  $\<\cO_{\De+\half,\ell\pm\half}\psi_{\De_3\pm\half}\psi_{\De_4}\>^{(b)}$ in terms of derivatives acting on scalar thee-point structures. This is standard, and this particular case was solved in~\cite{Iliesiu:2015qra}, so we do not explain it in detail. We only note that the operators which create the parity-even structures $t_1$ and $t_2$ should be parity even,
\be
t_1,\,t_2\sim \<\cD^{++}_3\cD^{++}_4\>,\, \<\cD^{-+}_3\cD^{-+}_4\>,
\ee
while operators which create parity-odd structures have to be parity-odd,
\be
t_3,\,t_4\sim \<\cD^{-+}_3\cD^{++}_4\>,\, \<\cD^{++}_3\cD^{-+}_4\>.
\ee

\paragraph{The recursion relation} Assembling everything together, we arrive at the following expressions for the seed blocks in terms of third-order differential operators acting on scalar blocks,
\be\label{eq:3DparityEvenSeed}
	G^{++}_\mathrm{seed},G^{--}_\mathrm{seed}=
		&v_1\<\cD^{-+}_1\cD^{--}_3\>\<\cD^{-+}_3\cD^{++}_4\>\<\phi_{\De_1+\half}\phi_{\De_2}\phi_{\De_3+1}\phi_{\De_4-\half}\>\nn\\
		&+v_2\<\cD^{-+}_1\cD^{--}_3\>\<\cD^{++}_3\cD^{-+}_4\>\<\phi_{\De_1+\half}\phi_{\De_2}\phi_{\De_3}\phi_{\De_4+\half}\>\nn\\
		&+v_3\<\cD^{-+}_1\cD^{+-}_3\>\<\cD^{++}_3\cD^{++}_4\>\<\phi_{\De_1+\half}\phi_{\De_2}\phi_{\De_3-1}\phi_{\De_4-\half}\>\nn\\
		&+v_4\<\cD^{-+}_1\cD^{+-}_3\>\<\cD^{-+}_3\cD^{-+}_4\>\<\phi_{\De_1+\half}\phi_{\De_2}\phi_{\De_3}\phi_{\De_4+\half}\>,\\
	\label{eq:3DparityOddSeed}
	G^{+-}_\mathrm{seed},G^{-+}_\mathrm{seed}=
		& v_1\<\cD^{-+}_1\cD^{--}_3\>\<\cD^{++}_3\cD^{++}_4\>\<\phi_{\De_1+\half}\phi_{\De_2}\phi_{\De_3}\phi_{\De_4-\half}\>\nn\\
		&+ v_2\<\cD^{-+}_1\cD^{--}_3\>\<\cD^{-+}_3\cD^{-+}_4\>\<\phi_{\De_1+\half}\phi_{\De_2}\phi_{\De_3+1}\phi_{\De_4+\half}\>\nn\\
		&+ v_3\<\cD^{-+}_1\cD^{+-}_3\>\<\cD^{-+}_3\cD^{++}_4\>\<\phi_{\De_1+\half}\phi_{\De_2}\phi_{\De_3}\phi_{\De_4-\half}\>\nn\\
		&+ v_4\<\cD^{-+}_1\cD^{+-}_3\>\<\cD^{++}_3\cD^{-+}_4\>\<\phi_{\De_1+\half}\phi_{\De_2}\phi_{\De_3-1}\phi_{\De_4+\half}\>.
\ee
The coefficients $v_i$ are different for each of the blocks, and we give the explicit expressions in appendix~\ref{app:3Dblocks}. The scalar blocks in the above expressions for $G^{+\pm}_\text{seed}$ correspond to exchange of $[\De+\half,\ell-\half]$, while for  $G^{-\pm}_\text{seed}$ the exchanged primary is $[\De+\half,\ell+\half]$. 

\paragraph{Decomposition into components} Note that the scalar conformal blocks have the form
\be
	\<\phi_{\De_1}\phi_{\De_2}\phi_{\De_3}\phi_{\De_4}\>= \frac{1}{x_{12}^{\De_1+\De_2} x_{34}^{\De_3+\De_4}} \p{\frac{x_{14}^2}{x_{24}^2}}^\a \p{\frac{x_{14}^2}{x_{13}^2}}^\b G^{\alpha,\beta}_{\De,\ell}(z,\bar z),
\ee
where $\a=-\frac 1 2 \De_{12}$, $\b=\frac 1 2 \De_{34}$, and depend essentially only on $\alpha$ and $\beta$ and not the individual dimensions $\De_i$. We then see that e.g.\ for $G^{++}_\mathrm{seed}$ we only need the scalar blocks $G^{\alpha-\tfrac{1}{4},\beta-\tfrac{1}{4}}_{\De+\half,\ell- \half}$ and $G^{\alpha-\tfrac{1}{4},\beta+\tfrac{3}{4}}_{\De+\half,\ell- \half}$. There exists a second-order differential operator (see \cite{DO3} and section~\ref{sec:DolanOsbornShifts}) which relates these two blocks,
\be\label{eq:3Dscalardimshift}
	G^{\alpha-\tfrac{1}{4},\beta+\tfrac{3}{4}}_{\De+\half,\ell- \half}(z,\bar z)\sim (\ptl_z\ptl_{\bar z}+\ldots) G^{\alpha-\tfrac{1}{4},\beta-\tfrac{1}{4}}_{\De+\half,\ell- \half}(z,\bar z).
\ee
In~\eqref{eq:3DparityEvenSeed} only a first order operator acts on $G^{\alpha-\tfrac{1}{4},\beta+\tfrac{3}{4}}_{\De+\half,\ell- \half}$, and thus we can use~\eqref{eq:3Dscalardimshift} to  reduce~\eqref{eq:3DparityEvenSeed} to another third-order operator acting on the single scalar block. In particular, we can write 
\be
	G^{++}_\mathrm{seed}=
		g^{++}_1(z,\bar z)\frac{[-\thalf,0,0,-\thalf]+[\thalf,0,0,\thalf]}{2}
		+g^{++}_2(z,\bar z)\frac{[-\thalf,0,0,\thalf]+[\thalf,0,0,-\thalf]}{2},
\ee
where the tensor structures are defined in appendix~\ref{app:3Dblocks} and
\be
	g^{++}_k(z,\bar z)=\frac{i(-1)^{\ell-\half}}{\ell(\De-\ell-1)(\De-1)}(z \bar z)^{-\frac{\De_1+\De_2+\half}{2}}\mathfrak{D}^{++}_k G^{\alpha-\tfrac{1}{4},\beta-\tfrac{1}{4}}_{\De+\half,\ell-\half}(z,\bar z).
\ee
The differential operators $\mathfrak{D}^{++}_k$ are given by\footnote{In simplifying these expressions for the differential operators we made use of the quadratic Casimir equation satisfied by the scalar conformal blocks.}
\be
	\mathfrak{D}^{++}_1(z,\bar z)=&
		\bar z\ptl_{\bar z}D_z-z\ptl_z D_{\bar z}-(z\ptl_z-\bar z\ptl_{\bar z})\frac{z\bar z}{2(z-\bar z)}\big(
		(1-z)\ptl_z-(1-\bar z)\ptl_{\bar z}
		\big)\nn\\
		&+\frac{(\De-\ell)(\De-\ell-3)}{4}(z\ptl_z-\bar z\ptl_{\bar z})+\frac{\De-\ell-3}{2}(D_z-D_{\bar z}),\\
	\mathfrak{D}^{++}_2(z,\bar z)=&
		\nabla_z D_{\bar z}+\nabla_{\bar z} D_{z}+(\nabla_z+\nabla_{\bar z})\frac{z\bar z}{2(z-\bar z)}\big(
			(1-z)\ptl_z-(1-\bar z)\ptl_{\bar z}
		\big)\nn\\
		&-\frac{(\Delta-\ell)(\Delta-\ell-3)}{4}(\nabla_z+\nabla_{\bar z})+\frac{(2\ell+1)(\De-\ell-3)(\De-\tfrac{3}{2})}{4},
\ee
where
\be
	D_z &= z^2(1-z)\ptl_z^2-(\alpha'+\beta'+1)z^2\ptl_z-\alpha'\beta'z,\quad \alpha'=\alpha-\tfrac{1}{4},\beta'=\beta-\tfrac{1}{4},\\
	\nabla_z&=z\ptl_z+\frac{z}{z-\bar z},
\ee
and $D_{\bar z},\, \nabla_{\bar z}$ are defined by exchanging $z$ and $\bar z$.

The same reduction to a single block happens for $G^{--}_\mathrm{seed}$. For $G^{+-}_\mathrm{seed}$ and $G^{-+}_\mathrm{seed}$ the situation is a little trickier since there is a second order differential operator acting on the ``wrong'' scalar block. However, it turns out that its second-order piece is in fact coming precisely from the dimension shifting operator, and we again can reduce to a third-order differential operator acting on a single scalar block. Explicit expressions for these blocks can be written in a compact form given in appendix~\ref{app:3Dblocks} together with an explanation of the normalization conventions.

\subsubsection{Example: seed blocks in 4d}
\label{sec:seed4d}

In 4-dimensions the operators in a generic spin representation are labeled by 2 non-negative integers\footnote{Notice a difference in conventions relative to the 3-dimensional case where $\ell$ can be half-integer for fermionic operators.} $\ell$ and $\bar\ell$
\begin{equation}
\cO_{\Delta,\rho}= \cO_{\Delta}^{(\ell,\bar\ell)}.
\end{equation}
It is convenient to distinguish different classes of representations by a parameter $p$ defined as
\begin{equation}
p\equiv |\ell-\bar\ell|.
\end{equation}
Operators with $p=0$ are the symmetric traceless tensors. Using~\eqref{eq:cfthreeptcount} one can easily check that any given four-point function can exchange operators with only a finite number of different values of $p$. This implies that contrary to the 3-dimensional case, in 4-dimensions we need infinitely many seed conformal blocks, parametrized by $p$. 

A calculation of the general 4-dimensional seed conformal blocks was first performed in~\citep{Echeverri:2016dun}, where the explicit expressions for $p\leq 8$ were found.  In this section we perform an alternative computation of the seed blocks by using our new machinery and the strategy outlined in section~\ref{sec:seedalgo}. Our approach is to express the $p$ seed blocks in terms of the $p-1$ seed blocks. Knowing such a relation allows one to apply it recursively $p$ times to get an expression of the $p$ seed block in terms of the derivatives of the scalar $p=0$ Dolan-Osborn block~\cite{DO1,DO2}. Since the latter is known in terms of ${}_2F_1$ hypergeometric functions, this also gives hypergeometric expressions for the seed blocks, equivalent to those in~\cite{Echeverri:2016dun}.\footnote{With normalization conventions derived in~\cite{Cuomo:2017wme}. We performed the check for $p\leq4$.}

Let us note that the explicit hypergeometric expressions of~\cite{Echeverri:2016dun} are quite complex already for $p=2$. In numerical conformal bootstrap one usually requires simple rational approximations to conformal blocks~\cite{Simmons-Duffin:2015qma,Poland:2011ey,Kos:2014bka}, which are hard to construct from these expressions. On the other hand, our differential recurrence relation is rather simple, and we thus hope that it will find applications in the numerical bootstrap.

As in section~\ref{sec:basic_differential_operators_4D}, it will be convenient to use the 6d embedding formalism described in~\cite{SimmonsDuffin:2012uy,Elkhidir:2014woa,Echeverri:2015rwa,Echeverri:2016dun,Cuomo:2017wme}. In what follows we use the conventions of~\cite{Cuomo:2017wme}, and all the computations are performed using the Mathematica package described therein. To avoid repetition, the notation and conventions from~\cite{Cuomo:2017wme} will be used in this section without explanation.\footnote{The only difference is that we avoid using the terminology of~\cite{Echeverri:2015rwa,Echeverri:2016dun,Cuomo:2017wme} in which ``conformal partial waves'' refer to what we normally mean by conformal blocks, while ``conformal blocks'' refer to the coordinates in a basis of four-point tensor structures. When there is a danger of misinterpretation, we call the latter simply the components of conformal blocks. We do so to avoid the possible confusion with conformal partial waves from harmonic analysis.} 

A simple choice for the seed four-point function where the operator $\cO_{\Delta}^{(\ell,\bar\ell)}$ with a given $p$ can be exchanged in the $s$-channel is\footnote{The seed 4-point functions are chosen so that there is a unique conformal block for the exchange of $\cO^{(\ell,\bar\ell)}_\De$. There is an ambiguity in choosing the seed 4-point function, here we use the convention of~\cite{Echeverri:2016dun}.}
\begin{equation}\label{eq:seed_4_point_function}
\langle \cF_1^{(0,0)}\cF_2^{(p,0)}\cF_3^{(0,0)}\cF_4^{(0,p)}\rangle.
\end{equation}
The conformal block associated to the exchange of $\cO_{\Delta}^{(\ell,\bar\ell)}$ in the seed 4-point function is 
\begin{equation}\label{eq:seed_CPWs_4D}
W^{(p)}_{\ell,\bar\ell}\equiv
\langle F_{\Delta_1}^{(0,0)}F_{\Delta_2}^{(p,0)}\cO_\Delta^{(\ell,\bar\ell)}\rangle
\bowtie
\langle\conj \cO_{\Delta}^{(\bar\ell,\ell)} F_{\Delta_3}^{(0,0)}F_{\Delta_4}^{(0,p)} \rangle.
\end{equation}
We distinguish 2 cases depending on the sign of $\ell-\bar\ell$. Using the convention of~\citep{Echeverri:2016dun}
we define the ``seed"\footnote{In this paper we sometimes use ``primal seed'' to distinguish from the dual seeds.} and ``dual seed" conformal blocks as
\begin{align}
W^{(p)}_{seed}                 &\equiv W^{(p)}_{\ell,\bar\ell},\quad \ell\leq \bar\ell,\\
\overline W^{(p)}_{dual\;seed} &\equiv W^{(p)}_{\ell,\bar\ell},\quad \ell\geq \bar\ell.
\end{align}
The seed and the dual seed conformal blocks can be further decomposed into components as
\begin{align}\label{eq:seed_blocks_4D}
W_{seed}^{(p)} &=\mathcal{K}_4\sum_{e=0}^p (-2)^{p-e}H_{e}^{(p)}(z,\bar z) \big[\hat \II^{42}\big]^e \big[\hat \II^{42}_{31}\big]^{p-e},\\
\label{eq:dual_seed_blocks_4D}
\overline W_{dual\;seed}^{(p)} &=\mathcal{K}_4\sum_{e=0}^p (-2)^{p-e}\overline H_{e}^{(p)}(z,\bar z) \big[\hat \II^{42}\big]^e \big[\hat \II^{42}_{31}\big]^{p-e}.
\end{align}
The parameter $e=0,\ldots,p$ labels the possible 4-point tensor structures.
In this section we focus solely on the seed blocks $H_{e}^{(p)}(z,\bar z)$. The case of the dual blocks $\overline H_{e}^{(p)}(z,\bar z)$ is completely analogous and will be addressed in appendix~\ref{app:dual_seeds}.

The calculation essentially follows the algorithm in section~\ref{sec:seedalgo}, the main difference being that we go from exchange of $(\ell,\ell+p)$ to $(\ell,\ell+p-1)$ instead of going directly to an STT exchange. The calculation is also largely analogous to the 3-dimensional calculation in section~\ref{sec:seed3d}. For convenience, we start the algorithm from the right three-point structure instead of going from the left. 

We first rewrite the right three-point function entering~\eqref{eq:seed_CPWs_4D} as
\begin{equation}\label{eq:seed4dStep1}
\langle \conj \cO_{\Delta}^{(\ell+p,\ell)} F_{\Delta_3}^{(0,0)}F_{\Delta_4}^{(0,p)} \rangle=
(\Dmpz{0} \cdot \Dmzp{4})\;
\langle \conj \cO_{\Delta+1/2}^{(\ell+p-1,\ell)} F_{\Delta_3}^{(0,0)}F_{\Delta_4+1/2}^{(0,p-1)} \rangle.
\end{equation}
The subscript $0$ indicates that $\Dmpz{0}$ acts on the internal operator $\bar\cO$. We would like to move it across $\bowtie$ (integrate by parts) using the rule~\eqref{eq:integrationbyparts_algebraic}.

\paragraph{Crossing of 2-point functions}
The definition of the $6j$ symbol entering~\eqref{eq:integrationbyparts_algebraic} in the present case is
\begin{align}\label{eq:2pf_6j_seeds}
&\Dpzm{2\,a} \langle \conj \cO_{\Delta}^{(\ell+p,\ell)}(X_1,S_1,\bar S_1) \cO_{\Delta}^{(\ell,\ell+p)}(X_2,S_2,\bar S_2)\rangle
\nn\\
&=
\mathcal{A}\,
\Dmpz{1\,a} \langle \conj \cO_{\Delta+1/2}^{(\ell+p-1,\ell)}(X_1,S_1,\bar S_1) \cO_{\Delta+1/2}^{(\ell,\ell+p-1)}(X_2,S_2,\bar S_2)\rangle,
\end{align}
where
\begin{equation}
\mathcal{A}\equiv
\left\{
		\begin{matrix}
		\conj \cO_{\Delta}^{(\ell+p,\ell)} & \mathbf{1} & \conj \cO_{\Delta+1/2}^{(\ell+p-1,\ell)} \\
		\cO_{\Delta+1/2}^{(\ell,\ell+p-1)} & \mathcal{S} & \cO_{\Delta}^{(\ell,\ell+p)}
		\end{matrix}
	\right\}^{\uniq (+0-)}_{\uniq (-+0)}
=2i(\ell+p)(\Delta-\tfrac{p}{2}-1)(\Delta-\ell-\tfrac{p}{2}-2).
\end{equation}
Applying~\eqref{eq:integrationbyparts_algebraic} and~\eqref{eq:seed4dStep1} to~\eqref{eq:seed_CPWs_4D} we arrive at
\begin{equation}\label{eq:seed_CPW_rewritten}
W^{(p)}_{seed}=
\mathcal{A}^{-1}\,
(\Dpzm{0} \cdot \Dmzp{4})
\langle F_{\Delta_1}^{(0,0)}F_{\Delta_2}^{(p,0)}\cO_{\Delta}^{(\ell,\ell+p)}\rangle
\bowtie
\langle \overline \cO_{\Delta+1/2}^{(\ell+p-1,\ell)} F_{\Delta_3}^{(0,0)}F_{\Delta_4}^{(0,p-1)} \rangle,
\end{equation}
where $\Dpzm{0}$ now acts on the left three-point function.

\paragraph{Crossing of 3-point functions}
We now use the crossing equation for the 3-point function
\begin{align}\label{eq:6j_symbol_3pf}
\Dpzm{0\,a} \langle F_{\Delta_1}^{(0,0)}F_{\Delta_2}^{(p,0)}\cO_\Delta^{(\ell,\ell+p)}\rangle=
&\sum_{n=1}^2
\mathcal{B}^{(n)} \Dmmz{1\,a}\langle F_{\Delta_1+1/2}^{(1,0)}F_{\Delta_2}^{(p,0)}\cO_{\Delta+1/2}^{(\ell,\ell+p-1)}\rangle^{(n)}+\nonumber\\
&\sum_{n=1}^2
\mathcal{C}^{(n)} \Dpzm{1\,a}\langle F_{\Delta_1-1/2}^{(0,1)}F_{\Delta_2}^{(p,0)}\cO_{\Delta+1/2}^{(\ell,\ell+p-1)}
\rangle^{(n)},
\end{align}
where $\mathcal{B}^{(n)}$ and $\mathcal{C}^{(n)}$ denote the $6j$ symbols
\begin{align}
\mathcal{B}^{(n)} &\equiv
\left\{
		\begin{matrix}
		F_{\Delta_1}^{(0,0)} & F_{\Delta_2}^{(p,0)} & F_{\Delta_1+1/2}^{(1,0)} \\
		\cO_{\Delta+1/2}^{(\ell,\ell+p-1)} & \mathcal{S} & \cO_\Delta^{(\ell,\ell+p)}
		\end{matrix}
	\right\}^{\uniq (+0-)}_{(n) (--0)}, \nn\\
\mathcal{C}^{(n)} &\equiv
\left\{
		\begin{matrix}
		F_{\Delta_1}^{(0,0)} & F_{\Delta_2}^{(p,0)} & F_{\Delta_1-1/2}^{(0,1)} \\
		\cO_{\Delta+1/2}^{(\ell,\ell+p-1)} & \mathcal{S} & \cO_\Delta^{(\ell,\ell+p)}
		\end{matrix}
	\right\}^{\uniq (+0-)}_{(n) (+0-)}.
\end{align}
The 3-point functions in the right-hand side of~\eqref{eq:6j_symbol_3pf} have the following form
\begin{align}
\langle F_{\Delta_1+1/2}^{(1,0)}F_{\Delta_2}^{(p,0)}\cO_{\Delta+1/2}^{(\ell,\ell+p-1)}\rangle^{(i)} &=
\mathcal{K}_3 [\hat I^{32}]^{p-1} [\hat J^3_{12}]^{\ell-1} \begin{pmatrix} \hat I^{32} \hat K_{2}^{13} \nn\\
\hat I^{31} \hat K_{1}^{23} \end{pmatrix},\\
\langle F_{\Delta_1-1/2}^{(0,1)}F_{\Delta_2}^{(p,0)}\cO_{\Delta+1/2}^{(\ell,\ell+p-1)}\rangle^{(i)} &=
\mathcal{K}'_3 [\hat I^{32}]^{p-1} [\hat J^3_{12}]^{\ell-1} \begin{pmatrix} \hat I^{13} \hat I^{32} \\
\hat I^{12} \hat J^{3}_{12} \end{pmatrix},
\end{align}
Again, we can find the $6j$ symbols $\mathcal{B}^{(n)}$ and $\mathcal{C}^{(n)}$ by an explicit calculation,
\begin{align}
\mathcal{B}^{(1)} &=   
\mathcal{B}^{(2)}-\frac{\ell (\Delta_1+\Delta_2+\Delta-\ell-p-6) }{4 (\Delta_1-2)}\times\nn\\
&\left(4 (\ell+p+1) (\Delta_1-\Delta_2+\ell+\frac{p}{2}+1)+(\Delta_1-\Delta_2+\Delta+\ell) (2 \Delta-4 \ell-3 p-6)\right),\nn\\
\mathcal{B}^{(2)} &= -\frac{p (\Delta_1-\Delta_2+\Delta+\ell) (2 \Delta-2 \ell-p-4) (\Delta_1+\Delta_2+\Delta-\ell-p-6)}{4 (\Delta_1-2)},\nn\\
\mathcal{C}^{(1)} &= -\frac{\ell (2 \Delta+p-2) (\Delta_1-\Delta_2-\Delta+\ell+p+2)}{4 (\Delta_1-3) (\Delta_1-2)} ,\nn\\
\mathcal{C}^{(2)} &=  \frac{p (-2 \Delta+2 \ell+p+4) (\Delta_1-\Delta_2-\Delta+\ell+p+2)}{4 (\Delta_1-3) (\Delta_1-2)}.
\end{align}

\paragraph{Differential basis}
The last step is to relate the 3-point functions entering~\eqref{eq:6j_symbol_3pf} to the seed 3-point functions
$\langle F_{\Delta_1^\prime}^{(0,0)}F_{\Delta_2^\prime}^{(p-1,0)}\cO_{\Delta+1/2}^{(\ell,\ell+p-1)}\rangle$ with shifted dimensions by using the differential basis trick. This is standard~\cite{Costa:2011dw,Echeverri:2015rwa}, so we simply note that we use the following differential operators
\be\label{eq:differentialbasis4d}
\langle F_{\Delta_1+1/2}^{(1,0)}F_{\Delta_2}^{(p,0)}\cO_{\Delta+1/2}^{(\ell,\ell+p-1)}
\rangle^{(n)}&\quad\sim\quad (\Dmpz{1} \cdot \Dppz{2})
,
(\Dppz{1} \cdot \Dmpz{2}),\nn\\
\langle F_{\Delta_1-1/2}^{(0,1)}F_{\Delta_2}^{(p,0)}\cO_{\Delta+1/2}^{(\ell,\ell+p-1)}\rangle^{(n)} &\quad\sim\quad (\Dpzp{1} \cdot \Dppz{2}),
(\Dmzp{1} \cdot \Dmpz{2}).
\ee

\paragraph{The recursion relation}
Combining the expressions~\eqref{eq:seed_CPW_rewritten}, \eqref{eq:6j_symbol_3pf}, and the differential basis~\eqref{eq:differentialbasis4d} we find the following recursion relation
\begin{align}\label{eq:seed_CPWs_recursion_relation}
W^{(p)}_{\Delta,\ell;\;\Delta_1,\Delta_2,\Delta_3,\Delta_4} &=\nn\\
 \mathcal{A}^{-1}\Bigg( &v_1
(\Dmmz{1}\cdot \Dmzp{4})(\Dmpz{1}\cdot \Dppz{2})\;W^{(p-1)}_{\Delta+\tfrac{1}{2},\ell;\;\Delta_1+1,\Delta_2-\tfrac{1}{2},\Delta_3,\Delta_4+\tfrac{1}{2}}\nn\\
+&v_2(\Dmmz{1}\cdot\Dmzp{4})( \Dppz{1}\cdot \Dmpz{2})\;W^{(p-1)}_{\Delta+\tfrac{1}{2},\ell;\;\Delta_1,\Delta_2+\tfrac{1}{2},\Delta_3,\Delta_4+\tfrac{1}{2}}\nn\\
+&v_3(\Dpzm{1}\cdot\Dmzp{4})(\Dpzp{1}\cdot \Dppz{2})\;W^{(p-1)}_{\Delta+\tfrac{1}{2},\ell;\;\Delta_1-1,\Delta_2-\tfrac{1}{2},\Delta_3,\Delta_4+\tfrac{1}{2}}\nn\\
+&v_4(\Dpzm{1}\cdot\Dmzp{4})(\Dmzp{1}\cdot \Dmpz{2})\;W^{(p-1)}_{\Delta+\tfrac{1}{2},\ell;\;\Delta_1,\Delta_2+\tfrac{1}{2},\Delta_3,\Delta_4+\tfrac{1}{2}}\Bigg),
\end{align}
where the coefficients $v_i$ are given explicitly by
\begin{align}
v_1 &= \frac{(\Delta +\Delta_1-\Delta_2+\ell) (-\Delta -\Delta_1+\Delta_2+\ell+2) (\Delta +\Delta_1+\Delta_2-\ell-p-6)}{4 (\Delta_1-2) (2 \Delta_2+p-4)}  ,\nn\\
v_2 &= \frac{(-\Delta +\Delta_1-\Delta_2+\ell+p+2) (\Delta +\Delta_1-\Delta_2-\ell-2 p-2) (\Delta +\Delta_1+\Delta_2-\ell-p-6)}{8 (\Delta_1-2) (\Delta_1-1)} ,\nn\\
v_3 &= \frac{-\Delta +\Delta_1-\Delta_2+\ell+p+2}{4 (\Delta_1-3) (\Delta_1-2)^2 (2 \Delta_2+p-4)}  ,\nn\\
v_4 &= -\frac{(-\Delta +\Delta_1-\Delta_2+\ell+p+2) (-\Delta +\Delta_1+\Delta_2+\ell+2 p-2) (\Delta +\Delta_1+\Delta_2-\ell-p-6)}{8 (\Delta_1-3) (\Delta_1-2)}.
\end{align}

\paragraph{Decomposition into components}
By using~\eqref{eq:seed_blocks_4D} one can write the recursion relation~\eqref{eq:seed_CPWs_recursion_relation} at the level of components of the seed conformal blocks $H_{e}^{(p)}(z,\bar z)$.

First let us notice that according to~\citep{Echeverri:2016dun} the components $H_{e}^{(p)}(z,\bar z)$ of the seed blocks depend on the external scaling dimensions $\Delta_i$ only via the quantities
\begin{equation}
a^{p}_{e}\equiv a^{(p)},\quad
b^{p}_{e}\equiv b^{(p)}+p-e,\quad
c^{p}_{e}\equiv p-e,
\end{equation}
where
\begin{equation}
a^{(p)}\equiv-\frac{\Delta_1-\Delta_2-p/2}{2},\quad
b^{(p)}\equiv+\frac{\Delta_3-\Delta_4-p/2}{2}.
\end{equation}

Let us now analyze the expression~\eqref{eq:seed_CPWs_recursion_relation}. Almost all the conformal blocks entering the right hand side of~\eqref{eq:seed_CPWs_recursion_relation} correspond to the same parameters $a^{(p)}$ and $b^{(p)}$ (the difference in p is compensated by a difference in $\Delta_i$). The only exception is the conformal block 
\begin{equation}\label{eq:rouge_CPW}
W^{(p-1)}_{\Delta+\tfrac{1}{2},\ell;\;\Delta_1+1,\Delta_2-\tfrac{1}{2},\Delta_3,\Delta_4+\tfrac{1}{2}}
\end{equation}
which contains $a^{(p)}-1$ and $b^{(p)}$. Just as in the case of 3-dimensions in section~\ref{sec:seed3d}, we can use a dimension shifting operator to simplify the structure of the recursion relation~\eqref{eq:seed_CPWs_recursion_relation}. The only difference is that we need to shift the external dimensions of a general seed block. This can be done by generalizing the construction of dimension-shifting operator outlined in section~\ref{sec:DolanOsbornShifts}. We find
\begin{equation}\label{eq:relation_shifted_dimensions_CPWs}
W^{(p-1)}_{\Delta+\tfrac{1}{2},\ell;\;\Delta_1+1,\Delta_2-\tfrac{1}{2},\Delta_3,\Delta_4+\tfrac{1}{2}}=
\mathcal{E}^{-1}(\Dpmz{1}\cdot \Dmmz{2})(\Dppz{1}\cdot \Dmpz{2})
W^{(p-1)}_{\Delta+\tfrac{1}{2},\ell;\;\Delta_1,\Delta_2+\tfrac{1}{2},\Delta_3,\Delta_4+\tfrac{1}{2}},
\end{equation}
where
\begin{equation}
\mathcal{E}\equiv-(p+1) (\Delta_1-1) (\Delta_1-2) (\Delta +\Delta_1-\Delta_2+\ell) (\Delta +\Delta_1-\Delta_2-\ell-2).
\end{equation}
Note that this is in fact completely analogous to the differential basis trick, except that instead of changing the external spins, we change the external dimensions.

Plugging the relation~\eqref{eq:relation_shifted_dimensions_CPWs} in~\eqref{eq:seed_CPWs_recursion_relation}, stripping off the kinematic factor
and decomposing this relation into components according to~\eqref{eq:seed_blocks_4D} one obtains a recursion relation for the seed blocks of the form
\begin{equation}\label{eq:seed4dfinal}
H_e^{(p)}(z,\bar z)=-\frac{\mathcal{A}^{-1}}{z-\bar z}
\left(
\mathit{D}_0\; H_{e}^{(p-1)}(z,\bar z)
-2\mathit{D}_1\; H_{e-1}^{(p-1)}(z,\bar z)
+4\cParam{p-1}{e-2}z\bar z\mathit{D}_2\; H_{e-2}^{(p-1)}(z,\bar z)
\right),
\end{equation}
where the conformal block in the l.h.s depends on $[\Delta,\ell;\;\Delta_1,\Delta_2,\Delta_3,\Delta_4]$ while the conformal blocks in the r.h.s.\ depend on
$[\Delta+\tfrac{1}{2},\ell;\;\Delta_1,\Delta_2+\tfrac{1}{2},\Delta_3,\Delta_4+\tfrac{1}{2}]$.
The differential operators $D_i$ are given by
\begin{align}
\label{eq:d_0}
D_0 \equiv 
&\normalD{\bar z}{\bParam{p-1}{e}}\hyperD{z}{p-1,e}-\normalD{z}{\bParam{p-1}{e}}\hyperD{\bar z}{p-1,e}\nn\\
&+k\left(\hyperD{z}{p-1,e}-\hyperD{\bar z}{p-1,e}\right)-(\cParam{p-1}{e}+1)L[\bParam{p-1}{e}]B\!\left[-\frac{k(k-2)}{1+\cParam{p-1}{e}}\right],\\
\label{eq:d_1}
D_1 \equiv
&z\normalD{\bar z}{\bParam{p-1}{e-1}+\cParam{p-1}{e-1}}\hyperD{z}{p-1,e-1}-\bar z\normalD{z}{\bParam{p-1}{e-1}+\cParam{p-1}{e-1}}\hyperD{\bar z}{p-1,e-1}\nn\\
&+k\left( z\hyperD{z}{p-1,e-1}-\bar z\hyperD{\bar z}{p-1,e-1}\right)\nn\\
&+(2\cParam{p-1}{e-1}+1)z\bar z L[\bParam{p-1}{e-1}](z-\bar z)^{-1}L[a]-(k-2)(k-\cParam{p-1}{e-1}-1) (z-\bar z)B[k],\\
\label{eq:d_2}
D_2 \equiv
&\hyperD{z}{p-1,e-2}-\hyperD{\bar z}{p-1,e-2}- L[a] B\!\left[k-\cParam{p-1}{e-2}-1\right],
\end{align}
where the coefficient $k$ is
\begin{equation}
k\equiv\frac{4-\Delta+\ell}{2}+\frac{3p}{4}.
\end{equation}
The elementary differential operators\footnote{Exactly the same differential operators (except for $\normalD{x}{\mu}$) enter the quadratic Casimir equation for the seed blocks~\cite{Echeverri:2016dun}. Note that here the definition of $L$ differs by a factor of $z-\bar z$} used here are
\begin{align}
D_x^{(a,b;c)} &\equiv x^2(1-x)\ptl_x^2-\left((a+b+1)x^2-cx\right)\ptl_x-a b x,\\
\normalD{x}{\mu} &\equiv -x(1-x)\ptl_x+\mu\,x,\\
L[\mu] &\equiv \normalD{z}{\mu} - \normalD{\bar z}{\mu},\\
B[\mu] &\equiv \frac{z \bar z}{z-\bar z}\left( (1-z)\partial_x - (1-\bar z)\partial_{\bar z}  \right)+\mu,
\end{align}
and we also use the following short-hand notation
\begin{equation}
\hyperD{x}{p,e}\equiv D_x^{(\aParam{p}{e},\bParam{p}{e};\cParam{p}{e})}.
\end{equation}

\subsection{Dimension-shifting and spin-shifting}
\label{sec:DolanOsbornShifts}

Using our techniques, we can explain some of the identities for scalar conformal blocks which were derived by Dolan and Osborn in~\cite{DO3}. For the ease of comparison, in this section we use the notation of~\cite{DO3}, which we now briefly recall. The scalar conformal block is defined as
\be\label{eq:DolanOsbornDefn}
	\<\phi_{\De_1}(x_1)\phi_{\De_2}(x_2)|\cO_{\De,\ell}|\phi_{\De_3}(x_3)\phi_{\De_4}(x_4)\>=\frac{1}{x_{12}^{\De_1+\De_2}x_{34}^{\De_3+\De_4}}\p{\frac{x_{24}}{x_{14}}}^{-2a}\p{\frac{x_{14}}{x_{13}}}^{2b} F_{\lambda_1\lambda_2}(a,b,x,\bar x),
\ee
where $x$ and $\bar x$ are the standard Dolan-Osborn coordinates denoted by $z$ and $\bar z$ in the rest of this paper, 
\be
	x\bar x=\frac{x_{12}^2x_{34}^2}{x_{13}^2x_{24}^2},\quad 	(1-x)(1-\bar x)=\frac{x_{23}^2x_{14}^2}{x_{13}^2x_{24}^2},\quad 
\ee
and
\be
	a=-\frac{1}{2}\De_{12},\quad b=\frac{1}{2}\De_{34},
\ee
while the parameters $\lambda_i$ are defined as
\be
	\lambda_1=\frac{1}{2}(\De+\ell),\quad \lambda_2=\frac{1}{2}(\De-\ell).
\ee

\paragraph{Operators $\mathcal{H}_k$} Let us consider acting on~\eqref{eq:DolanOsbornDefn} with the following contraction of the vector operators~\eqref{eq:vectoroperators},
\be 
-2\cD^{-0}_1\cdot \cD^{-0}_4=-2X_1\cdot X_4=x_{14}^2.
\ee
The resulting four-point function will have scaling dimensions at positions $1$ and $4$ shifted by $-1$. Accordingly, we can remove the prefactor for the new set of scaling dimensions to find the resulting action of this operator on $F_{\lambda_1\lambda_2}$,
\be
(x\bar x)^{-\frac{1}{2}}F_{\lambda_1\lambda_2}(a,b,x,\bar x).
\ee
This operation is equivalent to the following diagram,
\be\label{eq:DOsimple}
\diagramEnvelope{\begin{tikzpicture}[anchor=base,baseline]
	\node (opO1) at (-1,-2) [below, inner sep=0pt] {$[\De_1-1,0]$};
	\node (opO1p) at (-0.5,-0.5) [left, inner sep=0pt] {$[\De_1,0]$};
	\node (opO2) at (-0.5,1) [above, inner sep=0pt] {$[\De_2,0]$};
	\node (vertWL) at (-0.5,-1) [threept] {};
	\node (vertW) at (2.5,-1) [threept] {};
	\node (vertL) at (0,0) [twopt] {};
	\node (shad) at (1,0) [cross] {};
	\node (vertR) at (2,0) [twopt] {};
	\node (opW) at (1,-1.1) [below] {$\myng{(1)}$};
	\node at (1,0.1) [above] {$[\De,\ell]$};
	\node (opO3) at (2.5,1) [above, inner sep=0pt] {$[\De_3,0]$};
	\node (opO4) at (3,-2) [below] {$[\De_4-1,0]$};
	\node (opO4p) at (2.5,-0.5) [right, inner sep=0pt] {$[\De_4,0]$};
	\draw [spinning] (vertL)-- (shad);
	\draw [spinning] (vertR)-- (shad);
	\draw [scalar] (vertR)-- (vertW);
	\draw [finite with arrow] (vertWL)-- (vertW);
	\draw [scalar] (vertW)-- (opO4);
	\draw [scalar] (vertR)--(opO3);
	\draw [scalar] (vertWL)-- (opO1);
	\draw [scalar] (vertL)--(vertWL);
	\draw [scalar] (vertL) -- (opO2);
	\end{tikzpicture}},
\ee
and thus according to our general analysis can be expanded using the finite-dimensional crossing~\eqref{eq:6jdefinition} in terms of scalar conformal blocks with shifted external dimensions  and the internal representations appearing in
\be
\myng{(1)}\otimes [\De,\ell]=[\De-1,\ell]\oplus [\De,\ell+1]\oplus [\De,\ell-1]\oplus [\De+1,\ell]\oplus\ldots,
\ee
where $\ldots$ represents non-STT representations which do not appear in a four-point function of scalars. In the notation of~\cite{DO3}, this corresponds to an equality of the form
\be
(x\bar x)^{-\frac{1}{2}}F_{\lambda_1\lambda_2}(a,b)=\,&r\, F_{\lambda_1-\half\,\lambda_2-\half}(a+\thalf,b+\thalf)+s\, F_{\lambda_1+\half\,\lambda_2-\half}(a+\thalf,b+\thalf)\nn\\
&+t\, F_{\lambda_1-\half\,\lambda_2+\half}(a+\thalf,b+\thalf)+u\, F_{\lambda_1+\half\,\lambda_2+\half}(a+\thalf,b+\thalf),
\ee
where the coefficients $r,s,t,u$ are some combinations of the $6j$ symbols~\eqref{eq:6jdefinition}. This is precisely the equation (4.18) in~\cite{DO3}. Dolan and Osborn also introduce $k$-th order differential operators $\mathcal{H}_k$ for $k=1,2,3$, which act on $F_{\lambda_1\lambda_2}$ in the same way but with different sets of coefficients $r_k,s_k,t_k,u_k$. In particular, they all increase $a$ and $b$ by $\thalf$. In our formalism we can also find 3 other operators with such a property,
\be\label{eq:ourHkoperatorsPre}
\cD_{13}&=\cD^{-0}_1\cdot \cD^{+0}_3,\nn\\
\cD_{24}&=\cD^{+0}_2\cdot \cD^{-0}_4,\nn\\
\cD_{23}&=\cD^{+0}_2\cdot \cD^{+0}_3,
\ee
all of which also exchange the vector representation in a way similar to~\eqref{eq:DOsimple}, and thus act in the same way as $\mathcal{H}_k$. In fact, one can express $\mathcal{H}_k$ in terms of these operators, and we provide explicit expressions in appendix~\ref{app:Hkoperators}.

\paragraph{Operators $\mathcal{F}_k$} Another class of operators introduced in~\cite{DO3} can be interpreted as exchanges of the adjoint representation of conformal group. The simplest of such exchanges is given by
\be
	\cF_0=8 \cD^{-0}_{1,[m}\cD^{-0}_{2,n]}\cD^{-0,[m}_{3}\cD^{-0,n]}_{4},
\ee
whose action on the functions $F_{\lambda_1\lambda_2}$ is equivalent to
\be
	\cF_0=\frac{1}{x}+\frac{1}{\bar x}-1,
\ee
which is precisely how $\cF_0$ is defined in~\cite{DO3}. The action of this operator on a conformal block corresponds to the following diagram,
\be\label{eq:DOadjoint}
\diagramEnvelope{\begin{tikzpicture}[anchor=base,baseline]
\node (opO1) at (-1,-2) [below, inner sep=0pt] {$[\De_1-1,0]$};
\node (opO1p) at (-0.5,-0.5) [left, inner sep=0pt] {$[\De_1,0]$};
\node (opO2) at (-1,2) [above, inner sep=0pt] {$[\De_2-1,0]$};
\node (opO2p) at (-1,0.3) [above, inner sep=0pt] {$[\De_2,0]$};
\node (vertWL) at (-0.5,-1) [threept] {};
\node (vertW) at (3.5,-1) [threept] {};
\node (vertWUL) at (-0.5,1) [threept] {};
\node (vertWU) at (3.5,1) [threept] {};
\node (vertW3L) at (0.3,-0.5) [twopt] {};
\node (vertW3R) at (2.7,-0.5) [twopt] {};
\node (vertL) at (0,0) [twopt] {};
\node (shad) at (1.5,0) [cross] {};
\node (vertR) at (3,0) [twopt] {};
\node (opW) at (1.5,-0.6) [below] {$\myng{(1,1)}$};
\node at (1.5,0.1) [above] {$[\De,\ell]$};
\node (opO3) at (4,2) [above, inner sep=0pt] {$[\De_3-1,0]$};
\node (opO3p) at (4,0.3) [above, inner sep=0pt] {$[\De_3,0]$};
\node (opO4) at (4,-2) [below] {$[\De_4-1,0]$};
\node (opO4p) at (3.5,-0.5) [right, inner sep=0pt] {$[\De_4,0]$};
\draw [spinning] (vertL)-- (shad);
\draw [spinning] (vertR)-- (shad);
\draw [scalar] (vertR)-- (vertW);
\draw [finite] (vertWL)-- (vertW3L)--(vertW3R)--(vertW);
\draw [scalar] (vertW)-- (opO4);
\draw [scalar] (vertR)--(vertWU);
\draw [scalar] (vertWU)--(opO3);
\draw [scalar] (vertWL)-- (opO1);
\draw [scalar] (vertL)--(vertWL);
\draw [scalar] (vertL) -- (vertWUL);
\draw [scalar] (vertWUL)--(opO2);
\draw [finite] (vertWUL) to[in=100,out=0] (vertW3L);
\draw [finite] (vertWU) to[in=80,out=180] (vertW3R);
\end{tikzpicture}},
\ee
where the individual differential operators have indices in the vector representation and are then joined into the adjoint representation $\myng{(1,1)}\in \myng{(1)}\otimes\myng{(1)}$. Therefore, it decomposes into scalar blocks appearing in the tensor product
\be
\myng{(1,1)}\otimes[\De,\ell]=[\De-1,\ell+1]\oplus[\De-1,\ell-1]\oplus[\De+1,\ell+1]\oplus[\De+1,\ell-1]\oplus[\De,\ell]\oplus\ldots,
\ee
where ``$\ldots$" represents non-STT representations which do not appear in scalar conformal blocks. Thus there exists an identity of the form
\be
\cF_0 F_{\lambda_1\,\lambda_2}=r_0 F_{\lambda_1\,\lambda_2-1}+s_0 F_{\lambda_1-1\,\lambda_2}+t_0 F_{\lambda_1+1\,\lambda_2}+u_0 F_{\lambda_1\,\lambda_2+1}+w_0 F_{\lambda_1\,\lambda_2},
\ee
with coefficients $r_0,s_0,t_0,u_0,w_0$ being some combinations of the $6j$ symbols~\eqref{eq:6jdefinition}. This is precisely (4.28) of \cite{DO3}. The operators $\cF_k$ with $k=1,2,3$ can be constructed analogously.

\paragraph{Operator $\cD^{(\varepsilon)}$} Finally, let us consider the identity (4.50) of~\cite{DO3}, which is\footnote{Note that there is a typo in the second part of (4.43) in~\cite{DO3}. The correct definition is $\cD^{(\varepsilon)}=(x\bar x)^{-\half}\cH_2$.}
\be\label{eq:DOdimshift}
	(x\bar x)^{\varepsilon-b+1}\cD^{(\varepsilon)}(x\bar x)^{b-\varepsilon}F_{\lambda_1\,\lambda_2}(a,b,x,\bar x)=(\lambda_1+b)(\lambda_2+b-\epsilon)F_{\lambda_1\,\lambda_2}(a,b+1,x,\bar x).
\ee
We see that the left hand side of this expression gives a differential operator which shifts $b$ by $1$. In our formalism, it is extremely easy to construct this operator, namely
\be
	(x\bar x)^{\varepsilon-b+1}\cD^{(\varepsilon)}(x\bar x)^{b-\varepsilon}=\frac{\cD^{+0}_3\cdot\cD^{-0}_4}{(\De_3-1)(d-2-\De_3)}.
\ee
From the definition it is clear that it simply shifts $b$ by $1$. The coefficient in the right hand side of~\eqref{eq:DOdimshift} can be easily expressed in terms of $6j$ symbols~\eqref{eq:6jdefinition}.

\subsection{Recursion relations for conformal blocks}
\label{sec:recursionrelationsection}

In sections~\ref{sec:generalblockexpression} and~\ref{sec:seedalgo} we have managed to express an arbitrary conformal block in terms of derivatives of scalar blocks, schematically
\be\label{eq:generalblockschematic}
G_{\De,\rho}=\sum_k c_k(\De) D_k G^\text{scalar}_{\De+\delta_k,\ell_k},
\ee
where $[\De,\rho]$ is the representation of the exchanged operator, $D_k$ are some $\De$-independent differential operators, and $c_k(\De)$ are rational functions.  All ingredients in this formula implicitly depend on the dimensions and representations of the external operators, as well as on $\rho$. In practice we often have a generic spin parameter $\ell$ in $\rho$, and we can keep it generic in this formula as we did in the examples in sections~\ref{sec:seed3d} and~\ref{sec:seed4d}. The spins $\ell_k$ are then finite shifts of $\ell$, $\ell_k=\ell+\delta\ell_k$. 

Explicit examples of such expressions are given in~\eqref{eq:3DparityEvenSeed}, \eqref{eq:3DparityOddSeed} and~\eqref{eq:seed_CPWs_recursion_relation}. They readily allow us to compute the spinning conformal blocks numerically. But they also allow us to analytically infer properties of the spinning blocks from the known properties of the scalar blocks. 

For example, a general method for numerical computation of conformal blocks is based on Zamolodchikov recursion relations~\cite{Zamolodchikov:1985ie,Zamolodchikov:1987}. The basic idea is that for certain values $\De_i$ of the scaling dimension $\De$ the generalized Verma module for the representation $[\De,\rho]$ has null descendants $[\De'_i,\rho'_i]$, which lead to poles in the conformal block for $[\De,\rho]$ with the residue being proportional to the conformal block for $[\De'_i,\rho'_i]$,
\be
	G_{\De,\rho}\sim \frac{R_i}{\De-\De_i}G_{\De'_i,\rho'_i},
\ee
where $R_i$ are certain coefficients, which in the case of spinning blocks generically are matrices rotating the left and right three-point structures in $G$. For fixed $\rho$ there are in general several infinite families of poles $\De_i$. If we know the asymptotic behavior of the conformal blocks for $\De\to\infty$,
\be
	G_{\De,\rho}\sim r^\De h_{\infty,\rho},
\ee
where $r$ is the radial coordinate of~\cite{Pappadopulo:2012jk,Hogervorst:2013sma} and $h_{\infty,\rho}$ is some relatively easily computable function, then we can write the conformal block as a sum over residues~\cite{Kos:2013tga,Penedones:2015aga}. The resulting approximation is perfectly suited for numerical applications based on semidefinite methods~\cite{Kos:2013tga,Kos:2014bka,Dymarsky:2017xzb}. 

To accomplish this program, one needs to understand the pole positions $\De_i$, the representations of null states $[\De_i',\rho_i']$, and the residue matrices $R_i$. This data has been determined for general scalar blocks~\cite{Kos:2013tga,Kos:2014bka} as well as some examples of spinning blocks~\cite{Penedones:2015aga,Dymarsky:2017xzb,Iliesiu:2015akf}. Although the classification of the poles $\De_i$ and the null states~$[\De_i',\rho_i']$ is known~\cite{Penedones:2015aga,Yamazaki:2016vqi,Oshima:2016gqy}, the computation of the residue matrices $R_i$ may not be an easy task. 

Our expression~\eqref{eq:generalblockschematic} is perfectly suited for this problem. Indeed, from it the pole structure of $G_{\De,\rho}$ is completely apparent. In particular, the poles in $G_{\De,\rho}$ are given by the poles of the scalar blocks in the right hand side, and a \textit{finite}\footnote{For a fixed $\ell$.} number of poles of the coefficients $c_k(\De)$. The residues of the poles are easy to compute. Indeed, any residue is given by a sum of differential operators $D_k$ acting on some scalar blocks. Using the techniques of section~\ref{sec:generalblockexpression}, it is easy to express the action of $D_k$ on a general scalar block as a sum over conformal blocks which can appear for the given external operators,\footnote{In particular, substituting these expressions in~\eqref{eq:generalblockschematic}, we get a tautology.}
\be
	D_k G_{\De,\ell}^\text{scalar}\sim \sum_{\De',\rho'} G_{\De',\rho'}.
\ee

In other words, our techniques allow us to translate the known recursion relations for scalar blocks into recursion relations for general conformal blocks. This approach has already been used in~\cite{Dymarsky:2017xzb} for the exchange of traceless-symmetric representations. The new ingredient here is that we can now derive the recursion relation for general internal representations. For example, using the equations~\eqref{eq:3DparityEvenSeed} and~\eqref{eq:3DparityOddSeed}, we re-derived the recursion relation of~\cite{Iliesiu:2015akf} for the scalar-fermion seed blocks exchanging a fermionic representation.

In~\cite{Penedones:2015aga} the residues of the conformal blocks were computed explicitly by considering the action of the differential operators $\cD_i$ corresponding to the null states on the three-point functions, and the behavior of the norm of the null state near the pole. We expect that the conformally-covariant differential operators can be useful also in this approach. For example, the null state differential operators $\cD_i$ can be obtained by the translation functor from a set of basic operators~\cite{MR1255551}. In our language this means that one can write the operators $\cD_i$ as
\be
	\cD_i\propto\cD_{A}\cD_i'\cD^{A},
\ee
where $\cD_i'$ are some simpler differential operators (for instance, many null states can be obtained from $\cD'=d$ the exterior derivative acting on differential forms.). The action of $\cD_i$ on a three-point function can then be computed by applying a crossing transformation to move $\cD^{A}$ on a different leg and then acting with $\cD'_i$.

\section{Further Applications}
\label{sec:furtherapps}

\subsection{Inversion formulae and ``spinning-down" a four-point function}
\label{sec:spinningdown}

Orthogonality relations between conformal blocks are useful tools for analyzing crossing symmetry. By exploiting orthogonality, we can derive inversion formulae that express OPE data in terms of an integral of a conformal block against a four-point function \cite{Dobrev:1977qv,Dobrev:1975ru}. Applying an $s$-channel inversion formula to a $t$-channel conformal block expansion, we can study crossing directly in terms of CFT data.\footnote{Note that the integral in an inversion formula in general does not commute with the sum over conformal blocks in the $t$-channel, so this analysis must be done carefully.} The coefficients relating $t$-channel blocks and $s$-channel blocks are sometimes called ``crossing kernels." Inversion formulae and crossing kernels for scalar operators have been discussed recently in \cite{Hogervorst:2017sfd,Hogervorst:2017kbj,Caron-Huot:2017vep,Gadde:2017sjg}. Here, we briefly describe how our techniques are perfectly suited for studying inversion formulae and crossing kernels for spinning operators. We will omit details, and simply highlight how weight-shifting operators can be used in these computations. We leave detailed discussion and examples for later work \cite{ShadowFuture}.

Our starting point is a conformally-invariant pairing between a four-point function of operators $\cO_i$ in representations $[\De_i,\rho_i]$ and a four-point function of shadow operators $\tl \cO_i$ in representations $[d-\De_i,\rho_i^*]$. This can be written
\be
\<F,G\> &= \frac{1}{\mathrm{Vol}(\SO(d+1,1))}\int \p{\prod_{i=1}^4 d^d x_i} F_{a_1a_2a_3a_4}(x_i) G^{a_1 a_2 a_3 a_4}(x_i)
\nn\\
&=
\diagramEnvelope{\begin{tikzpicture}[anchor=base,baseline]
	\node (f) at (-1.5,0) [threept] {$F$};
	\node (g) at (1.5,0) [threept] {$G$};
	\node (o1) at (1.1,-0.4) [below] {$\cO_1$};
	\node (o2) at (1.1,0.55) [above] {$\cO_2$};
	\node (o3) at (2.2,0.6) [above] {$\cO_3$};
	\node (o4) at (2.1,-0.6) [below] {$\cO_4$};
	\draw [spinning] (g) to[in=45,out=135] (f);
	\draw [spinning] (g) to[in=-45,out=-135] (f);
	\draw [spinning] (g) to[in=135,out=45,looseness=1.6] (f);
	\draw [spinning] (g) to[in=-135,out=-45,looseness=1.6] (f);
\end{tikzpicture}}.
\label{eq:spinningpairing}
\ee
In our diagrammatic language, an incoming line for $\cO$ is equivalent to an outgoing line for $\tl \cO$, and connecting lines means contracting indices and integrating over Euclidean space. To get a finite result for $\<F,G\>$, we must divide by the volume of the conformal group acting on all four points $x_i$. In practice, this means gauge-fixing and inserting the appropriate Faddeev-Popov determinant.

Consider first the case of scalar operators $\cO_i$.
An orthogonal basis with respect to the pairing $\<\cdot,\cdot\>$ is given by linear combinations of blocks that are single-valued in Euclidean space,
\be
F_{\De,\ell} &= \frac 1 2 \p{G_{\De,\ell} + S_{\De,\ell} G_{d-\De,\ell}},
\ee
 where $\De=\frac d 2 + i \nu$ is restricted to the principal series.\footnote{We must also include the so-called ``discrete series" in non-even dimensions \cite{Dobrev:1977qv}.} The constant $S_{\De,\ell}$ depends on $\De,\ell$ and the external dimensions $\De_i$, and will not be important for the current discussion. We call the $F_{\De,\ell}$ ``Euclidean partial waves." Orthogonality follows from the fact that the Casimir operator is self-adjoint with respect to $\<\cdot,\cdot\>$, together with the fact that $F_{\De,\ell}$ is single-valued so there are no boundary contributions from integrating by parts. See \cite{Caron-Huot:2017vep} for more details.
 
A four-point function of scalars has a Euclidean partial wave decomposition of the form
\be
\label{eq:euclideandecomposition}
g(x_i) &= 1+\sum_\ell \oint_{\frac d 2 - i\oo}^{\frac d 2 + i\oo} \frac{d\De}{2\pi i} c(\De,\ell) F_{\De,\ell}(x_i) + \textrm{discrete series}.
\ee
The decomposition (\ref{eq:euclideandecomposition}) is not the usual conformal block decomposition, but it is closely related. When $g(x_i)$ is a four-point function in a unitary CFT, we expect that $c(\De,\ell)$ has (shadow-symmetric) simple poles in $\De$ on the real axis
\be
c(\De,\ell) &\sim \sum_i -c_{\De_i,\ell}\p{\frac{1}{\De-\De_i} + \frac{S_{\De,\ell}^{-1}}{d-\De-\De_i}}.
\ee
 We can then deform the $\De$-contour in (\ref{eq:euclideandecomposition}) to the right for $G_{\De,\ell}$ and to the left for $G_{d-\De,\ell}$ to obtain
\be
g(x_i) &= 1 + \sum_{\De_i,\ell} c_{\De_i,\ell} G_{\De_i,\ell}(x_i).
\ee
Thus, positions of poles in $c(\De,\ell)$ encode the spectrum of the theory, and the residues encode products of OPE coefficients.\footnote{When deforming the $\De$-contour, one must take into account poles in the blocks themselves, which interact in an intricate way \cite{Dobrev:1977qv,Cornalba:2007fs,Caron-Huot:2017vep}.}

For spinning operators, the Euclidean partial waves $F^{(a,b)}_{\De,\rho}$ and their coefficients $c_{(a,b)}(\De,\rho)$ are additionally labeled by a pair of three-point structures $(a,b)$. An inversion formula for the coefficients is given by\footnote{We sum over raised and lowered pairs of three-point structures $(a,b)$.}
\be
M^{(c,d)(a,b)}(\De,\rho) c_{(a,b)}(\De,\rho) &= \<\tl F^{(c,d)}_{\De,\rho},g\>,
\ee
where, roughly speaking,\footnote{We are neglecting an additional term proportional to $\de(\De+\De'-d)$ that is unimportant for the current discussion.}
\be
\<\tl F_{\De',\rho'}^{(c,d)}, F_{\De,\rho}^{(a,b)}\> &\sim M^{(c,d)(a,b)}(\De,\rho) \de_{\rho \rho'} \de(\De-\De').
\ee
Pictorially,
\be
\label{eq:pictorialinversionspin}
\<\tl F_{\De,\rho}^{(c,d)},g\> &=\diagramEnvelope{\begin{tikzpicture}[anchor=base,baseline]
	\node (f1) at (-2.5,0) [threept] {$c$};
	\node (f2) at (-0.5,0) [threept] {$d$};
	\node (fm) at (-1.5,0) [cross] {};
	\node (g) at (2.5,0) [threept] {$g$};
	\node (l) at (-1.05,0.1) [above] {$\cO$};
	\node (l) at (-1.9,0.1) [above] {$\cO^\dag$};
	\node (o1) at (-0.3,-0.4) [below] {$\cO_1$};
	\node (o2) at (-0.3,0.4) [above] {$\cO_2$};
	\node (o3) at (-3.2,0.4) [above] {$\cO_3$};
	\node (o4) at (-3.2,-0.4) [below] {$\cO_4$};
	\draw [spinning] (g) to[out=135,in=45] (f2);
	\draw [spinning] (g) to[out=-135,in=-45] (f2);
	\draw [spinning] (g) to[out=45,in=135] (f1);
	\draw [spinning] (g) to[out=-45,in=-135] (f1);
	\draw [spinning] (f1) -- (fm);
	\draw [spinning] (f2) -- (fm);
\end{tikzpicture}}.
\ee

One of our main observations is that spinning conformal blocks can be written as derivatives of scalar blocks. Schematically, we have
\be
\tl F_{\De,\rho}^\textrm{spin} &= \mathfrak{D} F^\textrm{scalar}_{\De,\ell},\nn\\
\mathfrak{D} &= \sum_t d_t(\De,\rho) t_{ABCD} \cD^{(a)A}_{1} \cD^{(b)B}_{2} \cD^{(c)C}_3 \cD^{(d)D}_4.
\ee
The operators $\cD^{(a_i)A_i}_i$ are spin-raising operators transforming in $W_i$, acting on the point $x_i$. 
Here, $t$ runs over invariant tensors in $(W_1\otimes W_2 \otimes W_3\otimes W_4)^*$.

To compute the pairing (\ref{eq:pictorialinversionspin}), it is useful to integrate $\mathfrak{D}$ by parts,
\be
\label{eq:afterintegratingbyparts}
\<\tl F^\textrm{spin}_{\De,\rho},g\> &= \<F^\textrm{scalar}_{\De,\ell},\mathfrak{D}^* g\>
\ee
where $\mathfrak{D}^*$ is the adjoint of $\mathfrak{D}$ under the pairing $\<\cdot,\cdot\>$, given by replacing each $\cD_i^{(a)}$ with its adjoint $(\cD_i^{(a)})^*$ (since we can integrate by parts individually on each leg). The adjoints $(\cD^{(a_i)}_i)^*$ are spin-{\it lowering}\ differential operators, and the right-hand side of (\ref{eq:afterintegratingbyparts}) is a pairing between {\it scalar}\ four-point functions.  We can thus proceed to study it in the same way as we study four-point functions of scalars. For example, one can derive spinning versions of the CFT Froissart-Gribov formula \cite{Caron-Huot:2017vep} using these techniques.\footnote{One of the consequences of the Froissart-Gribov formula is that CFT data can be analytically continued in spin. When non-STTs can appear as internal operators, analytic continuation in spin can be understood by expressing $V_{\De,\rho}$ as a subrepresentation of $V_{\De',\ell}\otimes W$ for some fixed $W$, and then analytically continuing in $\ell$. This is equivalent to analytically continuing in the length of the first row of the Young diagram for $\rho$.} We call this trick ``spinning-down" a four-point function.

In pictures, the right-hand side of (\ref{eq:afterintegratingbyparts}) is
\be
\<\tl F^\textrm{spin}_{\De,\rho},g\> &\sim \sum_t
\diagramEnvelope{\begin{tikzpicture}[anchor=base,baseline]
	\node (f1) at (-2.5,0) [twopt] {};
	\node (f2) at (-0.5,0) [twopt] {};
	\node (g) at (3,0) [threept] {$g$};
	\node (o1) at (2,-0.8) [threept] {};
	\node (o2) at (2,0.8) [threept] {};
	\node (o3) at (4,0.8) [threept] {};
	\node (o4) at (4,-0.8) [threept] {};
	\draw [scalar] (o2) to[out=180,in=45] (f2);
	\draw [scalar] (o1) to[out=180,in=-45] (f2);
	\draw [scalar] (o3) to[out=135,in=135] (f1);
	\draw [scalar] (o4) to[out=-135,in=-135] (f1);
	\draw [spinning] (g) -- (o1);
	\draw [spinning] (g) -- (o2);
	\draw [spinning] (g) -- (o3);
	\draw [spinning] (g) -- (o4);
	\draw [spinning] (f1) -- (f2);
	\node (t) at (3,-2) [threept] {$t$};
	\draw [finite] (o1) -- (t);
	\draw [finite] (o2) -- (t);
	\draw [finite] (o3) -- (t);
	\draw [finite] (o4) -- (t);
\end{tikzpicture}},
\label{eq:somepicture}
\ee
where the dashed lines represent scalars.

\subsection{$6j$ symbols for infinite-dimensional representations}

If we plug in a $t$-channel partial wave for $g$, then we can simplify (\ref{eq:somepicture}) further by using crossing to move the differential operators to the internal leg:
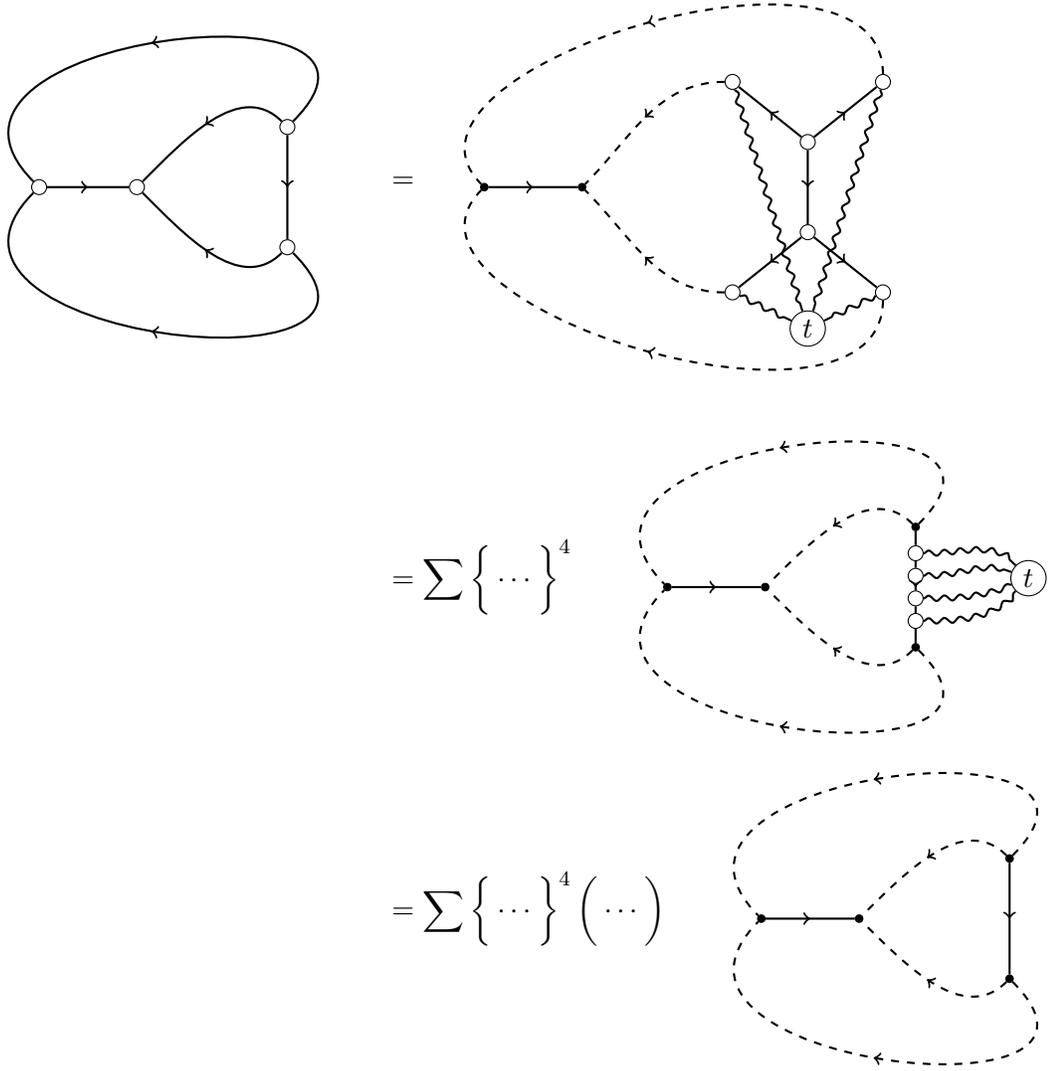
\be
\diagramEnvelope{\begin{tikzpicture}[anchor=base,baseline]
	\node (f1) at (-1.3,0) [threept] {};
	\node (f2) at (0,0) [threept] {};
	\node (g1) at (2,0.8) [threept] {};
	\node (g2) at (2,-0.8) [threept] {};
	\draw [spinning] (g2) to[out=-135,in=-45] (f2);
	\draw [spinning] (g1) to[out=135,in=45] (f2);
	\draw [spinning] (g1) to[out=45,in=135,looseness=2] (f1);
	\draw [spinning] (g2) to[out=-45,in=-135,looseness=2] (f1);
	\draw [spinning] (f1) -- (f2);
	\draw [spinning] (g1) -- (g2);
\end{tikzpicture}
\hspace{-0.3in}&=\hspace{-0.3in}
\begin{tikzpicture}[anchor=base,baseline]
	\node (f1) at (-1.3,0) [twopt] {};
	\node (f2) at (0,0) [twopt] {};
	\node (g1) at (3,0.6) [threept] {};
	\node (g2) at (3,-0.6) [threept] {};
	\node (o1) at (2,-1.4) [threept] {};
	\node (o2) at (2,1.4) [threept] {};
	\node (o3) at (4,1.4) [threept] {};
	\node (o4) at (4,-1.4) [threept] {};
	\draw [scalar] (o2) to[out=180,in=45] (f2);
	\draw [scalar] (o1) to[out=180,in=-45] (f2);
	\draw [scalar] (o3) to[out=90,in=135] (f1);
	\draw [scalar] (o4) to[out=-90,in=-135] (f1);
	\draw [spinning] (g2) -- (o1);
	\draw [spinning] (g1) -- (o2);
	\draw [spinning] (g1) -- (o3);
	\draw [spinning] (g2) -- (o4);
	\draw [spinning] (f1) -- (f2);
	\draw [spinning] (g1) -- (g2);
	\node (t) at (3,-2) [threept] {$t$};
	\draw [finite] (o1) -- (t);
	\draw [finite] (o2) -- (t);
	\draw [finite] (o3) -- (t);
	\draw [finite] (o4) -- (t);
\end{tikzpicture}
\nn\\[-0.5in]
&=\sum\left\{\phantom{\frac{1}{2}}\!\!\!\!\cdots\phantom{\frac{1}{2}}\!\!\!\!\right\}^4
\hspace{-0.3in}
\begin{tikzpicture}[anchor=base,baseline]
	\node (f1) at (-1.3,0) [twopt] {};
	\node (f2) at (0,0) [twopt] {};
	\node (g1) at (2,0.8) [twopt] {};
	\node (g2) at (2,-0.8) [twopt] {};
	\node (t) at (3.5,0) [threept] {$t$};
	\node (l1) at (2,0.45) [threept] {};
	\node (l2) at (2,0.15) [threept] {};
	\node (l3) at (2,-0.15) [threept] {};
	\node (l4) at (2,-0.45) [threept] {};
	\draw [scalar] (g2) to[out=-135,in=-45] (f2);
	\draw [scalar] (g1) to[out=135,in=45] (f2);
	\draw [scalar] (g1) to[out=45,in=135,looseness=2] (f1);
	\draw [scalar] (g2) to[out=-45,in=-135,looseness=2] (f1);
	\draw [spinning] (f1) -- (f2);
	\draw [spinning] (g1) -- (l1) -- (l2) -- (l3) -- (l4) -- (g2);
	\draw [finite] (l1) to[out=0,in=135] (t);
	\draw [finite] (l2) to[out=0,in=160] (t);
	\draw [finite] (l3) to[out=0,in=-160] (t);
	\draw [finite] (l4) to[out=0,in=-135] (t);
\end{tikzpicture}
\nn\\[-0.5in]
&=\sum\left\{\phantom{\frac{1}{2}}\!\!\!\!\cdots\phantom{\frac{1}{2}}\!\!\!\!\right\}^4
\left(\phantom{\frac{1}{2}}\!\!\!\!\cdots\phantom{\frac{1}{2}}\!\!\!\!\right)
\hspace{-0.3in}
\begin{tikzpicture}[anchor=base,baseline]
	\node (f1) at (-1.3,0) [twopt] {};
	\node (f2) at (0,0) [twopt] {};
	\node (g1) at (2,0.8) [twopt] {};
	\node (g2) at (2,-0.8) [twopt] {};
	\draw [scalar] (g2) to[out=-135,in=-45] (f2);
	\draw [scalar] (g1) to[out=135,in=45] (f2);
	\draw [scalar] (g1) to[out=45,in=135,looseness=2] (f1);
	\draw [scalar] (g2) to[out=-45,in=-135,looseness=2] (f1);
	\draw [spinning] (f1) -- (f2);
	\draw [spinning] (g1) -- (g2);
\end{tikzpicture}}
\label{eq:generalsixjintermsofscalarsixj}
\ee
The symbol $\{\cdots\}^4$ represents a product of four $6j$ symbols of the type in (\ref{eq:6jdefinition}), and the factor $(\cdots)$ is the result of taking a conformally-invariant product of differential operators on the right internal leg. For simplicity, we have omitted labels and shown only the topology of the various diagrams. Dashed lines represent scalar operators, and solid lines represent operators with spin. 

Equation (\ref{eq:generalsixjintermsofscalarsixj}) expresses an inner product of general spinning blocks in terms of inner products of scalar blocks. Such inner products are examples of $6j$ symbols for the conformal group, where all the representations are infinite-dimensional principal series representations. The corresponding graphs have the topology of a tetrahedron. The equality (\ref{eq:generalsixjintermsofscalarsixj}) is an example of a general set of relations between infinite-dimensional $6j$ symbols that we can derive as follows. We start with a tetrahedron graph and introduce a bubble with a finite-dimensional representation $W$ on one of the lines. We can then move the vertices of the bubble to a different internal line and collapse it.
\be
&
\diagramEnvelope{\begin{tikzpicture}[anchor=base,baseline]
\node (a) at (-2.00,0) [twopt] {};
\node (b) at (1.00, 1.732) [twopt] {};
\node (c) at (1.00, -1.732) [twopt] {};
\node (d) at (0,0) [twopt] {};
\draw [spinning no arrow] (b) -- (a);
\draw [spinning no arrow] (c) -- (a);
\draw [spinning no arrow] (a) -- (d);
\draw [spinning no arrow] (b) -- (c);
\draw [spinning no arrow] (b) -- (d);
\draw [spinning no arrow] (c) -- (d);
\end{tikzpicture}\nn\\
&=
\left(\phantom{\frac{1}{2}}\!\!\!\!\cdots\phantom{\frac{1}{2}}\!\!\!\!\right)
\begin{tikzpicture}[anchor=base,baseline]
\node (a) at (-2.00,0) [twopt] {};
\node (b) at (1.00, 1.732) [twopt] {};
\node (c) at (1.00, -1.732) [twopt] {};
\node (d) at (0,0) [twopt] {};
\node (w1) at (1,0.7) [twopt] {};
\node (w2) at (1,-0.7) [twopt] {};
\draw [spinning no arrow] (b) -- (a);
\draw [spinning no arrow] (c) -- (a);
\draw [spinning no arrow] (a) -- (d);
\draw [spinning no arrow] (b) -- (c);
\draw [spinning no arrow] (b) -- (d);
\draw [spinning no arrow] (c) -- (d);
\draw [finite] (w1) to[out=0,in=0] (w2);
\end{tikzpicture}
=
\sum\left(\phantom{\frac{1}{2}}\!\!\!\!\cdots\phantom{\frac{1}{2}}\!\!\!\!\right)\left\{\phantom{\frac{1}{2}}\!\!\!\!\cdots\phantom{\frac{1}{2}}\!\!\!\!\right\}^2
\begin{tikzpicture}[anchor=base,baseline]
\node (a) at (-2.00,0) [twopt] {};
\node (b) at (1.00, 1.732) [twopt] {};
\node (c) at (1.00, -1.732) [twopt] {};
\node (d) at (0,0) [twopt] {};
\node (w1) at (0.7, 1.212) [twopt] {};
\node (w2) at (0.7, -1.212) [twopt] {};
\draw [spinning no arrow] (b) -- (a);
\draw [spinning no arrow] (c) -- (a);
\draw [spinning no arrow] (a) -- (d);
\draw [spinning no arrow] (b) -- (c);
\draw [spinning no arrow] (b) -- (d);
\draw [spinning no arrow] (c) -- (d);
\draw [finite] (w1) to[out=-90,in=90] (w2);
\end{tikzpicture}
\nn\\
&=
\sum\left(\phantom{\frac{1}{2}}\!\!\!\!\cdots\phantom{\frac{1}{2}}\!\!\!\!\right)\left\{\phantom{\frac{1}{2}}\!\!\!\!\cdots\phantom{\frac{1}{2}}\!\!\!\!\right\}^3
\begin{tikzpicture}[anchor=base,baseline]
\node (a) at (-2.00,0) [twopt] {};
\node (b) at (1.00, 1.732) [twopt] {};
\node (c) at (1.00, -1.732) [twopt] {};
\node (d) at (0,0) [twopt] {};
\node (w1) at (0.7, 1.212) [twopt] {};
\node (w2) at (0.2, 0.346) [twopt] {};
\draw [spinning no arrow] (b) -- (a);
\draw [spinning no arrow] (c) -- (a);
\draw [spinning no arrow] (a) -- (d);
\draw [spinning no arrow] (b) -- (c);
\draw [spinning no arrow] (b) -- (d);
\draw [spinning no arrow] (c) -- (d);
\draw [finite] (w1) to[out=-70,in=-30] (w2);
\end{tikzpicture}
=
\sum
\left(\phantom{\frac{1}{2}}\!\!\!\!\cdots\phantom{\frac{1}{2}}\!\!\!\!\right)^2
\left\{\phantom{\frac{1}{2}}\!\!\!\!\cdots\phantom{\frac{1}{2}}\!\!\!\!\right\}^4
\begin{tikzpicture}[anchor=base,baseline]
\node (a) at (-2.00,0) [twopt] {};
\node (b) at (1.00, 1.732) [twopt] {};
\node (c) at (1.00, -1.732) [twopt] {};
\node (d) at (0,0) [twopt] {};
\draw [spinning no arrow] (b) -- (a);
\draw [spinning no arrow] (c) -- (a);
\draw [spinning no arrow] (a) -- (d);
\draw [spinning no arrow] (b) -- (c);
\draw [spinning no arrow] (b) -- (d);
\draw [spinning no arrow] (c) -- (d);
\end{tikzpicture}}
\label{eq:tetrahedra}
\ee

The above is essentially the pentagon identity for a mixture of finite-dimensional (degenerate) and infinite-dimensional representations. Because the crossing kernel for degenerate four-point functions is so simple, the pentagon identity becomes a useful tool for computing infinite-dimensional crossing kernels. The $6j$ symbol for six scalar representations of the conformal group was computed in \cite{Krasnov:2005fu} in terms of a four-fold Mellin-Barnes integral. That result, along with relations of the type illustrated in (\ref{eq:tetrahedra}) in principle allows one to compute an arbitrary $6j$ symbol.

\section{Discussion}
\label{sec:discussion}

In this work, we introduced new mathematical tools for computations in conformal representation theory. These include the construction of weight-shifting operators summarized in theorem~\ref{thm:operatorclassification}, the observation that they satisfy the crossing equation~(\ref{eq:6jdefinition}), and our discussion of how weight-shifting operators interact with conformally-invariant projectors (\ref{eq:integrationbyparts}). For concrete computations, we introduced the embedding space operators (\ref{eq:vectoroperators}), (\ref{eq:3Doperators}), and (\ref{eq:4Doperators}). We  explored in detail how these tools can be applied to compute conformal blocks. We also discussed some applications to harmonic analysis and inversion formulae. We plan to expand on the latter in future work \cite{ShadowFuture}.

However, many directions remain unexplored. One natural question is how weight-shifting operators interact with short multiplets of the conformal group. For simplicity, we specialized to simple generalized Verma modules (long multiplets) in this paper. However, we expect new phenomena in the presence of shortening conditions. Some questions include: How is the tensor product decomposition~\ref{eq:tensorproductdecomposition} modified for short multiplets? How are shortening conditions reflected in the zeros and poles of $6j$ symbols? Is the spinning-down procedure of section~\ref{sec:spinningdown} useful when external operators are in short multiplets?

Our construction of weight-shifting operators and their crossing equations is very general. As noted in the introduction, it also applies to generalized Verma modules of any Lie (super-)algebra.\footnote{In the language of~\cite{Gadde:2017sjg}, it works in a $G$FT for any group $G$.} In particular, supersymmetric weight-shifting operators should be useful for computing and studying superconformal blocks and tensor structures. It will be interesting to construct such operators and explore their applications. The question of how weight-shifting operators interact with shortening conditions becomes even more interesting in the superconformal case, since there are a wide variety of interesting short superconformal multiplets (see e.g.\ \cite{Cordova:2016emh}).

As discussed in section~\ref{sec:algebraofweightshifting}, the algebra of weight-shifting operators is governed by the fusion matrix $J(\l)$, which is closely related to solutions to the Yang-Baxter equation and integrability \cite{YangBaxter}. Does this structure have an interesting role to play in conformal field theory? Is it related to the ``superintegrability" of conformal blocks discussed in \cite{Isachenkov:2016gim,Chen:2016bxc,Schomerus:2016epl}?

It may also be interesting to explore the role of weight-shifting operators in holographic calculations.\footnote{After this draft was completed, \cite{Cardoso:2017qmj} appeared which describes the special case of weight-shifting operators that shift the mass of a scalar field in $AdS$.} We expect that they should help in the computation of Witten diagrams for operators with spin. Natural questions include: What is the flat-space limit of weight-shifting operators? Are they useful for amplitudes calculations (for example are they related to the differential operators introduced in \cite{Cheung:2017ems})? Weight-shifting operators may also be helpful for exploring spinning amplitudes in the conformal basis of \cite{Pasterski:2017kqt}.

\section*{Acknowledgements}

We are grateful to Clay C\'ordova, Tolya Dymarsky, Abhijit Gadde, Mikhail Isachenkov, Eric Perlmutter, Fernando Rejon-Barrera, Douglas Stanford and Emilio Trevisani for discussions. DK and PK would like to thank the organizers of the Bootstrap 2017 workshop, where part of this work was completed.
DSD is supported by DOE grant DE-SC0009988, a William D. Loughlin Membership at the
Institute for Advanced Study, and Simons Foundation grant 488657 (Simons Collaboration
on the Nonperturbative Bootstrap). PK is supported DOE grant DE-SC0011632.

\pagebreak

\appendix

\section{Conformal algebra}\label{app:conformalalgebra}
We use the following conventions for the conformal algebra,
\be
	[D,K_\mu]&=-K_\mu,\quad [D,P_\mu]=P_\mu,\\
	[K_\mu,P_\nu]&=2\delta_{\mu\nu}D-2M_{\mu\nu},\label{eq:KPcommutator}\\
	[M_{\mu\nu},P_{\r}]&=\delta_{\nu\r}P_\mu-\delta_{\mu\r}P_\nu,\\
	[M_{\mu\nu},K_{\r}]&=\delta_{\nu\r}K_\mu-\delta_{\mu\r}K_\nu,\\
	[M_{\mu\nu},M_{\r\s}]&=\delta_{\nu\r}M_{\mu\s}-\delta_{\mu\r}M_{\nu\s}+\delta_{\nu\s}M_{\r\mu}-\delta_{\mu\s}M_{\r\nu},
\ee
and all other commutators vanish. In Lorentzian signature, all generators are anti-Hermitian. In Euclidean signature $D=D^\dagger$, $K=P^\dagger$ and $M$ is anti-Hermitian. Notice how~\eqref{eq:KPcommutator} expresses the conformal Killing equation for the adjoint representation by saying that the rank-2 symmetric traceless tensor does not appear among level-1 descendants of the primary $K_\mu$.

\section{Verma modules and differential operators}
\label{app:verma}

In the main text we have seen that for every irreducible component $V_{\De',\rho}$ in the tensor product $W\otimes V_{\De,\rho}$ there is a conformally-covariant differential operator $\cD_A : [\De,\rho]\to [\De',\lambda]$ with a $W^*$-index $A$. Here we would like to state this relation more carefully and show that there is in fact a one-to-one correspondence.

\begin{theorem}
For generic $\De$ the decomposition~\eqref{eq:tensorproductdecomposition} holds. The irreducible components in the tensor product decomposition~\eqref{eq:tensorproductdecomposition} are in one-to-one correspondence with the conformally-covariant differential operators $\cD_A: [\De,\rho]\to[\De',\lambda]$ with an index $A$ transforming in a finite-dimensional representation $W$ of $\SO(d+1,1)$.
\end{theorem}
\begin{proof}
First we show that the tensor product decomposition~\eqref{eq:tensorproductdecomposition} holds. The discussion in section~\ref{sec:tensorproducts} essentially shows that the characters on the both sides agree. This statement holds for all $\De$. This however does not necessarily imply~\eqref{eq:tensorproductdecomposition} as an isomorphism between the representations. So our first step is to construct the isomorphism~\eqref{eq:tensorproductdecomposition}.
 
We can define on $W\otimes V_{\De,\rho}$ a conformally-invariant inner product, induced from the inner products on $W$ and $V_{\De,\rho}$. Suppose that there is a submodule $Y\subseteq W\otimes V_{\De,\rho}$. If the conformally-invariant inner product is non-degenerate\footnote{Note that if the inner-product is non-degenerate but not positive-definite, there still can exist subspaces on which it is degenerate. Finite-dimensional representations of non-compact groups such as $\SO(d+1,1)$ or $\SO(d,2)$ necessarily have indefinite inner products.} on $Y$, it follows that $Y$ is in fact a direct summand,
\be\label{eq:directsummand}
	M\equiv W\otimes V_{\De,\rho} = Y\oplus Y^\perp.
\ee
Starting from this observation, it is a standard argument to show that~\eqref{eq:tensorproductdecomposition} holds. We reproduce it here for completeness. The states~\eqref{eq:simplestensortprimary} are always primary because they have the smallest possible scaling dimension $\De-j$. We can decompose them into mutually orthogonal irreducibles of $\SO(d)$. Considering all the descendants of these states we form the submodule
\be
Y_{-j} = \bigoplus_{\lambda \in W_{-j}\otimes \rho} V_{\De-j,\lambda}.
\ee
For generic $\De$ the generalized Verma modules in this sum are irreducible, and thus the inner product is non-degenerate (otherwise the null states form a submodule). By \eqref{eq:directsummand} we then have 
\be
	M = Y_{-j}\oplus M_{1}, \quad M_1\equiv Y_{-j}^\perp.
\ee
We can now look at the states of the smallest scaling dimension inside of $M_0$. These all are again primary, and we can consider the submodule $Y_{-j+1}$ which they generate. Since we already know~\eqref{eq:tensorproductdecomposition} as a character identity, we know that 
\be
Y_{-j+1} = \bigoplus_{\lambda \in W_{-j+1}\otimes \rho} V_{\De-j+1,\lambda}.
\ee
Again, from~\eqref{eq:directsummand} we find
\be
	M_1 = Y_{-j+1}\oplus M_2.
\ee
We then continue recursively until we exhaust all states as controlled by~\eqref{eq:tensorproductdecomposition} as a character identity. Collecting everything together, we arrive at~\eqref{eq:tensorproductdecomposition} as a direct sum decomposition.

From the discussion in the main text it follows that the primaries which we identify in the tensor product $W\otimes V_{\De,\rho}$ give rise to conformally-covariant differential operators. At the same time, as observed in section~\ref{sec:algebraofweightshifting}, they give rise to homomorphisms~\eqref{eq:Vermahomomorphism}. In fact, there is a one-to-one correspondence between these objects.
\blem\label{lem:tautology}
For any fixed $\De,\De',\rho,\lambda$ the conformally-covariant differential operators $\cD_A:[\De,\rho]\to[\De',\lambda]$ are in a one-to-one correspondence with the homomorphisms of the form~\eqref{eq:Vermahomomorphism}.
\elem
\noindent The map implied by this lemma is essentially constructed in section~\ref{sec:diffopsfromtensor}. Looking at it one can easily convince oneself that the lemma is almost a tautology. We give a formal proof later in this appendix. 
 
Given lemma~\ref{lem:tautology}, to finish the proof of the theorem it only remains to show that generically the only homomorphisms of the form~\eqref{eq:Vermahomomorphism} are those which come from the embeddings of the direct summands in~\eqref{eq:tensorproductdecomposition}. This follows immediately from Schur's lemma and the fact that Verma modules are irreducible for generic scaling dimensions.
\end{proof}

\begin{proof}[Proof of lemma~\ref{lem:tautology}]
For $W=\bullet$, lemma~\ref{lem:tautology} is standard material in representation theory of generalized Verma modules~\cite{MR1463509}, and we need to only slightly modify it by introducing the non-trivial $W$. Let us give an elementary review of the proof with the appropriate modifications. 

First, we need to give the precise meaning to $[\De,\rho]$, which is in fact a vector bundle. The sections of $[\De,\rho]$ are the functions $f^a(x)$ on the conformal sphere $S^d$ with index $a$ in $\rho$ which transform as\footnote{The difference with~\eqref{eq:primarytransformation} comes from the fact that here we are defining the action on functions rather than operators, and the appearance of $g^{-1}$ in the argument of $f$ on the right hand side is dictated by compatibility with the group multiplication $(gh)f=g(h(f))$.}
\be
	(gf)^a(x) = \Omega(x)^{-\Delta}\rho^{a}{}_{b}(R(x))f^b(g^{-1}x), \quad g\in \SO(d+1,1). 
\ee
We also associate a vector bundle $\cW$ to $W$. The sections of $\cW$ are the functions $f_A$ which transform as
\be
	(gf)_A(x) = D_A{}^B(g)f_B(g^{-1}x).
\ee
The conformally-covariant differential operator $\cD_A$ is then a differential operator between the vector bundles
\be
	\cD: [\rho,\lambda]\to \cW\otimes [\De',\lambda],
\ee
which commutes with the action of the conformal group. We will refer to this property as equivariance. The idea now is to note that if we know that $\cD$ is equivariant, then it is completely specified by its action at zero, i.e.\ by the expression
\be\label{eq:Dat0}
	(\cD f)^a_A(0)=\text{derivatives of }f\text{ at 0}.
\ee
Indeed, let $t_x$ be the translation which takes $0$ to $x$. Then we can compute $\cD f$ at any $x$ by writing
\be\label{eq:Dat0extension}
	(\cD f)^a_A(x) = (t_x \cD t_{-x}f)^a_A(x)=D_A{}^B(t_x) (\cD t_{-x}f)^a_B(0),
\ee
and using~\eqref{eq:Dat0} for $t_{-x}f$. As usual, the only condition the expression~\eqref{eq:Dat0} has to satisfy in order for this construction to be self-consistent is that it has to be equivariant with respect to the transformations which fix the origin -- in our case with respect to dilatations, rotations and special conformal transformations, the algebra of which we will denote by $\mathfrak{p}$.\footnote{This is not to be confused with the subalgebra generated by translations. We use this notation to be consistent with the mathematics literature, where $\mathfrak{p}$ stands for ``parabolic''.} 

Instead of studying this condition in detail, we can just map it to the similar problem for Verma modules. If $\cD$ is of order $k$, the equation \eqref{eq:Dat0} can be understood as the map 
\be\label{eq:kjetmap}
	\cD: J^k_0[\De,\rho]\to W\otimes J^0_0[\De',\lambda],
\ee
where $J_0^k[\De,\rho]$ is the space of $k$-jets of sections of $[\De,\rho]$ at 0, i.e.\ the space of formal power series of sections of $[\De,\rho]$ around the origin, truncated to $k$-th order. One can extend the action of conformal algebra to these jets, and the problem of finding a $\mathfrak{p}$-equivariant map~\eqref{eq:Dat0} is equivalent to finding $\mathfrak{p}$-equivariant maps~\eqref{eq:kjetmap}. Using \eqref{eq:Dat0extension} we can extend such maps to $\mathfrak{so}(d+1,1)$-equivariant maps
\be\label{eq:inftyjetmap}
	\cD: J^\infty_0[\De,\rho]\to W\otimes J^\infty_0[\De',\lambda].
\ee
between the formal power series. These are the same as Verma module homomorphisms because $V_{\De,\rho}$ consists of the states like $\ptl_{\mu_1}\cdots\ptl_{\mu_n}\cO^a(0)$, which are naturally linear functionals on the formal power series $J^\infty_0[\De,\rho]$. In fact, one can show that as $\mathfrak{so}(d+1,1)$-representations,
\be
	V_{\De,\rho} \simeq \p{J^\infty_0[\De,\rho]}^*.
\ee
Thus by taking the dual of~\eqref{eq:inftyjetmap} we obtain a homomorhism
\be
	\cD^*: W^*\otimes V_{\De',\lambda}\to V_{\De,\rho}.
\ee
As usual, we can replace $W^*$ on the left with a $W$ on the right: we can define $\cD'(v) = e^A\otimes \cD^*(e^*_A\otimes v)$, so that
\be
	\cD': V_{\De',\lambda}\to  W \otimes V_{\De,\rho}
\ee
is a homomorphism of the form~\eqref{eq:Vermahomomorphism}. All the steps that we took to get from the differential operator $\cD_A$ to $\cD'$ were invertible, so we get a one-to-one correspondence.
\end{proof}

\section{Weight-shifting operators for the vector representation}
\label{app:D4coefficients}

Let us give more detail about the computation of the weight-shifting operators for the vector representation (\ref{eq:vectoroperators}). Recall that traceless symmetric tensor operators are homogeneous elements of $R/(R\cap I)$, where $R$ is the ring of functions of $X,Z\in \R^{d+1,1}$ that are invariant under $Z\to Z+\l X$ (equivalently they are killed by $X\.\pdr{}{Z}$), and $I$ is the ideal generated by $\{X^2,X\.Z,Z^2\}$.  For a differential operator $\cD$ to be well-defined on $R/(R\cap I)$, it must satisfy
\be
\label{eq:consistencyconditions}
\cD R &\subseteq R,\\
\cD (R\cap I) &\subseteq R\cap I.
\label{eq:consistencyconditionstoo}
\ee

Because we are searching for homogeneous differential operators, it suffices to consider their action on homogeneous elements of $R$. It is not hard to convince oneself that a general homogeneous element of $R$ can be written as a linear combination of functions of the form
\be
f_{\De,\ell}(X,Z) &\equiv (X\.Y)^{-\De-\ell}((Z\.P)(X\.Q) - (Z\.Q)(X\.P))^\ell,
\ee
for various $Y,P,Q$.

To find the weight-shifting operators $\cD^{(a)}_m$, we start by enumerating conformally-covariant terms with the correct homogeneity in $X$ and $Z$, modulo $X\.\pdr{}{X}$, $Z\.\pdr{}{Z}$, and $X\.\pdr{}{Z}$ (which act as $-\De$, $\ell$, and $0$, respectively). There are a finite number of such terms, and this leads to the ansatz (\ref{eq:vectoroperators}) with undetermined coefficients that are functions of $\De,\ell$.

To fix the coefficients, it suffices to check (\ref{eq:consistencyconditions}) and (\ref{eq:consistencyconditionstoo}) for a sufficient number of functions. In particular, we impose (\ref{eq:consistencyconditions}) in the form
\be
\p{X\.\pdr{}{Z}}\cD^{(a)}_m f_{\De,\ell}(X,Z) &= 0,
\ee
and (\ref{eq:consistencyconditionstoo}) in the form
\be
\cD^{(a)}_m((S\.Z)^2(X\.X)-2(S\.X)(S\.Z)(X\.Z) + (S\.X)^2 (Z\.Z))f_{\De+2,\ell-2}(X,Z) &\in R\cap I\nn\\
\cD^{(a)}_m((X\.X)(Z\.Z) - (X\.Z)^2) f_{\De+2,\ell-2}(X,Z) &\in R\cap I\nn\\
\cD^{(a)}_m((X\.X)(Z\.S) - (X\.S)(X\.Z)) f_{\De+2,\ell-1}(X,Z) &\in R\cap I\nn\\
\cD^{(a)}_m(X\.X) f_{\De+2,\ell}(X,Z) &\in R\cap I
\label{eq:specializedconsistencyconditions}
\ee
where $S,Y,P,Q\in\R^{d+1,1}$ are arbitrary vectors. (Because of (\ref{eq:consistencyconditions}), to check whether the left hand sides of (\ref{eq:specializedconsistencyconditions}) are in $R\cap I$, it suffices to check whether they are in $I$.  That is, we set $X^2,X\.Z,Z^2$ to zero and check whether the result is zero.) These conditions are sufficient to fix the unknown coefficients. In particular, for the most complicated weight-shiting operator $\cD^{+0}_m$, we find
\be
c_1 &= \left(\frac{d}{2}-\Delta -1\right) (\Delta +\ell-1) (d-\Delta +\ell-2)\nn\\
c_2 &= -\frac{1}{2} (\Delta +\ell-1) (d-\Delta +\ell-2)\nn\\
c_3 &= -\left(\frac{d}{2}-\Delta -1\right) (\Delta +\ell-2)\nn\\
c_4 &= -\left(\frac{d}{2}-\Delta -1\right) (d-\Delta +\ell-2)\nn\\
c_5 &= \frac{d}{2}+\ell-2\nn\\
c_6 &= \frac{d}{2}-\Delta -1\nn\\
c_7 &= -\frac{1}{2}.
\ee

\section{$6j$ symbols and the algebra of operators}
\label{app:sixjalgebra}
In this appendix we consider the crossing equation which is obtained by replacing $\cO_1$ in~\eqref{eq:6jdefinition} by a finite-dimensional representation,\footnote{We then find a third finite-dimensional representation arising from the tensor product of the first two.}
\be\label{eq:algebracrossing}
\diagramEnvelope{\begin{tikzpicture}[anchor=base,baseline]
	\node (vertL) at (0,-0.08) [threept] {$a$};
	\node (vertR) at (2,-0.133) [threept] {$b$};
	\node (opO1) at (-0.5,1) [above] {$U$};
	\node (opO2) at (-1,0) [left] {$\cO_2$};
	\node (opO3) at (3,0) [right] {$\cO_3$};
	\node (opW) at (2.5,1) [above] {$V$};	
	\node at (1,0.1) [above] {$\cO_3'$};	
	\draw [finite with arrow] (vertL)-- (opO1);
	\draw [spinning] (opO2)-- (vertL);
	\draw [spinning] (vertL)-- (vertR);
	\draw [spinning] (vertR)-- (opO3);
	\draw [finite with arrow] (vertR)-- (opW);
	\end{tikzpicture}}
\quad=\quad
\sum_{W,m,n}
\left\{
\begin{matrix}
	U & \cO_2 & W \\
	\cO_3 & V & \cO_3'
\end{matrix}
\right\}^{ab}_{mn}
\diagramEnvelope{\begin{tikzpicture}[anchor=base,baseline]
	\node (vertU) at (0,-0.78) [threept,inner sep=1.5pt] {$m$};
	\node (vertD) at (0,0.7) [threept] {$n$};
	\node (opO1) at (-1,1.5) [above] {$U$};
	\node (opO2) at (-1,-0.7) [left] {$\cO_2$};
	\node (opO3) at (1,-0.7) [right] {$\cO_3$};
	\node (opW) at (1,1.5) [above] {$V$};	
	\node at (0.1,0) [right] {$W$};	
	\draw [finite with arrow] (vertD)-- (opO1);
	\draw [spinning] (opO2)-- (vertU);
	\draw [finite with arrow] (vertU)-- (vertD);
	\draw [spinning] (vertU)-- (opO3);
	\draw [finite with arrow] (vertD)-- (opW);
	\end{tikzpicture}}.
\ee
Here, the sum is over $W\in U\otimes V$. Since restricting $\cO_1$ to a finite-dimensional representation changes the counting of structures on both sides, we should check that the numbers still agree. Let us assume that $\De_3=\De_2-l$. According to theorem~\ref{thm:operatorclassification}, the number of structures on the left is
\be
	\sum_{i+k=l} \dim (\rho_2^* \otimes U_i\otimes V_k \otimes \rho_3)^{\SO(d)},
\ee
while the number of structures on the right is
\be
\sum_{W\in U\otimes V} \dim (\rho_2^* \otimes W_l \otimes \rho_3)^{\SO(d)}\times \dim (U\otimes V\otimes W^*)^{\SO(d+1,1)}.
\ee
These numbers are the same due to
\be
	\bigoplus_{W\in U\otimes V} \dim (U\otimes V\otimes W^*)^{\SO(d+1,1)}\times W_l\simeq \bigoplus_{i+k=l} U_i\otimes V_k.
\ee

The crossing transformation~\eqref{eq:algebracrossing} defines the algebra of weight-shifting operators for general $\rho_2$ and generic $\De_2$. In section~\ref{sec:algebraofweightshifting} we described the same algebra in the situation when both $\rho_2$ and $\De_2$ are generic. As in section~\ref{sec:algebraofweightshifting},~\eqref{eq:algebracrossing} essentially expresses the associativity of the tensor product.

As a simple application, suppose that $U=V^*$ and let us contract $U$ and $V$ indices in~\eqref{eq:algebracrossing} to form the bubble diagram,
\be
\diagramEnvelope{\begin{tikzpicture}[anchor=base,baseline]
	\node (vertL) at (0,-0.08) [threept] {$a$};
	\node (vertR) at (2,-0.133) [threept] {$b$};
	\node (opO1) at (1,1) [above] {$U$};
	\node (opO2) at (-1,0) [left] {$\cO_2$};
	\node (opO3) at (3,0) [right] {$\cO_3$};
	\node at (1,0.1) [above] {$\cO_3'$};	
	\draw [finite with arrow] (vertL) to[out=90,in=90] (vertR);
	\draw [spinning] (opO2)-- (vertL);
	\draw [spinning] (vertL)-- (vertR);
	\draw [spinning] (vertR)-- (opO3);
	\end{tikzpicture}}
\quad=\quad
\sum_{W,m,n}
\left\{
\begin{matrix}
	U & \cO_2 & W \\
	\cO_3 & U^* & \cO_3'
\end{matrix}
\right\}^{ab}_{mn}
\diagramEnvelope{\begin{tikzpicture}[anchor=base,baseline]
	\node (vertU) at (0,-0.78) [threept,inner sep=1.5pt] {$m$};
	\node (vertD) at (0,0.7) [threept] {$n$};
	\node (opO1) at (0,1.5) [above] {$U$};
	\node (opO2) at (-1,-0.7) [left] {$\cO_2$};
	\node (opO3) at (1,-0.7) [right] {$\cO_3$};
	\node at (0.1,0) [right] {$W$};	
	\draw [spinning] (opO2)-- (vertU);
	\draw [finite with arrow] (vertU)-- (vertD);
	\draw [spinning] (vertU)-- (opO3);
	\draw [finite with arrow] (vertD) to[out=150,in=180] (0,1.5) to[out=0,in=30] (vertD);
	\end{tikzpicture}}.
\ee
The tadpole on the right can be non-zero only if $W=\bullet$ is the representation of the identity operator $\mathbf{1}$. But then $m$ exists only if $\De_2=\De_3$ and $\rho_2=\rho_3$. In this case there exists a unique structure for both $n$ and $m$. We can erase the line for the trivial $W$, and the $U$-loop gives a factor of $\dim U$. We thus find
\be
\diagramEnvelope{\begin{tikzpicture}[anchor=base,baseline]
	\node (vertL) at (0,-0.08) [threept] {$a$};
	\node (vertR) at (2,-0.133) [threept] {$b$};
	\node (opO1) at (1,1) [above] {$U$};
	\node (opO2) at (-1,0) [left] {$\cO_2$};
	\node (opO3) at (3,0) [right] {$\cO_3$};
	\node at (1,0.1) [above] {$\cO_3'$};	
	\draw [finite with arrow] (vertL) to[out=90,in=90] (vertR);
	\draw [spinning] (opO2)-- (vertL);
	\draw [spinning] (vertL)-- (vertR);
	\draw [spinning] (vertR)-- (opO3);
	\end{tikzpicture}}=
	(\dim U)\left\{
	\begin{matrix}
		U & \cO_2 & \mathbf{1} \\
		\cO_3 & U^* & \cO_3'
	\end{matrix}
	\right\}^{ab}_{\uniq\uniq}
\diagramEnvelope{\begin{tikzpicture}[anchor=base,baseline]
		\node (opO2) at (0,0) [left] {$\cO_2$};
		\node (opO3) at (1,0) [right] {$\cO_3$};
		\draw [spinning] (opO2) -- (opO3);
\end{tikzpicture}}
\ee
We thus conclude
\be
	\begin{pmatrix}
		\cO_3\\
		\cO_3'\ U
	\end{pmatrix}^{ba} = (\dim U)\left\{
\begin{matrix}
U & \cO_3 & \mathbf{1} \\
\cO_3 & U^* & \cO_3'
\end{matrix}
\right\}^{ab}_{\uniq\uniq}.
\ee

The algebra~\eqref{eq:algebracrossing} also shows tells us how to compose the two-point operators of~\cite{Costa:2011dw}. Indeed, suppose we have a composition of two two-point operators, ignoring the operator labels,
\be
\diagramEnvelope{\begin{tikzpicture}[anchor=base,baseline]
	\node (opO1) at (0,-1) [left] {$\cO_1$};
	\node (opO1p) at (3,-1) [right] {$\cO_1'$};
	\node (opO2) at (0,1) [left] {$\cO_2$};
	\node (opO2p) at (3,1) [right] {$\cO_2'$};
	\node (vertUup) at (0.8,1) [threept] {};
	\node (vertUdn) at (0.8,-1) [threept] {};
	\node (vertVup) at (2.2,1) [threept] {};
	\node (vertVdn) at (2.2,-1) [threept] {};
	\node (U) at (0.8,0) [left] {$U$};
	\node (V) at (2.2,0) [right] {$V$};
	\draw [spinning] (opO2) -- (vertUup) -- (vertVup) -- (opO2p);
	\draw [spinning] (opO1) -- (vertUdn) -- (vertVdn) -- (opO1p);
	\draw [finite with arrow] (vertUdn) -- (vertUup);
	\draw [finite with arrow] (vertVdn) -- (vertVup);
	\end{tikzpicture}}.
\ee
We can apply~\eqref{eq:algebracrossing} at the top and at the bottom to find, schematically,
\be
\diagramEnvelope{\begin{tikzpicture}[anchor=base,baseline]
	\node (opO1) at (0,-1) [left] {$\cO_1$};
	\node (opO1p) at (3,-1) [right] {$\cO_1'$};
	\node (opO2) at (0,1) [left] {$\cO_2$};
	\node (opO2p) at (3,1) [right] {$\cO_2'$};
	\node (vertUup) at (0.8,1) [threept] {};
	\node (vertUdn) at (0.8,-1) [threept] {};
	\node (vertVup) at (2.2,1) [threept] {};
	\node (vertVdn) at (2.2,-1) [threept] {};
	\node (U) at (0.8,0) [left] {$U$};
	\node (V) at (2.2,0) [right] {$V$};
	\draw [spinning] (opO2) -- (vertUup) -- (vertVup) -- (opO2p);
	\draw [spinning] (opO1) -- (vertUdn) -- (vertVdn) -- (opO1p);
	\draw [finite with arrow] (vertUdn) -- (vertUup);
	\draw [finite with arrow] (vertVdn) -- (vertVup);
	\end{tikzpicture}}=\sum_{W,W',\ldots}\{\cdots\}^2
\diagramEnvelope{\begin{tikzpicture}[anchor=base,baseline]
	\node (opO1) at (0,-1) [left] {$\cO_1$};
	\node (opO1p) at (3,-1) [right] {$\cO_1'$};
	\node (opO2) at (0,1) [left] {$\cO_2$};
	\node (opO2p) at (3,1) [right] {$\cO_2'$};
	\node (vertWup) at (1.5,1) [threept] {};
	\node (vertWdn) at (1.5,-1) [threept] {};
	\node (vertWmidup) at (1.5,0.3) [threept] {};
	\node (vertWmiddn) at (1.5,-0.3) [threept] {};
	\node (U) at (1.2,0) [left] {$U$};
	\node (V) at (1.8,0) [right] {$V$};
	\node (W) at (1.6,-0.7) [right] {$W$};
	\node (Wp) at (1.6,0.7) [right] {$W'$};
	\draw [spinning] (opO2) -- (vertWup) -- (opO2p);
	\draw [spinning] (opO1) -- (vertWdn) -- (opO1p);
	\draw [finite with arrow] (vertWdn) -- (vertWmiddn);
	\draw [finite with arrow] (vertWmidup) -- (vertWup);	
	\draw [finite with arrow] (vertWmiddn) to[out=180,in=180] (vertWmidup);
	\draw [finite with arrow] (vertWmiddn) to[out=0,in=0] (vertWmidup);
	\end{tikzpicture}},
\ee
where $\{\cdots\}^2$ is a product of two $6j$ symbols. By Schur's lemma, the bubble diagram in the middle can be non-zero only if $W=W'$, in which case it is a scalar. This scalar can be determined from finite-dimensional $6j$ symbols. We thus arrive at
\be
\diagramEnvelope{\begin{tikzpicture}[anchor=base,baseline]
	\node (opO1) at (0,-1) [left] {$\cO_1$};
	\node (opO1p) at (3,-1) [right] {$\cO_1'$};
	\node (opO2) at (0,1) [left] {$\cO_2$};
	\node (opO2p) at (3,1) [right] {$\cO_2'$};
	\node (vertUup) at (0.8,1) [threept] {};
	\node (vertUdn) at (0.8,-1) [threept] {};
	\node (vertVup) at (2.2,1) [threept] {};
	\node (vertVdn) at (2.2,-1) [threept] {};
	\node (U) at (0.8,0) [left] {$U$};
	\node (V) at (2.2,0) [right] {$V$};
	\draw [spinning] (opO2) -- (vertUup) -- (vertVup) -- (opO2p);
	\draw [spinning] (opO1) -- (vertUdn) -- (vertVdn) -- (opO1p);
	\draw [finite with arrow] (vertUdn) -- (vertUup);
	\draw [finite with arrow] (vertVdn) -- (vertVup);
	\end{tikzpicture}}=\sum_{W,\ldots}\{\cdots\}^3
\diagramEnvelope{\begin{tikzpicture}[anchor=base,baseline]
	\node (opO1) at (0,-1) [left] {$\cO_1$};
	\node (opO1p) at (3,-1) [right] {$\cO_1'$};
	\node (opO2) at (0,1) [left] {$\cO_2$};
	\node (opO2p) at (3,1) [right] {$\cO_2'$};
	\node (vertWup) at (1.5,1) [threept] {};
	\node (vertWdn) at (1.5,-1) [threept] {};
	\node (W) at (1.6,0) [right] {$W$};
	\draw [spinning] (opO2) -- (vertWup) -- (opO2p);
	\draw [spinning] (opO1) -- (vertWdn) -- (opO1p);
	\draw [finite with arrow] (vertWdn) -- (vertWup);
	\end{tikzpicture}},
\ee
where $\{\cdots\}^3$ are some coefficients involving three $6j$ symbols, and the sum is over $W\in U\otimes V$.

\section{Seed blocks in 3d}
\label{app:3Dblocks}

\paragraph{Basis of four-point tensor structures.} For the four-point tensor structures we use the conformal frame structures
\be
	[q_1q_2q_3q_4]
\ee
that we introduced in section~\ref{sec:3D6j_threepoint}. It is analogous to the basis used in~\cite{Kravchuk:2016qvl}, but we make a different choice of the conformal frame,
\be
	x_1&=(0,0,0),\nn\\
	x_2&=(\frac{\bar z-z}{2},0,\frac{\bar z+z}{2}),\nn\\
	x_3&=(0,0,1),\nn\\
	x_4&=(0,0,+\infty).
\ee
The configuration used in section~\ref{sec:3D6j_threepoint} corresponds then to $z=\bar z$.

In terms of these structures, for parity-even four-point functions ($G^{++}_\text{seed}$ and $G^{--}_\text{seed}$) we use the basis
\be
G=
g_1(z,\bar z)\frac{[-\thalf,0,0,-\thalf]+[\thalf,0,0,\thalf]}{2}
+g_2(z,\bar z)\frac{[\thalf,0,0,-\thalf]+[-\thalf,0,0,\thalf]}{2},
\ee
and for parity-odd four-point functions ($G^{+-}_\text{seed}$ and $G^{-+}_\text{seed}$) we use the basis
\be
G=
g_1(z,\bar z)\frac{[-\thalf,0,0,-\thalf]-[\thalf,0,0,\thalf]}{2}
+g_2(z,\bar z)\frac{[\thalf,0,0,-\thalf]-[-\thalf,0,0,\thalf]}{2}.
\ee
We will now provide explicit expressions for $g^{\pm\pm}_i(z,\bar z)$.

\paragraph{Explicit expressions for $g^{\pm\pm}_i(z,\bar z)$.}
First we strip off some normalization factors,
\be\label{eq:3DexplicitPreBegin}
g^{++}_k(z,\bar z)&=\frac{i(-1)^{\ell-\half}}{\ell(\De-\ell-1)(\De-1)}(z \bar z)^{-\frac{\De_1+\De_2+\half}{2}}\mathfrak{D}^{++}_iG^{\alpha-\tfrac{1}{4},\beta-\tfrac{1}{4}}_{\De+\half,\ell-\half}(z,\bar z),\nn\\
g^{--}_k(z,\bar z)&=\frac{i(-1)^{\ell-\half}}{(\ell+\half)(\De+\ell)(\De-1)}(z \bar z)^{-\frac{\De_1+\De_2+\half}{2}}\mathfrak{D}^{--}_iG^{\alpha-\tfrac{1}{4},\beta-\tfrac{1}{4}}_{\De+\half,\ell+\half}(z,\bar z),\nn\\
g^{+-}_k(z,\bar z)&=\frac{i(-1)^{\ell-\half}}{\ell(\De-\ell-1)(\De-1)}(z \bar z)^{-\frac{\De_1+\De_2+\half}{2}}\mathfrak{D}^{+-}_iG^{\alpha-\tfrac{1}{4},\beta+\tfrac{1}{4}}_{\De+\half,\ell-\half}(z,\bar z),\nn\\
g^{-+}_k(z,\bar z)&=\frac{i(-1)^{\ell-\half}}{(\ell+\half)(\De+\ell)(\De-1)}(z \bar z)^{-\frac{\De_1+\De_2+\half}{2}}\mathfrak{D}^{-+}_iG^{\alpha-\tfrac{1}{4},\beta+\tfrac{1}{4}}_{\De+\half,\ell+\half}(z,\bar z).
\ee
Here, $\alpha=-(\De_1-\De_2)/2$ and $\beta=(\De_3-\De_4)/2$, where $\De_i$ are the dimensions of the external operators in~\eqref{eq:3Dfermion_scalar_4pt}. To write down the expressions for $\mathfrak{D}^{\pm\pm}_i$, we introduce the following operators,
\be
D_z &= z^2(1-z)\ptl_z^2-(\alpha'+\beta'+1)z^2\ptl_z-\alpha'\beta'z,\nn\\
d_z &= z\ptl_z,\nn\\
\nabla_z&=\frac{1}{z-\bar z}d_z(z-\bar z)=z\ptl_z+\frac{z}{z-\bar z},\nn\\
\widetilde d_z&=(1-z)d_z-\alpha' z,\nn\\
\widetilde \nabla_z&=(1-z)\nabla_z-(\alpha'-1)z,
\ee
as well as their conjugates which are obtained by exchanging $z$ and $\bar z$. The variables $\alpha'$ and $\beta'$ in the above formulas are equal to the parameters of the scalar conformal blocks $G^{\alpha',\beta'}_{\De',\ell'}(z,\bar z)$ on which the differential operators act in~\eqref{eq:3DexplicitPreBegin}.

The differential operators $\mathfrak{D}^{\pm\pm}_i$ are given by\footnote{Note again that in order to simplify these expressions we made use of the quadratic Casimir equation satisfied by the scalar conformal blocks.}
\be
\mathfrak{D}^{\pm\pm}_1=&
	d_{\bar z}D_z-d_z D_{\bar z}
	-(d_z-d_{\bar z})\frac{z\bar z}{2(z-\bar z)}\big(
		(1-z)\ptl_z-(1-\bar z)\ptl_{\bar z}
		\big)\nn\\
	&+a^{\pm\pm}(d_z-d_{\bar z})
	+b^{\pm\pm}(D_z-D_{\bar z}),\nn\\
\mathfrak{D}^{\pm\pm}_2=&
	\nabla_z D_{\bar z}+\nabla_{\bar z} D_{z}+
	(\nabla_z+\nabla_{\bar z})\frac{z\bar z}{2(z-\bar z)}\big(
		(1-z)\ptl_z-(1-\bar z)\ptl_{\bar z}
		\big)\nn\\
	&-a^{\pm\pm}(\nabla_z+\nabla_{\bar z})
	+c^{\pm\pm},
\ee
where coefficients $a^{\pm\pm},\,b^{\pm\pm}$ and $c^{\pm\pm}$ are given below and, additionally, \textit{in $\mathfrak{D}^{+-}$ and $\mathfrak{D}^{-+}$ the operators $d$ and $\nabla$ need to be replaced by $\widetilde{d}$ and $\widetilde{\nabla}$ respectively.} We have
\be
	a^{++}&=a^{+-}=\frac{(\De-\ell)(\De-\ell-3)}{4},\nn\\
	b^{++}&=b^{+-}=\frac{\De-\ell-3}{2},\nn\\
	c^{++}&=c^{+-}=\frac{(2\ell+1)(\De-\ell-3)(\De-\frac{3}{2})}{4},
\ee
and the coefficients for parity-odd left structure are obtained by replacing $\ell\to-\ell-1$,
\be	
	a^{--}&=a^{-+}=\frac{(\De+\ell+1)(\De+\ell-2)}{4},\nn\\
	b^{--}&=b^{-+}=\frac{\De+\ell-2}{2},\nn\\
	c^{--}&=c^{-+}=-\frac{(2\ell+1)(\De+\ell-2)(\De-\frac{3}{2})}{4}.
\ee

\paragraph{Normalization conventions}
Our normalization conventions are fixed by our choice of two-point functions~\eqref{eq:3Dtwopoint}, the scalar-fermion three-point functions~\eqref{eq:3Dspinorthreepoint} and the scalar three-point functions~\eqref{eq:3Dscalartheepoint}. These conventions agree with~\cite{Iliesiu:2015akf}. In particular, if the scalar blocks are normalized as
\be
	G^{\alpha,\beta}_{\De,\ell}(z,\bar z)\sim\frac{(-1)^\ell(1)_\ell}{(1/2)_\ell}(z\bar z)^{\De/2} P_\ell\p{\frac{z+\bar z}{2\sqrt{z\bar z}}},\quad z,\bar z\ll 1,
\ee
where $P_\ell$ are Legendre polynomials, then the resulting seed blocks $G^{\pm\pm}_\text{seed}$ are normalized as in~\cite{Iliesiu:2015akf} with their $c_\cO=1$. To obtain the blocks at other values of $c_\cO$, one should divide our formulas by $c_\cO$.

\paragraph{The coefficients $v_i$}

For $G^{++}$ we have
\be
v_1&=-\frac{i (-1)^{\ell -\frac{1}{2}} \left(-2 \Delta -2 \Delta _3-2 \Delta _4+2 \ell +9\right) \left(2 \Delta +2 \Delta _3-2 \Delta _4+2 \ell -1\right) \left(-2 \Delta -2 \Delta _3+2 \Delta _4+2 \ell +1\right)}{64 (\Delta -1) \left(2 \Delta _3-3\right) \left(2 \Delta _4-3\right) \ell  (-\Delta +\ell +1)},\nn\\
v_2&=-\frac{i (-1)^{\ell -\frac{1}{2}} \left(-2 \Delta -2 \Delta _3-2 \Delta _4+2 \ell +9\right) \left(-2 \Delta +2 \Delta _3-2 \Delta _4+2 \ell +3\right) \left(-2 \Delta -2 \Delta _3+2 \Delta _4+2 \ell +3\right)}{128 (\Delta -1) \left(\Delta _3-1\right) \left(2 \Delta _3-3\right) \ell  (-\Delta +\ell +1)},\nn\\
v_3&=-\frac{i (-1)^{\ell -\frac{1}{2}} \left(-2 \Delta +2 \Delta _3-2 \Delta _4+2 \ell +3\right)}{64 (\Delta -1) \left(\Delta _3-2\right){}^2 \left(2 \Delta _3-3\right) \left(2 \Delta _4-3\right) \ell  (-\Delta +\ell +1)},\nn\\
v_4&=\frac{i (-1)^{\ell -\frac{1}{2}} \left(-2 \Delta -2 \Delta _3-2 \Delta _4+2 \ell +9\right) \left(-2 \Delta +2 \Delta _3-2 \Delta _4+2 \ell +3\right) \left(-2 \Delta +2 \Delta _3+2 \Delta _4+2 \ell -3\right)}{128 (\Delta -1) \left(\Delta _3-2\right) \left(2 \Delta _3-3\right) \ell  (-\Delta +\ell +1)}.
\ee
For $G^{--}$ we have
\be
v_1&=-\frac{i (-1)^{\ell -\frac{1}{2}} \left(2 \Delta +2 \Delta _3-2 \Delta _4+2 \ell +1\right) \left(-2 \Delta -2 \Delta _3+2 \Delta _4+2 \ell +3\right) \left(2 \Delta +2 \Delta _3+2 \Delta _4+2 \ell -7\right)}{32 (\Delta -1) \left(2 \Delta _3-3\right) \left(2 \Delta _4-3\right) (2 \ell +1) (\Delta +\ell )},\nn\\
v_2&=-\frac{i (-1)^{\ell -\frac{1}{2}} \left(2 \Delta +2 \Delta _3-2 \Delta _4+2 \ell -1\right) \left(2 \Delta -2 \Delta _3+2 \Delta _4+2 \ell -1\right) \left(2 \Delta +2 \Delta _3+2 \Delta _4+2 \ell -7\right)}{64 (\Delta -1) \left(\Delta _3-1\right) \left(2 \Delta _3-3\right) (2 \ell +1) (\Delta +\ell )},\nn\\
v_3&=-\frac{i (-1)^{\ell -\frac{1}{2}} \left(2 \Delta -2 \Delta _3+2 \Delta _4+2 \ell -1\right)}{32 (\Delta -1) \left(\Delta _3-2\right){}^2 \left(2 \Delta _3-3\right) \left(2 \Delta _4-3\right) (2 \ell +1) (\Delta +\ell )},\nn\\
v_4&=\frac{i (-1)^{\ell -\frac{1}{2}} \left(2 \Delta -2 \Delta _3-2 \Delta _4+2 \ell +5\right) \left(2 \Delta -2 \Delta _3+2 \Delta _4+2 \ell -1\right) \left(2 \Delta +2 \Delta _3+2 \Delta _4+2 \ell -7\right)}{64 (\Delta -1) \left(\Delta _3-2\right) \left(2 \Delta _3-3\right) (2 \ell +1) (\Delta +\ell )}.
\ee
For $G^{+-}$ we have
\be
v_1&=\frac{i (-1)^{\ell -\frac{1}{2}} \left(-2 \Delta -2 \Delta _3-2 \Delta _4+2 \ell +9\right) \left(-2 \Delta -2 \Delta _3+2 \Delta _4+2 \ell +3\right)}{64 (\Delta -1) \left(\Delta _3-1\right) \left(2 \Delta _3-3\right) \left(2 \Delta _4-3\right) \ell  (-\Delta +\ell +1)},\nn\\
v_2&=-\frac{i (-1)^{\ell -\frac{1}{2}} \left(-2 \Delta -2 \Delta _3-2 \Delta _4+2 \ell +7\right) \left(2 \Delta -2 \Delta _3-2 \Delta _4+2 \ell +3\right) }{128 (\Delta -1) \left(2 \Delta _3-3\right) \ell  (-\Delta +\ell +1)}\times\nn\\
&\quad\qquad\times \left(-2 \Delta -2 \Delta _3+2 \Delta _4+2 \ell +3\right) \left(2 \Delta +2 \Delta _3+2 \Delta _4+2 \ell -7\right),\nn\\
v_3&=-\frac{i (-1)^{\ell -\frac{1}{2}} \left(-2 \Delta +2 \Delta _3-2 \Delta _4+2 \ell +3\right) \left(-2 \Delta -2 \Delta _3+2 \Delta _4+2 \ell +3\right)}{64 (\Delta -1) \left(\Delta _3-2\right) \left(2 \Delta _3-3\right) \left(2 \Delta _4-3\right) \ell  (-\Delta +\ell +1)},\nn\\
v_4&=-\frac{i (-1)^{\ell -\frac{1}{2}} \left(-2 \Delta +2 \Delta _3-2 \Delta _4+2 \ell +1\right) \left(2 \Delta -2 \Delta _3+2 \Delta _4+2 \ell -1\right)}{128 (\Delta -1) \left(\Delta _3-2\right){}^2 \left(2 \Delta _3-3\right) \ell  (-\Delta +\ell +1)}.
\ee
For $G^{-+}$ we have
\be
v_1&=-\frac{i (-1)^{\ell -\frac{1}{2}} \left(2 \Delta +2 \Delta _3-2 \Delta _4+2 \ell -1\right) \left(2 \Delta +2 \Delta _3+2 \Delta _4+2 \ell -7\right)}{32 (\Delta -1) \left(\Delta _3-1\right) \left(2 \Delta _3-3\right) \left(2 \Delta _4-3\right) (2 \ell +1) (\Delta +\ell )},\nn\\
v_2&=\frac{i (-1)^{\ell -\frac{1}{2}} \left(-2 \Delta -2 \Delta _3-2 \Delta _4+2 \ell +9\right) \left(2 \Delta +2 \Delta _3-2 \Delta _4+2 \ell -1\right) }{64 (\Delta -1) \left(2 \Delta _3-3\right) (2 \ell +1) (\Delta +\ell )}\times\nn\\
&\quad\qquad\times\left(-2 \Delta +2 \Delta _3+2 \Delta _4+2 \ell -1\right) \left(2 \Delta +2 \Delta _3+2 \Delta _4+2 \ell -5\right),\nn\\
v_3&=\frac{i (-1)^{\ell -\frac{1}{2}} \left(2 \Delta +2 \Delta _3-2 \Delta _4+2 \ell -1\right) \left(2 \Delta -2 \Delta _3+2 \Delta _4+2 \ell -1\right)}{32 (\Delta -1) \left(\Delta _3-2\right) \left(2 \Delta _3-3\right) \left(2 \Delta _4-3\right) (2 \ell +1) (\Delta +\ell )},\nn\\
v_4&=\frac{i (-1)^{\ell -\frac{1}{2}} \left(-2 \Delta +2 \Delta _3-2 \Delta _4+2 \ell +3\right) \left(2 \Delta -2 \Delta _3+2 \Delta _4+2 \ell +1\right)}{64 (\Delta -1) \left(\Delta _3-2\right){}^2 \left(2 \Delta _3-3\right) (2 \ell +1) (\Delta +\ell )}.
\ee

\section{Dual seed blocks in 4d}
\label{app:dual_seeds}
In this appendix we provide the final expression for the dual seed conformal blocks recursion relation omitting all the derivations. All the quantities below carry a bar to distinguish them from their analogous in the seed case.

By performing the calculation completely analogous to the one in section~\ref{sec:seed4d}, we find that the dual seed conformal blocks obey the following recursion relation
\begin{align}\label{eq:dual_seed_CPWs_recursion_relation}
\overline W^{(p)}_{\Delta,\,\ell;\;\Delta_1,\,\Delta_2,\,\Delta_3,\,\Delta_4} &=\nn\\\nn
\overline{\mathcal{A}}^{-1}\Bigg( &\bar v_1\,
(\Dmmz{1}\cdot\Dmzp{4})\,(\Dmpz{1}\cdot \Dppz{2})\;\overline W^{(p-1)}_{\Delta-\tfrac{1}{2},\,\ell;\;\Delta_1+1,\,\Delta_2-\tfrac{1}{2},\,\Delta_3,\,\Delta_4+\tfrac{1}{2}}\\\nn
+&\bar v_2\,
(\Dmmz{1}\cdot\Dmzp{4})\,(\Dppz{1}\cdot \Dmpz{2})\;\overline W^{(p-1)}_{\Delta-\tfrac{1}{2},\,\ell;\;\Delta_1,\,\Delta_2+\tfrac{1}{2},\,\Delta_3,\,\Delta_4+\tfrac{1}{2}}\\\nn
+&\bar v_3\,
(\Dpzm{1}\cdot\Dmzp{4})\,(\Dpzp{1}\cdot \Dppz{2})\;\overline W^{(p-1)}_{\Delta-\tfrac{1}{2},\,\ell;\;\Delta_1-1,\,\Delta_2-\tfrac{1}{2},\,\Delta_3,\,\Delta_4+\tfrac{1}{2}}\\
+&\bar v_4\,
(\Dpzm{1}\cdot\Dmzp{4})\,(\Dmzp{1}\cdot \Dmpz{2})\;\overline W^{(p-1)}_{\Delta-\tfrac{1}{2},\,\ell;\;\Delta_1,\,\Delta_2+\tfrac{1}{2},\,\Delta_3,\,\Delta_4+\tfrac{1}{2}}\Bigg),
\end{align}
where the coefficients are\footnote{Here $\overline\cA$ is not the $6j$ symbol analogous to $\cA$, but simply an overall coefficient.}
\begin{equation}
\overline{\mathcal{A}}=-\frac{i\,(\ell+p) (\Delta +\Delta_3-\Delta_4+\ell+p-2)}{2 \Delta +2 \ell+p-2}
\end{equation}
and
\begin{align}
\bar v_1 &=   \frac{(\Delta -\Delta_1-\Delta_2+\ell+p+2) (-\Delta -\Delta_1+\Delta_2+\ell+p+2)}{2 (\Delta_1-2) (2 \Delta +p-4) (2 \Delta_2+p-4)}  ,\nn\\
\bar v_2 &=  -\frac{(\Delta -\Delta_1-\Delta_2+\ell+p+2) (\Delta -\Delta_1+\Delta_2+\ell+2 p-2)}{4 (\Delta_1-2) (\Delta_1-1) (2 \Delta +p-4)}  ,\nn\\
\bar v_3 &=  -\frac{1}{2 (\Delta_1-3) (\Delta_1-2)^2 (2 \Delta +p-4) (2 \Delta_2+p-4)} ,\nn\\
\bar v_4 &=  -\frac{(\Delta -\Delta_1-\Delta_2+\ell+p+2) (\Delta +\Delta_1+\Delta_2+\ell+2 p-6)}{4 (\Delta_1-3) (\Delta_1-2) (2 \Delta +p-4)} .
\end{align}
Analogously to the primal seed case, we replace one of the conformal blocks on the right hand side of~\eqref{eq:dual_seed_CPWs_recursion_relation} by using the dimension-shifting operator
\begin{equation}\label{eq:relation_shifted_dimensions_dual_CPWs}
\overline W^{(p-1)}_{\Delta-\tfrac{1}{2},\,\ell;\;\Delta_1+1,\,\Delta_2-\tfrac{1}{2},\,\Delta_3,\,\Delta_4+\tfrac{1}{2}}=
\overline{\mathcal{E}}^{-1}\,(\Dpmz{1}\cdot \Dmmz{2})(\Dppz{1}\cdot \Dmpz{2})\,
\overline W^{(p-1)}_{\Delta-\tfrac{1}{2},\,\ell;\;\Delta_1,\,\Delta_2+\tfrac{1}{2},\,\Delta_3,\,\Delta_4+\tfrac{1}{2}},
\end{equation}
where
\begin{equation}
\overline{\mathcal{E}}\equiv
(p+1)(\Delta_1-2) (\Delta_1-1)(\Delta +\Delta_1-\Delta_2+l+p-2) (-\Delta -\Delta_1+\Delta_2+l+p+2).
\end{equation}

\paragraph{Decomposition into components}
Plugging the relation~\eqref{eq:relation_shifted_dimensions_dual_CPWs} in~\eqref{eq:dual_seed_CPWs_recursion_relation}, stripping off the kinematic factor
and decomposing this relation into four-point tensor structures according to~\eqref{eq:dual_seed_blocks_4D} one obtains a recursion relation for the components of seed blocks of the form analogous to~\eqref{eq:seed4dfinal}
\begin{equation}
\overline H_e^{(p)}(z,\bar z)=-\frac{\overline{\mathcal{A}}^{\prime-1}}{z-\bar z}\,
\left(
\mathit{D}_0\; \overline H_{e}^{(p-1)}(z,\bar z)
-2\mathit{D}_1\; \overline H_{e-1}^{(p-1)}(z,\bar z)
+4\cParam{p-1}{e-2}z\bar z\mathit{D}_2\; \overline H_{e-2}^{(p-1)}(z,\bar z)
\right),
\end{equation}
where the blocks in the l.h.s depend on $[\Delta,\,\ell;\;\Delta_1,\,\Delta_2,\,\Delta_3,\,\Delta_4]$ and the blocks in the r.h.s.\ depend on
$[\Delta-\tfrac{1}{2},\,\ell;\;\Delta_1,\,\Delta_2+\tfrac{1}{2},\,\Delta_3,\,\Delta_4+\tfrac{1}{2}]$.

The overall coefficient is
\begin{equation}
\overline{\mathcal{A}}^{\prime}\equiv
-(\Delta + \frac{p}{2} - 2) (\Delta +\Delta_1 - \Delta_2 + l + p - 2)\,
\overline{\mathcal{A}}.
\end{equation}
The differential operators $\mathit{D}_i$ are given by the expression \eqref{eq:d_0}-\eqref{eq:d_2} with the parameter $k$ replaced by $\bar k$
\begin{equation}
\bar k\equiv\frac{\Delta+\ell}{2}+\frac{3p}{4}.
\end{equation}

\section{Operators $\cH_k$}
\label{app:Hkoperators}

First, let us define normalized versions of operators~\eqref{eq:ourHkoperatorsPre},
\be\label{eq:ourHkoperators}
\hat\cD_{13}&=\frac{\cD^{-0}_1\cdot \cD^{+0}_3}{(\De_3-1)(d-\De_3-2)},\nn\\
\hat\cD_{24}&=\frac{\cD^{+0}_2\cdot \cD^{-0}_4}{(\De_2-1)(d-\De_2-2)},\nn\\
\hat\cD_{23}&=\frac{\cD^{+0}_2\cdot \cD^{+0}_3}{(\De_3-1)(d-\De_3-2)(\De_2-1)(d-\De_2-2)}.
\ee
In terms of these, the operators $\cH_k$ have the following expressions,\footnote{Here and below in this section we sometimes abuse the notation by using the same symbols for the embedding-space differential operators and their action on $F_{\lambda_1\lambda_2}$.}
\be
\mathcal{H}_1&=\frac{\hat\cD_{13}-\hat\cD_{24}}{\De_{12}^+-\De_{34}^+}+\frac{1}{4}(\De_{12}^++\De_{34}^+-2\varepsilon) (x\bar x)^{-\half},\nn\\
\mathcal{H}_2&=\frac{\hat\cD_{13}+\hat\cD_{24}}{2}+\frac{\De_{12}^++\De_{34}^+}{2}\mathcal{H}_1-\frac{\De_{12}^+(\De_{12}^+-2\varepsilon)+\De_{34}^+(\De_{34}^+-2\varepsilon)}{8}(x\bar x)^{-\half},\nn\\
\mathcal{H}_3&=\frac{2\hat\cD_{23}}{\De_{12}^++\De_{34}^+-4\varepsilon-2}-(\De_{12}^++\De_{34}^+-2(\varepsilon-1))\cH_2+(2c_2+\De_{12}^+\De_{34}^+)\cH_1+\kappa_0 (x\bar x)^{-\half},
\ee
where we defined\footnote{Notice that in~\cite{DO3} $\varepsilon$ is defined as here, whereas in the earlier work~\cite{DO2} the definition $\varepsilon=d-2$ was used.}
\be
\Delta_{ij}^+=\Delta_i+\Delta_j,\quad \varepsilon=\frac{d-2}{2},
\ee
and the Casimir eigenvalues and the coefficient $\kappa_0$ are given by
\be
c_2&=\lambda_1(\lambda_1-1)+\lambda_2(\lambda_2-2\epsilon-1),\nn\\
c_4&=(\lambda_1-\lambda_2)(\lambda_1-\lambda_2+2\varepsilon)(\lambda_1+\lambda_2-1)(\lambda_1+\lambda_2-1-2\varepsilon),\nn\\
\kappa_0&=\frac{(\De_{12}^+-2\varepsilon)(\De_{34}^+-2\varepsilon)(\De_{12}^+\De_{34}^++4c_2)-4(c_4-c_2(c_2+2\varepsilon))}{4(\De_{12}^++\De_{34}^+-4\varepsilon-2)}.
\ee
Note that the identity for the operator $\cH_3$ is valid up to quadratic and quartic Casimir equations (i.e.\ only when acting on scalar conformal blocks $F_{\lambda_1\lambda_2}(a,b)$).

\bibliographystyle{JHEP}
\bibliography{refs}

\end{document}